\documentclass[letter,11pt]{article}
\usepackage{xcolor}

\usepackage[utf8]{inputenc} 
\usepackage{float}

\usepackage{geometry}
\usepackage{graphicx} 

\usepackage{booktabs} 
\usepackage{array} 
\usepackage{paralist} 
\usepackage{verbatim} 
\usepackage{amsmath}
\usepackage{upgreek}
\usepackage{amssymb}
\usepackage{mathtools}
\usepackage{amsthm}
\usepackage{empheq}
\usepackage[most]{tcolorbox}
\usepackage{pdfpages}
\usepackage{multirow}
\usepackage{leftidx}

\usepackage{amsmath}
\DeclareMathOperator{\plim}{plim}

\newtcbox{\mymath}[1][]{%
    nobeforeafter, math upper, tcbox raise base,
    enhanced, colframe=blue!30!black,
    colback=blue!30, boxrule=1pt,
    #1}

\newcommand{\overbar}[1]{\mkern 1.5mu\overline{\mkern-1.5mu#1\mkern-1.5mu}\mkern 1.5mu}

\usepackage{subcaption}
\usepackage{algorithm}
\usepackage[noend]{algpseudocode}

\makeatletter
\def\BState{\State\hskip-\ALG@thistlm}
\makeatother
 
\usepackage{fancyhdr} 
\usepackage{sectsty}
\allsectionsfont{\sffamily\mdseries\upshape} 

\usepackage[nottoc,notlof,notlot]{tocbibind} 
\usepackage[titles,subfigure]{tocloft} 

\usepackage{graphicx}
\usepackage{epstopdf}
\DeclareGraphicsExtensions{.eps}

\usepackage[math]{cellspace}
\usepackage{leftidx}

\usepackage{geometry}
\usepackage[colorinlistoftodos]{todonotes}
\usepackage[colorlinks=true, allcolors=blue]{hyperref}

\title{One-step Estimation of Networked Population Size:\\
Respondent-Driven Capture-Recapture with Anonymity}

\usepackage{authblk}
\author[1,*]{Bilal Khan}
\author[1]{Hsuan-Wei Lee}
\author[2]{Ian Fellows}
\author[1]{Kirk Dombrowski}

\affil[1]{Department of Sociology, University of Nebraska-Lincoln}
\affil[2]{Fellow Statistics}
\affil[*]{Corresponding author: bkhan2@unl.edu}

\newtheorem{proposition}{Proposition}
\newtheorem{corollary}{Corollary}
\newtheorem{lemma}{Lemma}
\newtheorem{definition}{Definition}

\newtheorem{assumption}{Assumption}
\newtheorem{notation}{Notation}

\newcounter{example}[section]

\begin{document}
\date{}
\maketitle
\begin{abstract}
Population size estimates for hidden and hard-to-reach populations are particularly important when members are known to suffer from disproportion health issues or to pose health risks to the larger ambient population in which they are embedded. Efforts to derive size estimates are often frustrated by a range of factors that preclude conventional survey strategies,  including social stigma associated with group membership or members' involvement in illegal activities. 

This paper extends prior research on the problem of network population size estimation, building on established survey/sampling methodologies commonly used with hard-to-reach groups. Three novel one-step, network-based population size estimators are presented, to be used in the context of uniform random sampling, respondent-driven sampling, and when networks exhibit significant clustering effects. Provably sufficient conditions for the consistency of these estimators (in large configuration networks) are given.  Simulation experiments across a wide range of synthetic network topologies validate the performance of the estimators, which are seen to perform well on a real-world location-based social networking data set with significant clustering.  Finally, the proposed schemes are extended to allow them to be used in settings where participant anonymity is required.  Systematic experiments show favorable tradeoffs between anonymity guarantees and estimator performance.

Taken together, we demonstrate that reasonable population estimates can be derived from anonymous respondent driven samples of 250-750 individuals, within ambient populations of 5,000-40,000. The method thus represents a novel and cost-effective means for health planners and those agencies concerned with health and disease surveillance to estimate the size of hidden populations. Limitations and future work are discussed in the concluding section.

\end{abstract}

Population size estimation for hidden and hard-to-reach populations is of considerable interest to health officials seeking to prevent health problems that may be concentrated in such populations \cite{magnani_review_2005}, or when ``reservoirs'' of infection among a hidden population pose health threats to the ambient population in which the hidden population is embedded \cite{dombrowski_topological_2013, reluga_reservoir_2007}. In the former, treatable maladies can remain out of reach, multiplying eventual treatment costs when cases are discovered only in their most severe form. Such is the situation, for example, with mental illness among homeless and street dwelling populations \cite{bonin_typology_2009, burt_critical_1995,ivanich2017relationship}. In other situations, the “hidden” nature of the population may frustrate intervention efforts that are effective in the ambient population, preventing control of infections despite well-known contagion dynamics \cite{potterat_aids_1993}. A simple example, long-known to public health officials, is the high prevalence of sexually transmitted infections among commercial sex workers \cite{abdul2014estimating, law_spatial_2004, zohrabyan_determinants_2013}. In such situations, health officials seek to know the overall prevalence levels of maladies within a hidden population \textit{and} the size of those populations in order to understand the scope of treatment needs and overall social risk. 

Efforts to ascertain prevalence and size estimates are frustrated by a range of factors that produce the “hiddenness” of the population initially. Such factors include heavy social stigma that precludes a willingness on the part of members of the hidden population to reveal their membership. Such is the situation with people who inject drugs (PWID), who may be unwilling to self-identify as such under ordinary survey conditions \cite{darke_self-report_1998, harwood_sampling_2012}. Hiddenness due to stigma can be further compounded when such activities are illegal, when they carry heavy personal costs (such as when self-identified heterosexual men also have sex with other men), or when disease status is unknown (such as undiagnosed HIV infection rates among PWID). In these situations, conventional sampling is unreliable, and ordinary multiplier methods based on conventional sampling are rendered ineffective.

A number of strategies have been devised to address either the prevalence or population size (or both) aspects of this problem. These include capture-recapture \cite{larson_indirect_1994, vuylsteke_capturerecapture_2010}, chain referral \cite{biernacki_snowball_1981, platt_methods_2006}, venue-based \cite{haley_venue-based_2014, muhib_venue-based_2001}, cluster sampling \cite{burnham_mortality_2006}, and combinations of these. Among the most popular is respondent-driven sampling (RDS) \cite{heckathorn_extensions_2007, RDS2002, salganik_sampling_2004}, which has been adopted for use in many of the situations described above, and which is employed widely in HIV surveillance efforts both within the United States and beyond \cite{johnston2016systematic}. RDS employs an incentivized chain referral process to recruit a sample of the hidden population. Under restricted but recognized conditions, RDS can be shown to result in a steady-state, ``equilibrium'' sample, and numerous means have been derived for producing reasonable prevalence estimates from such a sample while accounting for biases introduce in the referral process \cite{gile_respondent-driven_2010, gile_diagnostics_2015, mouw_network_2012,shi_model-based_2016, verdery_network_2015, wejnert_empirical_2009}. The ease of implementing RDS, the fact that it can operate under conditions of anonymity (via number coupons that track referrals), and its rigorous treatment under a range of statistical modeling strategies have made it a popular choice for researchers working with hidden populations \cite{heckathorn_network_2017}. However, equally rigorous means for estimating the overall size of the hidden population from RDS derived data have been less successful, often resulting in widely varying estimates \cite{sulaberidze_population_2016,crawford2017hidden}.   The ability of the RDS method to produce meaningful prevalence data remains, and presents considerable potential for use in size estimation.  

Other efforts restricted to size estimation of hidden populations have been developed, including various versions of capture-recapture procedures (sometimes call mark-recapture or multiplier procedures) \cite{domingo-salvany_analytical_1998, kruse_participatory_2003} and network scale-up methods (NSUM) \cite{bernard_counting_2010}. Multiplier efforts normally make use of a sample of the hidden population and some external, normally institutional knowledge-base (e.g. arrest records or hospital admissions) for estimation purposes \cite{hay_estimating_1996, vuylsteke_capturerecapture_2010}. In these cases, however, two assumptions must be met: (i) that the sample is representative of the hidden population more generally, and (ii) that everyone in the hidden population is equally likely to be “captured” in the official statistics used in the estimation \cite{jones_recapture_2014}. While representativeness can sometimes be assumed (as in the case of RDS), it is often difficult to establish the uniformity of the capture statistics, and often there are good reasons to believe that random capture is simply not the case. Frankly put, police arrests and hospital admission can seldom be assumed to draw randomly from the hidden population. Further, capture-recapture/multiplier methods often require that the sample be identifiable in the institutional record, requiring that expectations on the part of sample respondents for anonymity be sacrificed. When working with hidden and highly stigmatized populations, such a sacrifice can be highly detrimental to both recruitment and informant reliability \cite{wolitski_effects_2009}.

Network scale-up methods are also used to establish the size of hidden populations, though work in this area remains at an early stage. Here members of the entire population (ambient plus hidden) are asked to report on the number of known associates who fit the hidden population criteria \cite{ezoe_population_2012, guo_estimating_2013}. This technique has the advantage of being employable under ordinary random sampling conditions that can make use of known sampling frames (i.e. mail surveys and/or random digit dialing) \cite{habecker_improving_2015}. However, this method assumes that ordinary people know whom among their associates fit the criteria for inclusion in the hidden category \cite{killworth_investigating_2006, salganik_assessing_2011}. This assumption raises suspicion in many of those situations in which we ordinarily wish to use it, as when we seek to estimate populations of PWID or sex workers. Under these conditions, individuals from the hidden population may not want their friends and associates to know about their status, and may make efforts to hide this information. These efforts introduce “transmission” errors into NSUM estimates that are difficult to uncover or estimate. 

In previous work, we proposed a novel capture-recapture methodology for estimating the size of a hidden population from an RDS sample \cite{TELEFUNKEN2012}. Were such a result possible, it could easily be integrated into the conventional RDS framework, taking advantage of the wide body of work in that area and the ability of RDS to produce reasonable prevalence estimates. Our method was first proposed in several forms undertaken as quasi-experiments within actual data collection efforts with commercially sexually exploited children \cite{curtis_commercial_2008} and, later, users of methamphetamine \cite{wendel_dynamics_2011}. Both studies took place in New York City, and both made use of RDS samples. Subsequent implementations of the technique have lent further evidence of the effectiveness (and ease of implementation) of what we there referred to as the ``telefunken'' method. This method asks RDS sample respondents to report on others in the population known to them via an encoding of their associates telephone number and demographic features, avoiding the reliance on official statics or the need to draw two independent samples from the hidden population. The technique was referred to as \textit{telefunken} because it entailed an encoding of the phone numbers of known associates in the hidden population. The code was created by taking a specified number of phone number digits, in order from last to first, and encoding each digit as 0/1 for even/odd, and again 0/1 for 0-4/5-9. This produced a binary code of length 2 x the number of phone digits specified in the protocol. This many-to-one encoding allowed for ongoing anonymity for both respondents and their reported associates, while enabling the matching of contacts across numerous respondent interviews. It also introduced the need to estimate the number of expected false matches created by the many-to-one encoding. In essence, this ``one-step'' approach eased the assumptions normally associated with other capture-recapture methods, and can be accomplished via a single sample from the hidden population. If shown to be effective, such an approach lends simplicity and greater cost-effectiveness to the size estimation procedure, potentially allowing for widespread application.  

Independently, and in roughly the same time period, Fellows put forward a general framework of Privatized Network Sampling (PNS) design \cite{fellowsThesis}.  PNS addressed two of the major concerns with regard to RDS data, namely the assumption that coupons are passed at random among alters, and that subjects can accurately report the number of alters that they have. As PNS is closely related to RDS, the standard RDS estimators may be used on data collected with the PNS design.

Given the interest in telefunken and PNS-like techniques \cite{merli_sampling_2016, mouw_network_2012, sulaberidze_population_2016}, this paper proposes a more rigorous formalization of a one-step, network-based population estimation procedure that can be employed under conditions of anonymity. In what follows we describe the technique and simulate its performance under a range of implementation conditions across a range of hidden population sizes. The simulations show considerable promise for the technique under the kinds of research scenarios normally associated with research among “hidden populations”. Limitations and further efforts toward validation/extension are discussed at the end of the paper.

\section{Background}

Current network size estimation methods are based on quantifying the ``repetition'' or overlap observed across multiple samples \cite{heckathorn2017network}, where the category of objects  sampled may be nodes, edges, distances, paths, motifs, or substructures \cite{leskovec2006sampling, wang2014efficiently})---depending on the specific approach in question. 

\begin{itemize}
\item Node sampling methods often begin by taking independent uniform random samples of the population. In interpreting the overlap between samples \cite{bawa2003estimating, massoulie2006peer}, these methods are based on the same principle as the well-studied ``coupons collector's problem'' of probability theory, for which maximum likelihood estimators and conservative confidence intervals are long established  \cite{finkelstein1998confidence}. The classic method  considers two uniform independent random samples \cite{sekar1949method}; in ecology, this method is often referred to as the ``mark and recapture'' protocol. To wit, within a population $V$, the protocol first selects a uniform random ``capture'' sample $S\subseteq V$, and then a second {\em independent} uniform random ``recapture'' sample $R\subseteq V$. From independence assumptions one infers that
\begin{eqnarray}
\label{classical-prop}
\frac{|V|}{|S|} \approx \frac{|R|}{|S \cap R|}.
\end{eqnarray}
and hence
\begin{eqnarray}
\label{lincoln}
|V| \approx \frac{|S| \cdot |R|}{|S \cap R|}.
\end{eqnarray}
The right-hand-side expression in (\ref{lincoln}) is known as the Lincoln-Peterson estimator \cite{LINCOLN1930,PETERSON1896}.  Many extensions and performance improvements to this classical technique are known, such as those making use of weighted sampling techniques \cite{dasgupta2012social}, or sampling that is biased by the network's degree distribution \cite{katzir2011estimating}. 
\item Edge sampling approaches to population size estimation have also been developed \cite{krishnamurthy2005reducing, kurant2012graph, dasgupta2014estimating}. These methods not only consider a sampled set of nodes, but also elicit a sample of their network neighbors. While edge sampling encounters problems associated with a bias toward high degree nodes, these methods offers potential gains in efficiency in dense graphs and where independent random sampling of nodes is restricted. 
\item Lastly, sampling via random walks represents a practical approach that is commonly used in estimating the size of social networks. Random walk methods start from an arbitrary node, then move to a neighboring node uniformly at random, and iterate. A typical random walk visits every node with a frequency proportional to its degree, but this bias can be quantified and corrected by Markov Chain analysis. One can estimate the population size of a graph based on the frequency with which sampled nodes appear (and reappear) during the random walk process. This method is now widely used to measurement the size of graphs such as online social networks and is used in conjunction with a variety of data from web crawlers \cite{mislove2007measurement,ahn2007analysis,gjoka2010walking,kurant2011walking,hardiman2013estimating}.
\end{itemize}
The approach developed here is inspired by and builds on several of the above strategies, including random walks and edge elicitation. An outline of the paper is as follows: In Section \ref{sec:n1}, we present a population estimator for uniform random samples.  This estimator is extended for respondent-driven samples in Section \ref{sec:n2}.  The two estimators are evaluated over a broad range of graph families (see Subsection \ref{sec:graph-families}) using a general experimental framework (see Subsection \ref{sec:exp-framework}).  The experimental results are presented in Sections \ref{sec:eval-n1} and \ref{sec:eval-n2}. In Section 
   \ref{sec:n3}, we adapt the estimators for use in networks with clustering, showing in Section \ref{sec:eval-n3} that the revised schemes continue to perform well on synthetic networks.  In Section \ref{sec:anonymity}, we extend the network size estimation schemes to allow for protection of subject anonymity.  These anonymity-preserving schemes are evaluated through simulation experiments in Sections 
\ref{sec:eval-n2-psi} and \ref{sec:eval-n3-psi}.  The RDS-based estimators are used on a real-world network in Section \ref{sec:Brightkite}.  Finally, discussion and limitations are presented in Section \ref{sec:discussion}.


\section{New Population Size Estimators}

We seek to generalize the Lincoln-Peterson framework of overlapping capture and recapture sets to the context of networked populations, expressed formally in the language of graphs. The following definition provides graph-theoretic notations which will be necessary in order to precisely define the proposed sampling and estimation processes.

\begin{definition}
\label{def:basics}
Let $G=(V,E)$ be a graph.  For each $v \in V$, we denote the degree of $v$ in $G$ as $d(v)$.  Given $A \subseteq V$, we denote the (arithmetic) mean degree of vertices in $A$ as:
   \begin{align}
d(A) &\coloneqq \frac{1}{|A|} \sum_{v \in A} d(v)
\intertext{and the (harmonic) mean degree of vertices in $A$ as}
\widetilde{d}(A) &\coloneqq \frac{|A|}{\sum_{v \in A} \frac{1}{d(v)}}.
\intertext{noting that the latter is more robust against the presence of high-degree outliers.  If $H=(S,F)$ is a subgraph on $S\subseteq V$ with edge set $F \subseteq E \cap (S \times S)$, the {\em free neighborhood} of $u$ (in $G$ modulo $H$) is defined as}
N(u,F) &\coloneqq \{ v \;|\; (u,v) \in E\;\backslash\; F\} \subseteq V. \label{def:N}
\intertext{Note that when $G$ is allowed to have parallel edges (as is the case when it is obtained through configuration graph sampling), then $N(u,F)$ may be a multiset.  The {\em free ends} of $S$ (in $G$ modulo $H$) are taken to be the disjoint union (multiset)}
R(S,F) &\coloneqq \coprod_{u\in S} N(u,F)  \subseteq V \label{def:R}
\intertext{and the {\em matches} of (in $G$ modulo $H$) are taken to be the disjoint union (multiset)} 
M(S,F) &\coloneqq \coprod_{u\in S} \left( N(u,F) \cap S \right)  \subseteq V. \label{def:M}
\intertext{We denote their respective multiset cardinalities as}
 \langle R(S,F) \rangle &\coloneqq \sum_{u\in S} \left| N(u,F) \right| \nonumber\\
 \langle M(S,F) \rangle &\coloneqq \sum_{u\in S} \left| N(u,F) \cap S \right|.\nonumber
  \end{align}
  \end{definition}
  
\begin{notation}
In the arguments that follow, graph-theoretic quantities (such as those formalized in Definition \ref{def:basics}) will sometimes be considered simultaneously in the context of more than one graph---e.g. $G_1=(V_1,E_1)$,  and $G_2=(V_2,E_2)$.  To avoid ambiguity, in such settings, we will
make the context unambiguous by appending the graph as a parameter---e.g. the average degree of vertices in $G_1$ is denoted $d(V_1; G_1)$, while the average degree of vertices in $G_2$ is expressed as $d(V_2; G_2)$.
\end{notation}

\begin{notation}
\label{def:multiset-cardinality}
Whenever we are considering a multiset $X$, we will denote to its multiset cardinality as $\langle X \rangle$, while its set cardinality will be written as $|X^*|$. For example, if $X=\{1,1,2,8,8,8\}$ then $\langle X \rangle = 6$, while $|X^*| = 3$.
\end{notation}
\begin{definition}
\label{def:multiset-operations}
Given multisets of vertices $A,B \subseteq V$ we denote their characteristic functions as $\chi_A,\chi_B: V \rightarrow {\mathbb N}$ and define the multisets 
$A \backslash B$, $A \cap B$, $A \cup B$ by the respective characteristic functions
$$
\chi_{A \backslash B}, \chi_{A \cap B}, \chi_{A \cup B}:V\rightarrow {\mathbb N}
$$
where for each $v\in V$
\begin{eqnarray*}
\chi_{A \backslash B}(v) &\coloneqq& \max\{0, \chi_{A}(v) - \chi_{B}(v)\}\\
\chi_{A \cap B}(v) &\coloneqq& \min\{\chi_{A}(v), \chi_{B}(v)\}\\
\chi_{A \cup B}(v) &\coloneqq& \chi_{A}(v) + \chi_{B}(v).
\end{eqnarray*}
We say that $A \subseteq B$ as multisets, if $\forall v\in V$, we have $\chi_A(v) \leqslant \chi_B(v)$.
\end{definition}

\subsection{Estimating Population Size from Uniform Random Samples in Graphs}
\label{sec:n1}

With the formalisms of Definition \ref{def:basics} in place, we can formally express the estimator $n_1$, which given a uniform random subset of vertices $T \subseteq V$, yields an estimate $|V|$.

\begin{definition}
\label{def:n1}
Given a graph $G=(V,E)$ and $T \subseteq V$, define
\begin{align}
\label{eq:n1}
n_1(T) &\coloneqq \frac{|T|\cdot \langle R(T,\emptyset) \rangle}{\langle M(T,\emptyset) \rangle}
\end{align}
\end{definition}

Lemma \ref{lemma:n1-samplesize} shows that as the sample size grows, $n_1$ converges to $|V|$.

\begin{lemma}
\label{lemma:n1-samplesize}
Let $G=(V,E)$ be a graph and $T_1\subseteq T_2\subseteq T_3\subseteq \ldots \subseteq V$ an ascending chain converging to $\bigcup_{i=1}^{\infty} T_i = V$.  Then 
$$\lim_{i \rightarrow \infty} \frac{n_1(T_i)}{|V|} = 1$$
\end{lemma}
\begin{proof}

Let $R_i \coloneqq R(T_i,\emptyset)$, $M_i \coloneqq M(T_i,\emptyset)$, and $\Delta_i \coloneqq R_i \backslash M_i$.  Note that $R_1\subseteq R_2\subseteq R_3\subseteq \ldots$ and $M_1\subseteq M_2\subseteq M_3\subseteq \ldots$ are ascending chains of multisets, and $M_i \subseteq R_i$ ($i=1,2,\ldots$).  Suppose $u \in \Delta_i$ and $\chi_{R_i}(u) = a$; clearly $0 < a \leqslant d(u)$.  Then since the ascending chain $(T_i)_{i=1,2,\ldots}$ converges to $V$, there exists a least $j_0>i$ for which $\chi_{M_j}(u) =  d(u)$ and therefore $\chi_{\Delta_j}(u) = 0$ for all $j\geqslant j_0$.  It follows that
  \begin{align*}
\bigcap_{i=1}^{\infty} R_i \backslash M_i &= \emptyset
\intertext{(where multiset intersection and difference are as described in Definition \ref{def:multiset-operations}), and thus} 
\lim_{i\rightarrow \infty} \frac{\langle R_i \rangle}{\langle M_i \rangle} &= 1
  \end{align*}
which implies $\lim_{i\rightarrow \infty} n_1(T_i)/|T_i| = 1$, completing the proof.
\end{proof}


The next proposition gives sufficient conditions under which uniform random samples $T\subseteq V$ produce consistent estimates $n_1(T) \sim |V|$ when $|V|$ is large.  Concrete realizations of these conditions are presented in Corollary \ref{corr:pop1}.

\begin{proposition}
\label{prop:n1-limit-n}
For $n=1,2,\ldots$ let $G_n=(V_n,E_n)$ be a graph on $|V_n|=f(n)$ vertices, where $f(n)$ grows unbounded.  Let $c_n \in (0,1]$ and take $T_n\subseteq V_n$ to be a subset of size $|T_n| = \lfloor c_n\cdot f(n) \rfloor$ selected using uniform random sampling in $V_n$.  If $c_n\cdot f(n)$ diverges as $n$ goes to infinity, while
    \begin{align}
c_n^2 \cdot d(V_n) \xrightarrow{\phantom{xxx}\phantom{xxx}} \Theta_1
  \end{align}
for some finite constant $\Theta_1>0$, then $\frac{n_1(T_n)}{|V_n|}$ necessarily converges to $1$.
\end{proposition}
\begin{proof}
Define random variables
  \begin{align}
\overbar{R}_n &\coloneqq \frac{1}{f(n)}\langle R(T_n,\emptyset) \rangle = \frac{1}{f(n)}\sum_{u\in T_n} d(u) \label{ref:n1-rn-def}\\
\overbar{M}_n &\coloneqq \frac{1}{f(n)}\langle M(T_n,\emptyset) \rangle.\label{ref:n1-mn-def}
  \end{align}
  For uniform random $u \in V_n$, we have that $E[d(u)]=d(V_n)$.  Since $|T_n|=\lfloor c_n\cdot f(n) \rfloor$ diverges, by the law of large numbers and linearity of expectation, as $n$ tends to infinity
    \begin{align}
  \langle R(T_n,\emptyset) \rangle = \sum_{u\in T_n} d(u) &\xrightarrow{\phantom{xxx}p\phantom{xxx}} \sum_{u\in T_n} d(V_n) = |T_n| \cdot d(V_n)\\
  \intertext{and thus,}
c_n \cdot \overbar{R}_n = \frac{1}{f(n)}\langle R(T_n,\emptyset) \rangle &\xrightarrow{\phantom{xxx}p\phantom{xxx}} c_n \cdot \frac{1}{f(n)}\cdot |T_n| \cdot d(V_n) = c_n^2 \cdot d(V_n) \xrightarrow{\phantom{xxx}p\phantom{xxx}} \Theta_1. \label{ref:n1-rn}
\intertext{Now for each $u\in T_n$ we have $E[\langle N(u,F_n) \cap T_n \rangle] = d(u) \cdot |T_n|/f(n)$.  Again, by the law of large numbers and linearity of expectation, as $n$ tends to infinity} 
\overbar{M}_n &\xrightarrow{\phantom{xxx}p\phantom{xxx}} \overbar{R}_n \cdot \frac{|T_n|}{f(n)} =  \overbar{R}_n \cdot c_n \xrightarrow{\phantom{xxx}p\phantom{xxx}} \Theta_1. \label{ref:n1-mn}
  \end{align}
Considering (\ref{ref:n1-rn}) and (\ref{ref:n1-mn}) as preconditions of Slutsky’s theorem \cite{slutsky1925uber}, we conclude:
\begin{eqnarray*}
\frac{n_1(T_n)}{f(n)}  
= \frac{1}{f(n)} \cdot \frac{c_n \cdot f(n) \cdot  \overbar{R}_n}{\overbar{M}_n} \xrightarrow{\phantom{xxx}d\phantom{xxx}}   \frac{\plim_{n\rightarrow \infty} c_n \cdot  \overbar{R}_n}{\plim_{n\rightarrow \infty} \overbar{M}_n} = \frac{\Theta_1}{\Theta_1} = 1.  
\end{eqnarray*}
\end{proof}

The correspondence between equation (\ref{eq:n1}) in Definition \ref{def:n1} and our previous ``telefunken'' estimator is clear \cite{dombrowski_estimating_2012}. In addition, equation (\ref{eq:n1}) demonstrates a parallel structure with the Lincoln-Peterson estimator shown in expression (\ref{lincoln}): $T$ represents the first assay (set); $R(T,\emptyset)$ stands for the second assay (a multiset); the multiset $M(T,\emptyset)$ is the subpopulation of the first assay that is recaptured by the second assay.  Of course, in the present setting, the second assay $R(T,\emptyset)$ is far from independent of the first assay $T$, since the two are sets are intrinsically linked through the network geometry of $G$.  Nevertheless, the fact that $T$ is a random subset of $V$ is enough to neutralize the impact of this non-independence and enable consistent estimation of population size. 

\begin{corollary}
\label{corr:pop1}
Several special cases of Proposition \ref{prop:n1-limit-n} are of interest.  In each of these cases, it is straightforward to verify that as $n$ goes to infinity, $c_n\cdot f(n)$ diverges, while $c_n^2 \cdot d(V_n)$ tends to some finite strictly positive constant.
\begin{itemize}
\item When $f(n)=O(n)$, $c_n = O(1)$ constant, and $d(V_n) = O(1)$ constant.  In this case, we have a family of graphs of increasing size and constant average degree, in which we are taking uniform random samples whose size is a constant proportion of the entire population.
\item When $f(n)=O(n)$, $c_n = O(g(n)/n)$, and $d(V_n) = O(n^{1-\epsilon}/g(n)^2)$, where $g(n)$ is a function which diverges, and $\epsilon > 0$ is a constant.  For example, if we take $g(n)=n^\epsilon$, then $c_n = O(1/n^{1-\epsilon})$, and $d(V_n) = O(n^{1-3\epsilon})$.  As $\epsilon$ tends to 0, we approach a family of graphs of increasing size and linear average degree, in which we are taking uniform random samples of an absolute constant size.  This special limit case is manifested by Erd\H{o}s-R\'enyi graphs \cite{citeulike:4012374}.
\end{itemize}
\end{corollary}

\subsection{Estimating Population Size from Respondent-Driven Samples in Configuration Graphs}
\label{sec:n2}
Though $n_1$ shows robust performance under uniform random sampling (see Section \ref{sec:eval-n1}), such an approach to data collection is seldom possible or reliable when working with hidden populations. As discussed above, sampling-hard to-reach populations represents considerable practical challenges \cite{heckathorn2017network}, and many current surveys of hidden populations have come to depend on a tracked ``peer referral'' process knowns as respondent driven sampling \cite{RDS2002}. Given the tendency of RDS to oversample high degree nodes, issues arise for edge-based size estimation regarding the sampling of incomplete edges, and steps must be taken to account for differences between the average degree of the sample and the average degree of the population from which the sample is drawn. 

For purposes of estimation, we consider a respondent-driven sample to be a random variable $\text{RDS-CAPTURE}(G, |D|, c, n_0)$ requiring four parameters: an underlying networked population $G=(V,E)$, 
a specified number of seeds $|D|$, the number of recruiting coupons $c$ to be given to each subject, and $r$ the target sample size. Informally stated, the procedure chooses $s$ random initial ``seed'' subjects in the network.  Each of these subjects is asked to participate in a ``referral'' process by being given $c$ recruiting coupons to be distribute among their peers.  When one or more of those peers come in for interview with the coupon received from their recruiter, they too are given $c$ coupons and participate in the referral process.  The scheme proceeds recursively in this manner until $r$ individuals have been recruited and interviewed. If and whenever the referral process stalls before $r$ subjects have been interviewed, a new seed is recruited.
Participation incentives are arranged to ensure that no subject will be the recipient of more than one coupon.
Note that this breadth-first search process always yields a collection of disjoint trees \cite{Bollobas98a}.
In a simulated setting this real-world process is implemented as formally described in the RDS-CAPTURE procedure presented as Algorithm \ref{rds-algo} (pp. \pageref{rds-algo}).

\begin{algorithm}[b]
\caption{Capturing a respondent-driven sample}\label{rds-algo}
\begin{algorithmic}[1]
\Procedure{RDS}{$G$, $|D|$, $c$, $r$}
\State $t \gets 0$
\State $S_0 \gets \{ v_1, \ldots, v_s\}$ a set of $|D|$ distinct ``seeds'' uniformly at random from $V[G]$.
\State $T_0 \gets \emptyset$.
\State $F_0 \gets S_0$.
\Repeat 
\State $t \gets t+1$
\State $x_{t} \gets $ a uniformly randomly chosen element from $F_{t-1}$
\State $N(x_{t}) \gets \{ v\in V[G] \setminus S_{t-1} \;|\; (x_{t},v) \in E[G] \}$ its undiscovered neighbors
\If {$|N(x_{t})| \leqslant c$}
\State $R(x_{t}) \gets N(x_{t})$
\Else
\State $R(x_{t}) \gets$ a uniformly  random chosen size-$c$ subset of $N(x_{t})$
\EndIf
\State $S_{t} \gets S_{t-1} \cup \{ x_{t} \} \cup R(x_{t})$
\State $T_{t} \gets T_{t-1} \cup \{ (x_t, v) \;|\; v \in R(x_{t})\}$
\State $F_{t} \gets F_{t-1} \setminus \{ x_{t} \} \cup R(x_{t})$
\If {$F_{t} = \emptyset$ and $|S_{t}| < n_0$}
\State $F_{t} \gets \{ v \}$ a single ``seed'' chosen at random from $V[G] \setminus S_{t}$
\EndIf
\Until {$|S_{t}| \geqslant r$}
\State \Return $(S_t, T_t)$
\EndProcedure
\end{algorithmic}
\end{algorithm}

\begin{assumption}
\label{assumption-harmonic}
Whenever we are considering $H=(S,F)$ to be a subgraph on $S\subseteq V$ obtained through an RDS process inside graph $G=(V,E)$, we will assume $\widetilde{d}(S) \sim d(V)$.  This assumption is justified in prior work \cite{salganik_sampling_2004,heckathorn_extensions_2007}, provably true for configuration graphs \cite{giles2009}, and reflects the basic fact that the harmonic mean is robust against the inclusion of high-degree outliers that we may face when $S$ is obtained via a non-uniform sampling process like RDS.
\end{assumption}

The next estimator $n_2$, provides an estimate $|V|$ from a respondent driven sample $S \subseteq V$.

\begin{definition}
\label{def:n2}
Given a graph $G=(V,E)$, a set $S \subseteq V$, and $H=(S,F)$ a subgraph with edge set $F \subseteq E \cap (S \times S)$, define
\begin{align}
\label{eq:n2}
n_2(S,F) &\coloneqq \frac{\frac{d(S)-1}{\widetilde{d}(S)} \cdot |S|\cdot \langle R(S,F) \rangle }{\langle M(S,F) \rangle}
\end{align}
\end{definition}

The next proposition gives sufficient conditions under which respondent-driven samples $S\subseteq V$ produce consistent estimates $n_2(T) \sim |V|$ when $|V|$ is large.

\begin{proposition}
\label{prop:n2}
For $n=1,2,\ldots$ let $G_n=(V_n,E_n)$ be a graph on $|V_n|=f(n)$ vertices obtained by configuration graph sampling via degree distribution ${\cal D}_n$, where $f(n)$ grows unboundedly.  Let $c_n \in (0,1]$, and take $S_n\subseteq V_n$ to be a subset of size $|S_n| = \lfloor c_n\cdot f(n) \rfloor$ selected using RDS sampling in $G_n$ from $|D_n|$ seeds chosen uniformly at random.  Define the random variable
  \begin{align*}
\Delta_n &\coloneqq \frac{d(S_n)-1}{\widetilde{d}(S_n)}.
  \end{align*}
Accepting Assumption \ref{assumption-harmonic}, if $c_n\cdot f(n)/D_n$ diverges as $n$ goes to infinity, while 
    \begin{align}
\Delta_n^2 \cdot c_n^2 \cdot d(V_n) = \frac{{\left( d(S_n)-1 \right)}^2 \cdot c_n^2}{\widetilde{d}(S_n)} \xrightarrow{\phantom{xxx}p\phantom{xxx}} \Theta_2
  \end{align}
for some finite constant $\Theta_2 > 0$, then $\frac{n_2(S_n)}{|V_n|}$ necessarily converges to $1$.
\end{proposition}
\begin{proof}
Let $(S_n,F_n)$ be subgraph determined by an RDS sampling process in $G_n$ and let $T_n\subseteq V_n$ be an equal-sized set of vertices chosen by uniform random sampling, $|T_n|=|S_n|$. 
For random $u \in S_n$ and $v\in T_n$, as $n$ tends to infinity
\begin{eqnarray}
 \frac{|N(u,\emptyset)|}{d(S_n)} - \frac{|N(v,\emptyset)|}{d(T_n)} = \frac{|N(u,\emptyset)|}{d(S_n)} - \frac{|N(v,\emptyset)|}{d(V_n)} = \frac{|N(u,\emptyset)|}{d(S_n)} - \frac{|N(v,\emptyset)|}{\widetilde{d}(S_n)} &\xrightarrow{\phantom{xxx}p\phantom{xxx}} 0.\label{eq:ratio-eq}
 \end{eqnarray}
 where the first equality stems from the law of large numbers, and the second from Assumption \ref{assumption-harmonic}.
 Now $S_n$ is an RDS sample and hence is the disjoint union of $D_n$ many trees.  It follows that
\begin{align*}
\frac{|F_n|}{|S_n|} &= 1-\frac{|D_n|}{\lfloor c_n\cdot f(n) \rfloor}
\end{align*}
Since $|S_n|=\lfloor c_n\cdot f(n) \rfloor$ diverges and $c_n\cdot f(n)/D_n$ diverges, we may conclude that 
\begin{eqnarray}
\lim_{n \rightarrow \infty} \frac{|F_n|}{|S_n|} = 1. \label{eq:limit-1}
\end{eqnarray}
We note that $|N(u,F_n)| \leqslant  |N(u,\emptyset)|$, and incorporating (\ref{eq:limit-1}) back into the final expression in (\ref{eq:ratio-eq}), we deduce
\begin{align}
\frac{|N(u,F_n)|}{d(S_n) - 1} - \frac{|N(v,\emptyset)|}{\widetilde{d}(S_n)} &\xrightarrow{\phantom{xxx}p\phantom{xxx}} 0.\label{eq:ratio-eq2}
\intertext{Definition \ref{def:basics} equation (\ref{def:R}) and linearity of expectation then imply that as $n$ tends to infinity}
\label{eq:r-fix}
\langle R(S_n,F_n) \rangle &\xrightarrow{\phantom{xxx}p\phantom{xxx}} \frac{d(S_n)-1}{\widetilde{d}(S_n)}  \cdot \langle R(T_n,\emptyset) \rangle.
\end{align}
The configuration graph sampling process dictates that as $n$ tends to infinity, for uniformly random $u\in S_n$ 
$$E[\langle N(u,F_n) \cap S_n \rangle] = [d(u) - 1] \cdot \frac{\langle R(S_n,F_n) \rangle}{2|E_n|} = [d(u) - 1] \cdot \frac{\langle R(S_n,F_n) \rangle}{d(V_n) \cdot f(n)}.$$
Definition \ref{def:basics} equation (\ref{def:M}), expression (\ref{eq:r-fix}), the law of large numbers, and linearity of expectation, together imply that as $n$ tends to infinity
\begin{align}
\langle M(S_n,F_n) \rangle &\xrightarrow{\phantom{xxx}p\phantom{xxx}} \frac{\langle R(S_n,F_n) \rangle^2}{d(V_n) \cdot f(n)} \xrightarrow{\phantom{xxx}p\phantom{xxx}} \frac{1}{d(V_n) \cdot f(n)} \cdot \left[ \frac{d(S_n)-1}{\widetilde{d}(S_n)} \right]^2 \cdot \langle R(T_n,\emptyset) \rangle^2. \label{eq:m-fix}
\end{align}
Define the following random variables, closely related to (\ref{ref:n1-rn-def}) and (\ref{ref:n1-mn-def}) of Proposition \ref{prop:n1-limit-n}:
  \begin{align}
R_n^* \coloneqq \langle R(S_n, F_n) \rangle \;/\; f(n) = \Delta_n \cdot \overbar{R}_n &\xrightarrow{\phantom{xxx}p\phantom{xxx}} \Delta_n \cdot c_n \cdot d(V_n) \label{def:rv1-prop2}\\
M_n^* \coloneqq \langle M(S_n, F_n) \rangle \;/\; f(n) = \Delta_n^2 \cdot \overbar{R}_n^2 / d(V_n) &\xrightarrow{\phantom{xxx}p\phantom{xxx}} \Delta_n^2 \cdot c_n^2 \cdot d(V_n) \label{def:rv2-prop2}
\intertext{From our assumptions on the convergence of $\Delta_n^2 \cdot c_n^2 \cdot d(V_n)$, we see that as $n$ tends to infinity}
\Delta_n \cdot c_n \cdot R_n^* = \Delta_n^2 \cdot c_n^2 \cdot d(V_n)&\xrightarrow{\phantom{xxx}p\phantom{xxx}} \Theta_2 \label{ref:n1-rn2}\\
M_n^* &\xrightarrow{\phantom{xxx}p\phantom{xxx}} \Theta_2 \label{ref:n1-mn2}
  \end{align}
Considering (\ref{ref:n1-rn2}) and (\ref{ref:n1-mn2}) as preconditions of Slutsky’s theorem \cite{slutsky1925uber}, we conclude:
\begin{eqnarray*}
\frac{n_2(S_n)}{f(n)} 
= \frac{1}{f(n)}\cdot\frac{\Delta_n \cdot c_n f(n) \cdot R_n^*}{M_n^*} \xrightarrow{\phantom{xxx}d\phantom{xxx}}   \frac{\plim_{n\rightarrow \infty} \Delta_n \cdot c_n \cdot R_n^*}{\plim_{n\rightarrow \infty} M_n^*} = \frac{\Theta_2}{\Theta_2} = 1.  
\end{eqnarray*}
\end{proof}




\section{Evaluating the $n_1$ and $n_2$ Estimators}
\label{sec:n1-n2-eval}

To evaluate the proposed estimators $n_1$ (\ref{eq:n1}) and $n_2$ (\ref{eq:n2}), we conducted simulation experiments on samples drawn from synthetic networks using uniform and respondent-driven sampling, respectively.  Underlying networks were selected by configuration sampling techniques \cite{BENDER1978296,BOLLOBAS1980311,PhysRevE.64.026118} relative to Lognormal, Poisson, and Exponential distributions.   We also consider Barab\'asi-Albert graphs \cite{RevModPhys.74.47}, and Erd\H{o}s-R\'enyi graphs \cite{citeulike:4012374}. 

\subsection{Families of Synthetic Networks}
\label{sec:graph-families}

While little attention has been paid to the performance of RDS estimation on idealized graph topologies, the tendency of RDS to over-recruit high degree nodes is well known. Other attempts to model peer-referral/``snowball'' recruitment point to the fact that degree distribution can influence the performance of mean estimators \cite{illenberger2012estimating}, and Bayesian approaches to RDS  make considerable use of degree distribution to for population estimation \cite{handcock2014estimating,crawford2017hidden}. To validate the $n_1$ and $n_2$ against a wide range of possible topologies, five ideal families of random graphs were used to perform initial tests. Later, we take up the issue of clustering (Section \ref{sec:n3}), anonymity (Section \ref{sec:anonymity}), and the performance of further augmented estimators on a large real-world network (Section \ref{sec:Brightkite}). 

In what follows, configuration graphs were sampled (relative to the degree distribution) by first attaching the prescribed number of free half-edges to each node. Pairs of free half-edges are then chosen uniformly at random and bound together to form an edge, repeatedly, until no free half-edges remain. Note that in this model, the final graphs can have multiple parallel edges and self loops. 

\begin{definition}
Given set $V$ with $|V|=n$, for each positive $\lambda \in {\mathbb R}$, let distributions ${\cal D}_{{\cal L}(\lambda)}$, ${\cal D}_{{\cal P}(\lambda)}$, ${\cal D}_{{\cal X}(\lambda)}$, and ${\cal D}_{{\cal R}(\lambda)}:V \rightarrow {\mathbb N}$ be defined such that for each $v \in V$:
\begin{itemize}
\item ${\cal D}_{{\cal L}(\lambda,n)}(v)=1+X$ where $X$ is a Lognormal with mean $\lambda-1$ and standard deviation 1.
\item ${\cal D}_{{\cal P}(\lambda,n)}(v)=1+X$ where $X$ is a Poisson random variable with rate parameter $\lambda-1$.
\item ${\cal D}_{{\cal X}(\lambda,n)}(v)=1+X$ where $X$ is an Exponential random variable with mean $\lambda-1$.
\end{itemize}
Corresponding to each of the three distributions above, let ${\cal L}(\lambda,n), {\cal P}(\lambda,n), {\cal X}(\lambda,n), {\cal R}(\lambda,n)$ be the sample spaces of configuration graphs $G~=~(V,E)$ where $|V|=n$. Note that a random graph $G~=~(V,E)$ drawn from these sample spaces will have expected mean vertex degree $E[d(V)]=\lambda$.  
\end{definition}

\begin{definition}
Let ${\cal B}(\lambda, n)$ be the sample space of $n$-vertex Barab\'asi-Albert graphs $G=(V,E)$.  Each such graph is the final output of a process which produces a sequence of graphs $G^{i}=(V^{i},E^{i})$ on $V^{i} \coloneqq \{v_1, \ldots v_{i}\}$ with  $\lambda \leqslant i \leqslant n$.  $G^{\lambda}=(V^{\lambda},V^{\lambda}\times V^{\lambda})$ is taken to be the complete graph.  At each stage $i> \lambda$ of the process, node $v_i$ $( \lambda < i \leqslant n)$ connects to a random number
$$
\Delta_i \coloneqq |E_i \backslash E_{i-1}| = \left\{
\begin{array}{cl}
\lfloor \lambda/2 \rfloor & \text{ with probability } 1+\lfloor \lambda \rfloor - \lambda\\
1+\lfloor \lambda/2 \rfloor & \text{ otherwise.}
\end{array}
\right.
$$
of prior nodes $\{p_{i,1}, \ldots p_{i,\Delta_i}\} \subseteq V^{i-1}$.  This set of neighbors is constructed by sequential sampling without replacement, i.e. as $l=1, \ldots, \Delta_i$, each of the candidates $w \in C_{i,l} \coloneqq V^{i-1} \backslash \{v_{i,1}, \ldots v_{i,l-1}\}$ is chosen with probability  
$$
Prob(p_{i,l}=w) = 
\frac{1+d(w; G^{i-1})}{\sum_{w' \in C_{i,l}} 1+d(w'; G^{i-1})}.
$$
where $d(w; G^{i-1})$ is defined to be the degree of vertex $w$ in graph $G^{i-1}=(V^{i-1},E^{i-1})$. 
Note that when $n \gg \lambda$, a random graph $G=(V,E)$ drawn from ${\cal B}(\lambda, n)$ has expected mean vertex degree $E[d(V)] \sim \lambda$.
\end{definition}

\begin{definition}
Let ${\cal E}(\lambda, n)$ be the sample space of $n$-vertex Erd\H{o}s-R\'enyi graphs $G=(V,E)$, where $E\subseteq V \times V$ is a random subset constructed uniformly at random by taking:
$$
Prob((u,v) \in E) = \left\{
\begin{array}{cl}
\lambda/(n-1) & u\neq v\\
0 & u=v\\
\end{array}
\right.
$$
for each $(u,v)\in V\times V$.  Note that a random graph $G=(V,E)$ drawn from ${\cal E}(\lambda, n)$ will have expected mean vertex degree $E[d(V)] \sim \lambda$.  
\end{definition}

\subsection{Experimental Framework}
\label{sec:exp-framework}

For each of the $5$ families ${\cal L}(\lambda,n), {\cal P}(\lambda,n), {\cal X}(\lambda,n), {\cal B}(\lambda,n)$, and ${\cal E}(\lambda,n)$ defined in Section \ref{sec:graph-families}, we varied $\lambda=3, 5, 10$; from each of these $15$ concrete sample spaces, we used configuration graph sampling to select $30$ random graphs of  sizes $n=5000$, $10K$, $20K$ and $40K$.  In each of these $5 \times 3 \times 4 \times 30 = $1,800 graphs, we generated 30 uniform and 30 RDS samples of size $r=250,500$ and $750$.  In this manner, a total of $1,800\times 30 \times 3 x 2 = $324,000 simulations were conducted.  The Section \ref{sec:eval-n1}, these simulation experiments will involve uniform random samples; in Section \ref{sec:eval-n1} they will be based on respondent-driven samples.

\subsection{Evaluating $n_1$ on Synthetic Networks}
\label{sec:eval-n1}

The experiments here follow the framework described in Section \ref{sec:exp-framework} and use uniform random samples.  The $12$ graphs in Figure \ref{results:n1} present the performance of the $n_1$ estimator as the true population size $n$ is varied from $5\cdot 10^3$ to $40\cdot 10^3$ (vertical axis of the grid) and the size of the uniform sample is varied from $250$ to $750$ (horizontal axis of the grid).  In each of the $12$ graphs, the x-axis varies the average degree $\lambda$ from $3$ to $10$.  For each choice of $\lambda$, the medians and quartile ranges of $n_1$ are given for each of the $5$ graph families.  Each of these is determined by $900$ simulations ($30$ graphs times $30$ uniformly drawn samples in each graph).  

\begin{figure}[tbp]
\setlength{\belowcaptionskip}{12pt}
\centering
\begin{tabular}{ScScScSc}
\subcaptionbox{$n$ = $5\cdot 10^3$, $r=250$\label{1a}}{\includegraphics[width = 0.3\linewidth]{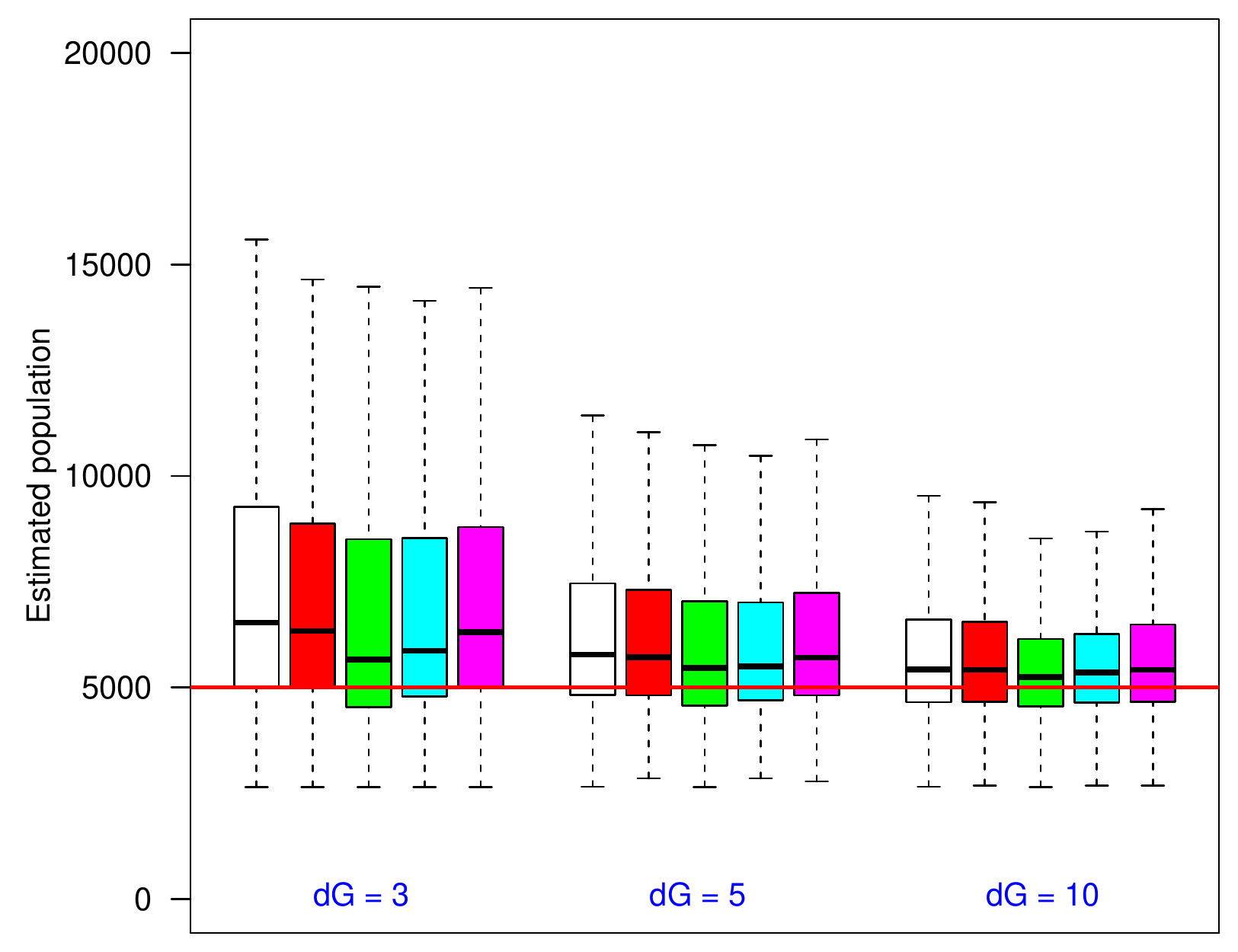}} &
\subcaptionbox{$n$ = $5\cdot 10^3$, $r=500$\label{2a}}{\includegraphics[width = 0.3\linewidth]{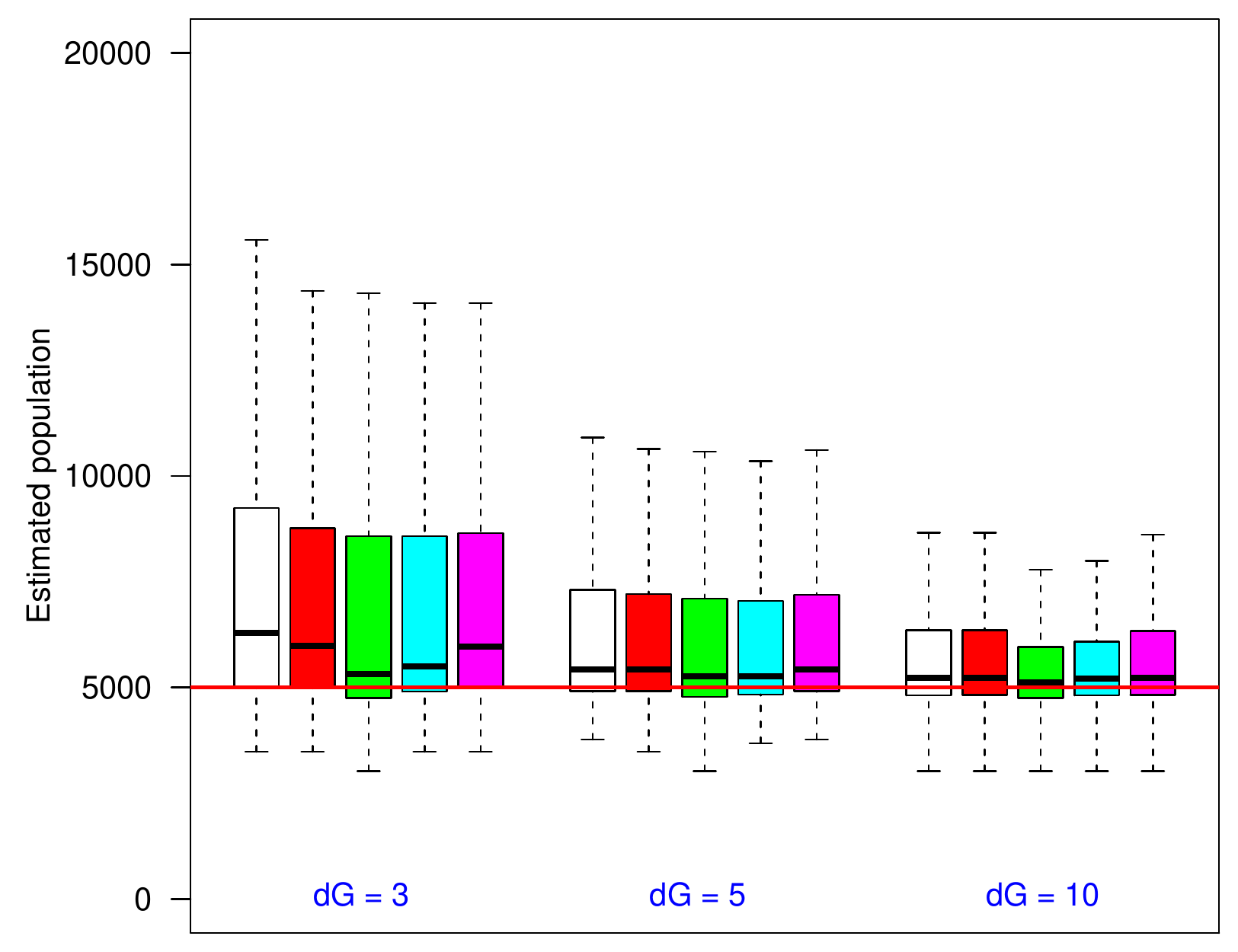}} &
\subcaptionbox{$n$ = $5\cdot 10^3$, $r=750$\label{3a}}{\includegraphics[width = 0.3\linewidth]{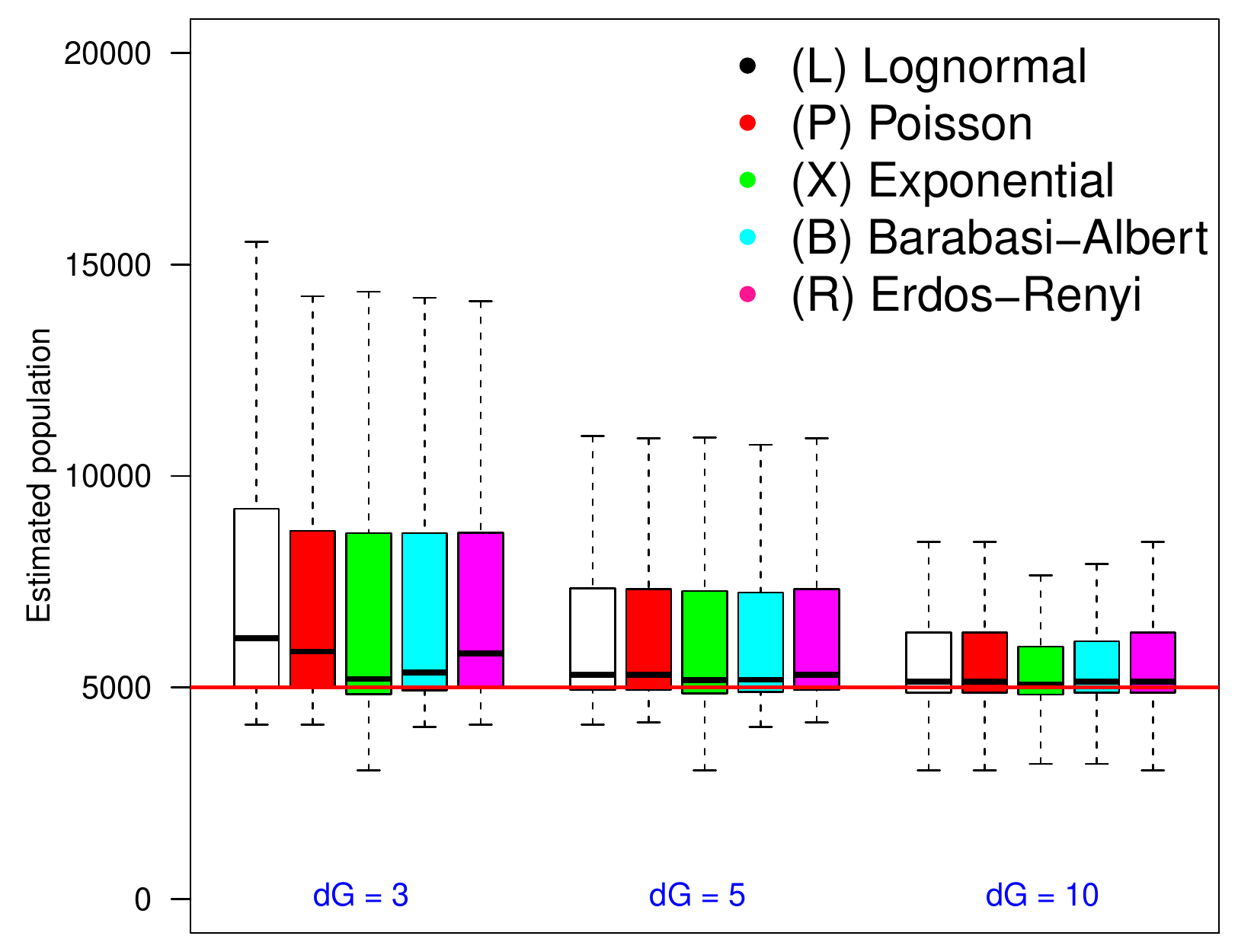}}\\
\subcaptionbox{$n$ = $10\cdot 10^3$, $r=250$\label{1b}}{\includegraphics[width = 0.3\linewidth]{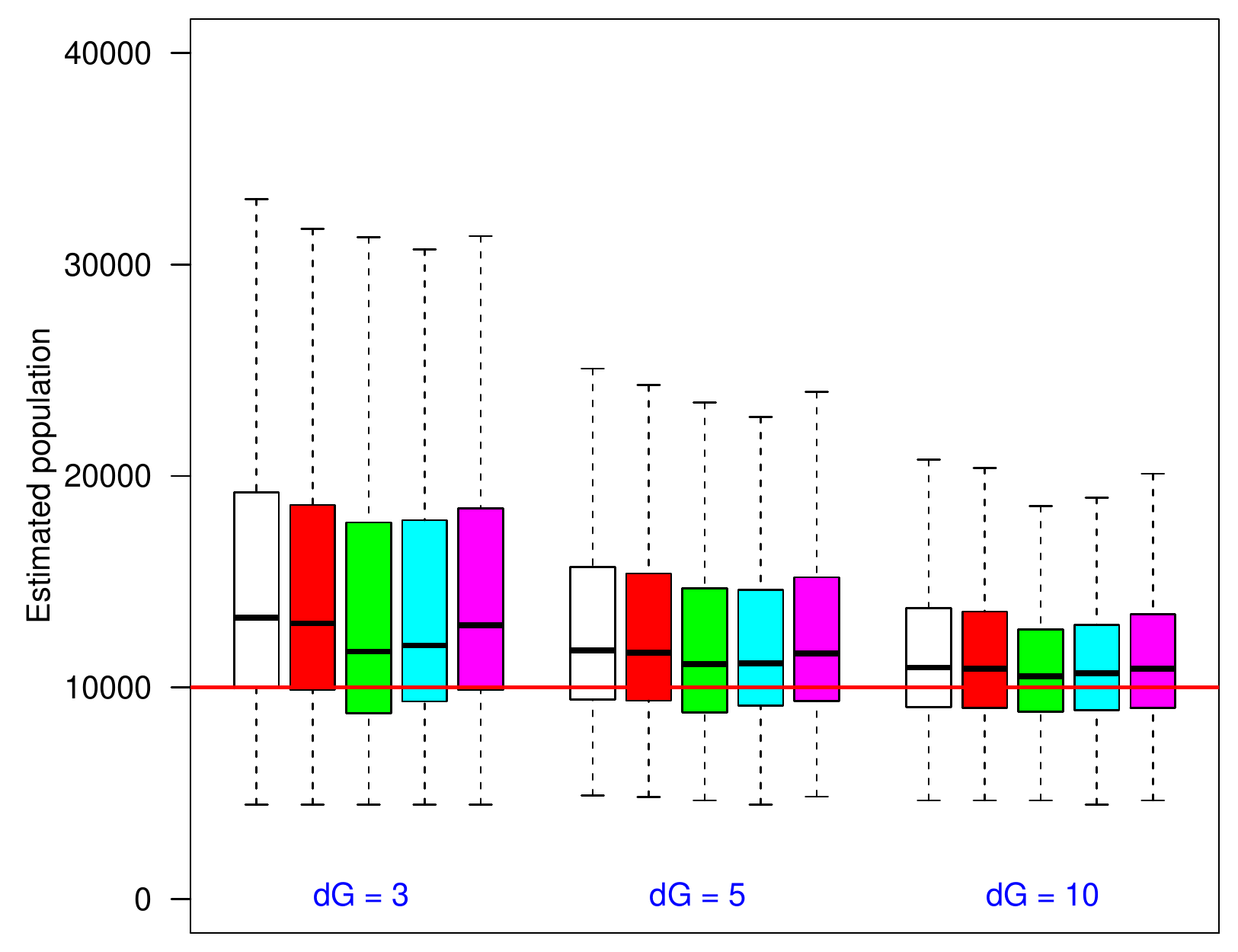}} &
\subcaptionbox{$n$ = $10\cdot 10^3$, $r=500$\label{2b}}{\includegraphics[width = 0.3\linewidth]{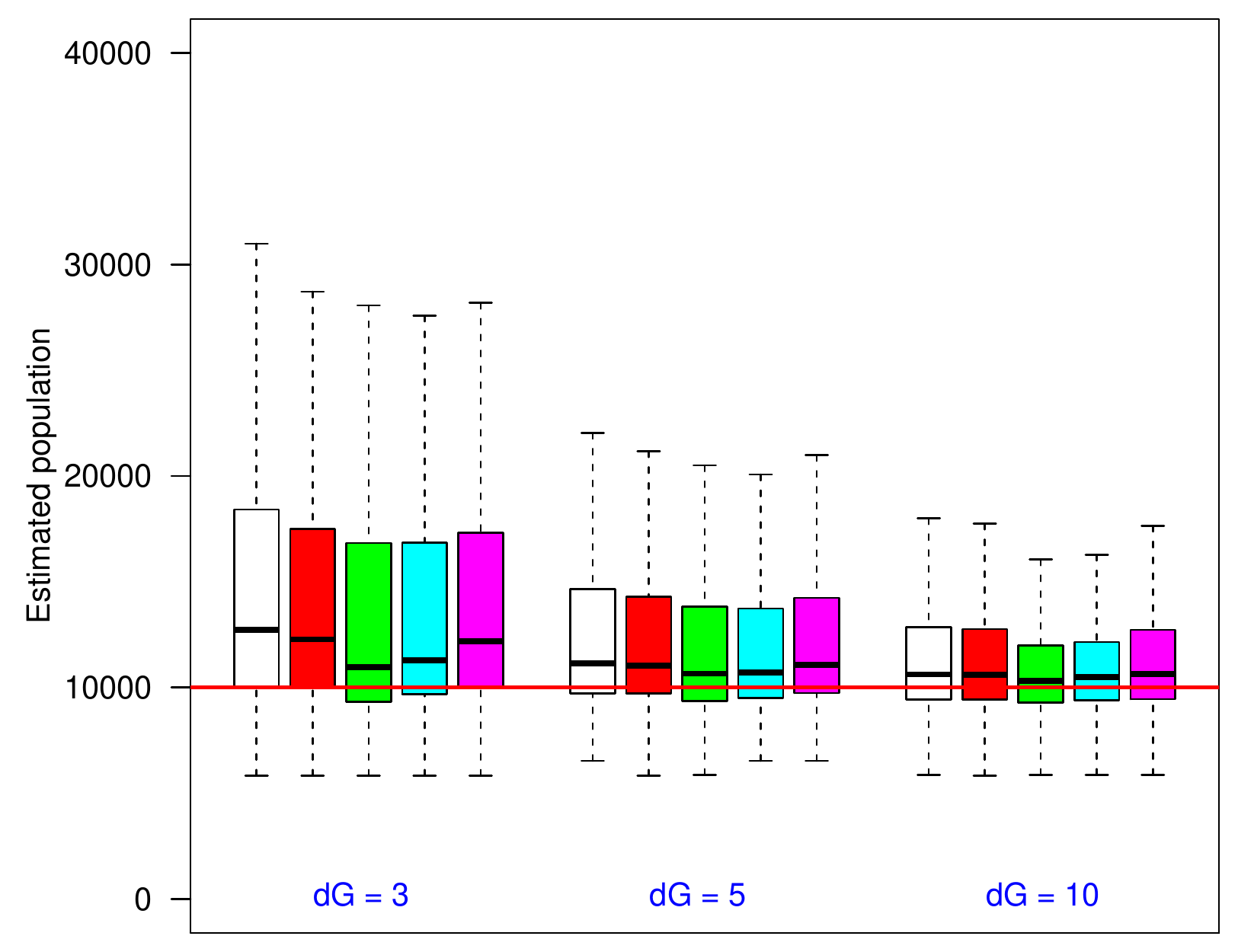}} &
\subcaptionbox{$n$ = $10\cdot 10^3$, $r=750$\label{3b}}{\includegraphics[width = 0.3\linewidth]{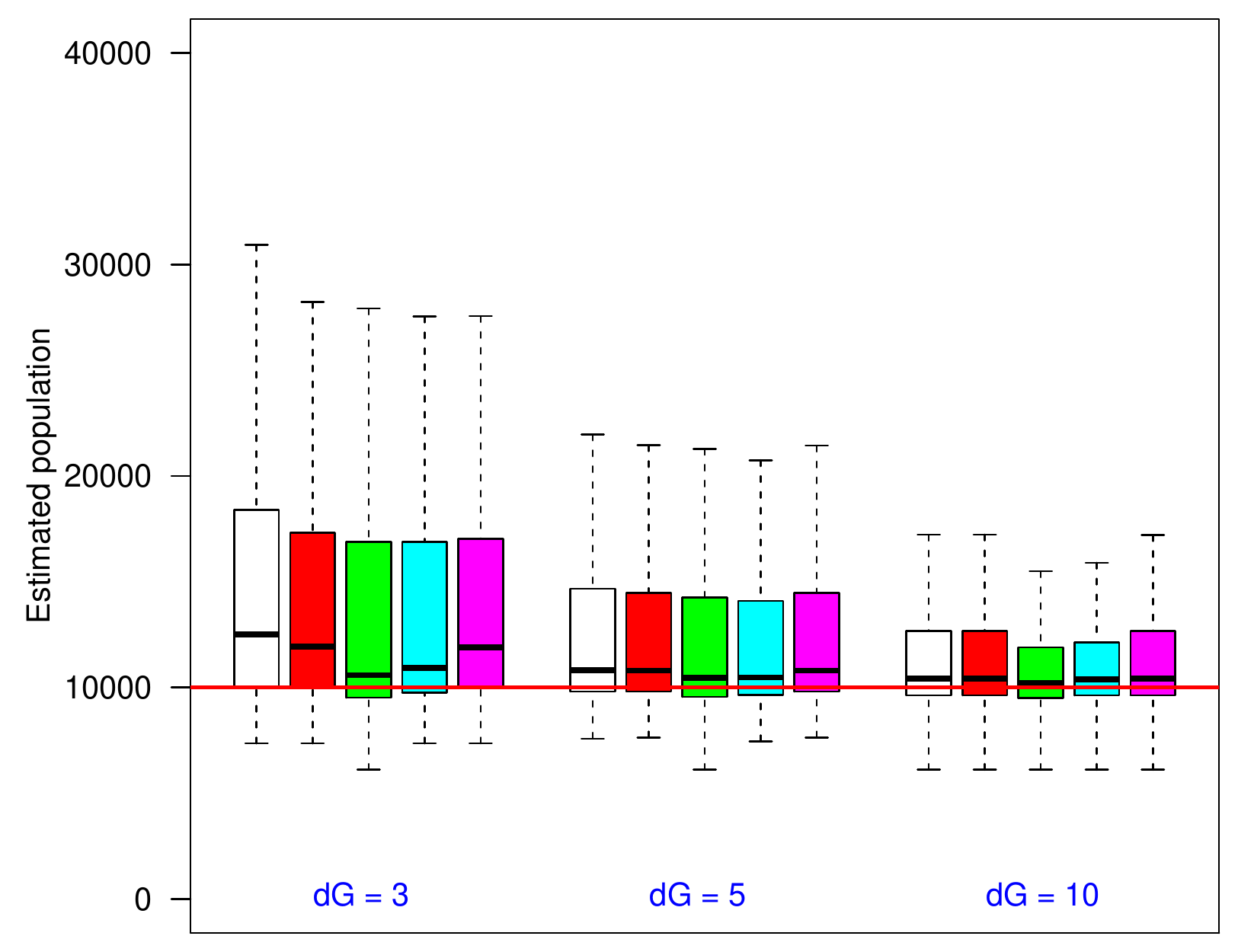}}\\
\subcaptionbox{$n$ = $20\cdot 10^3$, $r=250$\label{1c}}{\includegraphics[width = 0.3\linewidth]{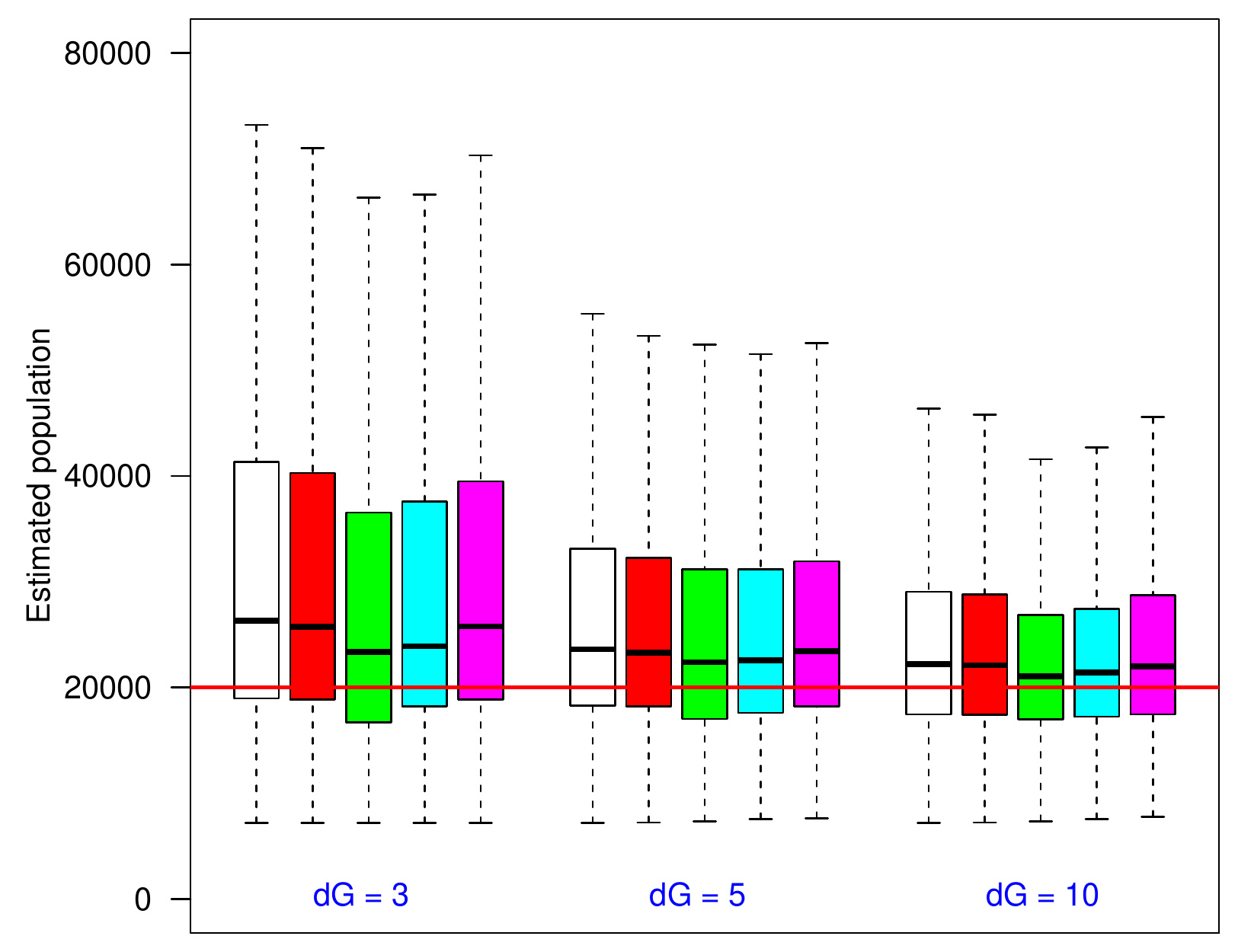}} &
\subcaptionbox{$n$ = $20\cdot 10^3$, $r=500$\label{2c}}{\includegraphics[width = 0.3\linewidth]{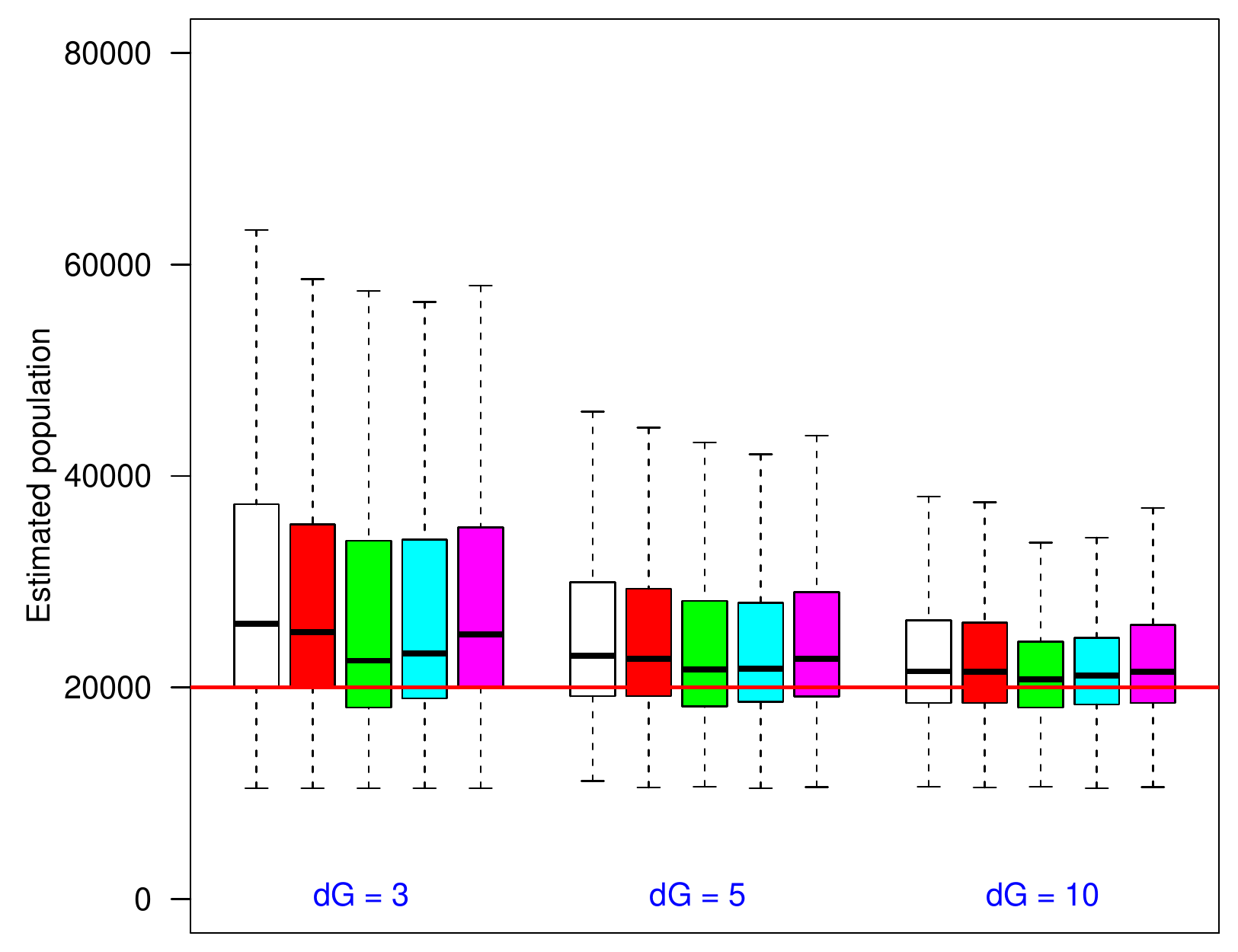}} &
\subcaptionbox{$n$ = $20\cdot 10^3$, $r=750$\label{3c}}{\includegraphics[width = 0.3\linewidth]{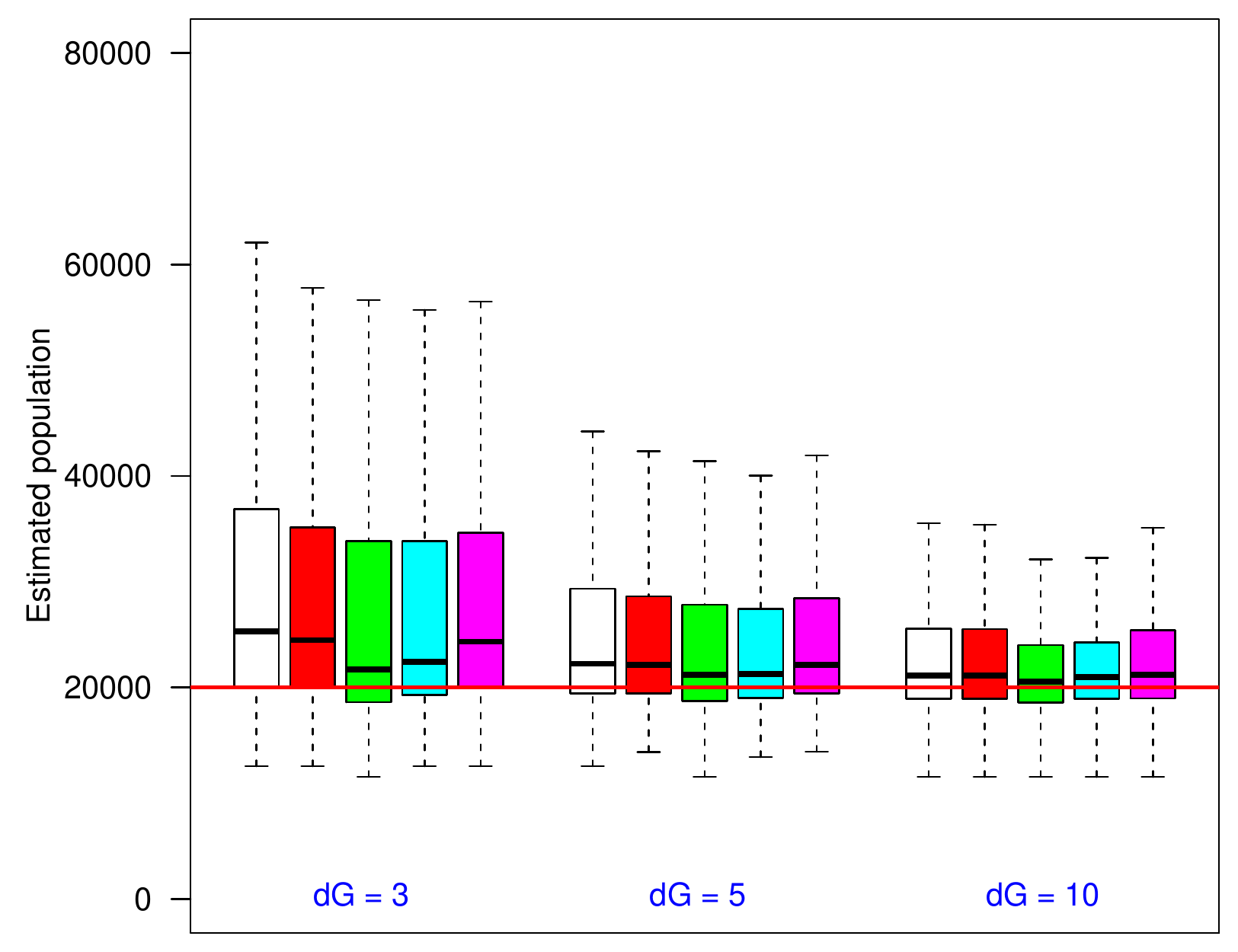}}\\[0pt]
\subcaptionbox{$n$ = $40\cdot 10^3$, $r=250$\label{1d}}{\includegraphics[width = 0.3\linewidth]{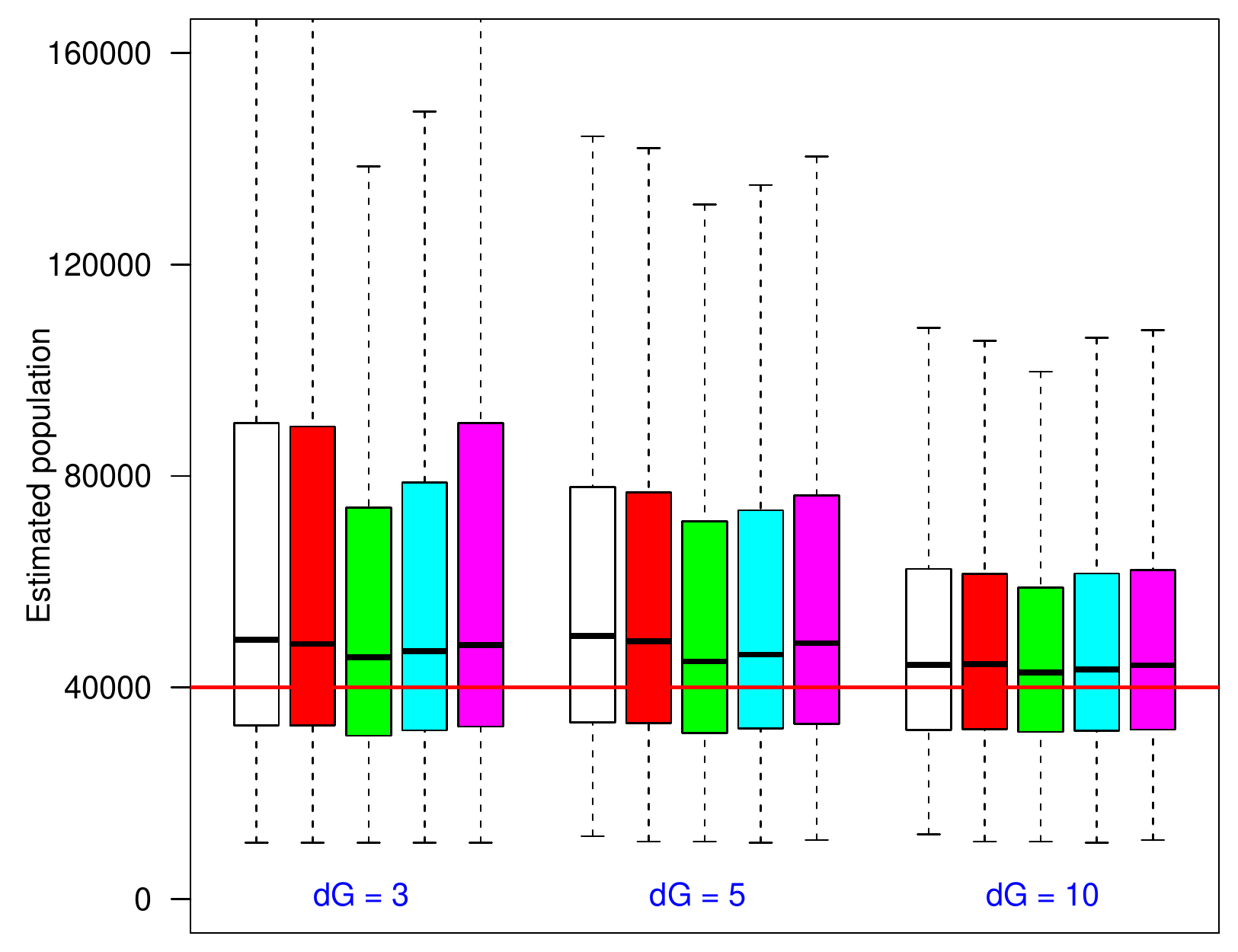}} &
\subcaptionbox{$n$ = $40\cdot 10^3$, $r=500$\label{2d}}{\includegraphics[width = 0.3\linewidth]{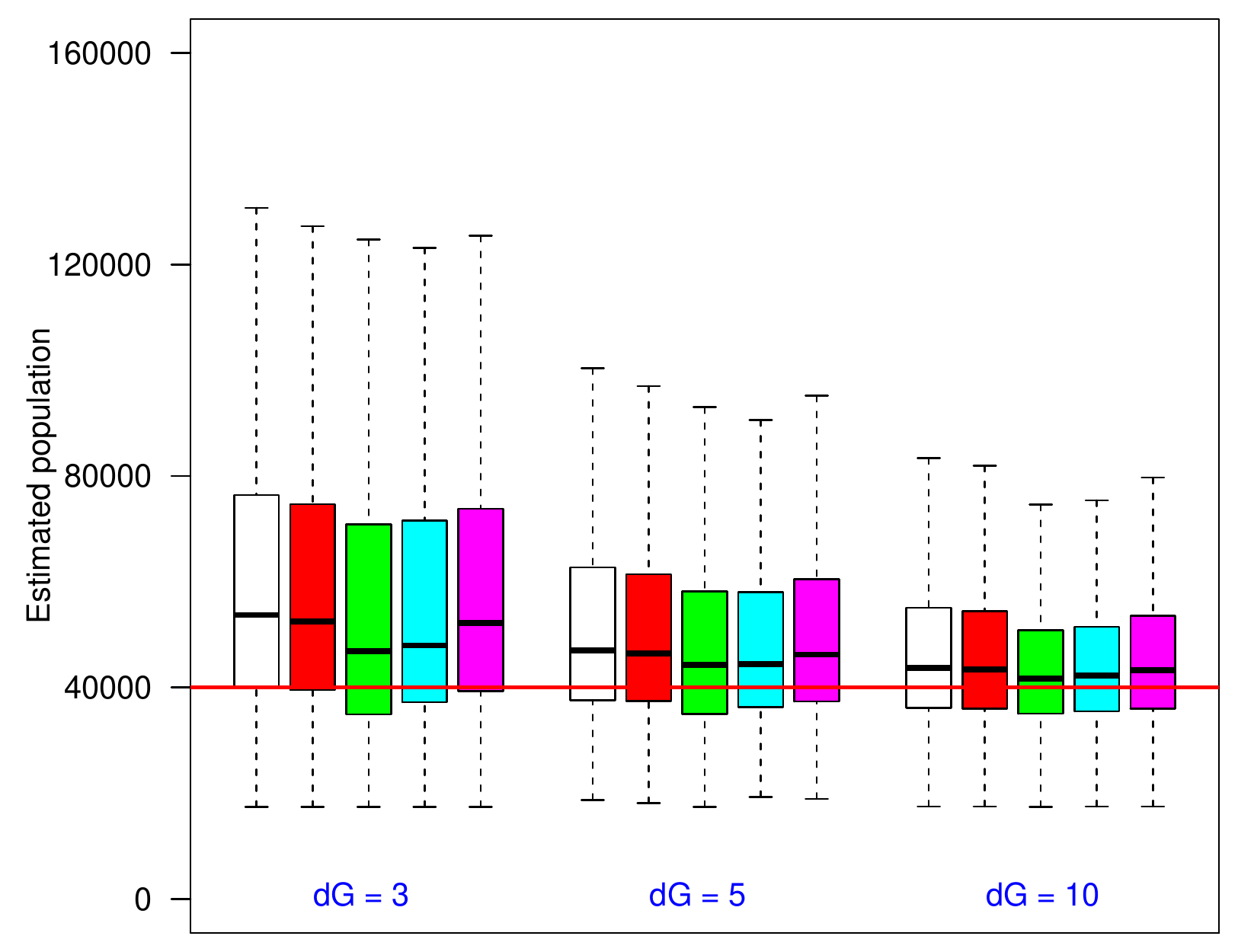}} &
\subcaptionbox{$n$ = $40\cdot 10^3$, $r=750$\label{3d}}{\includegraphics[width = 0.3\linewidth]{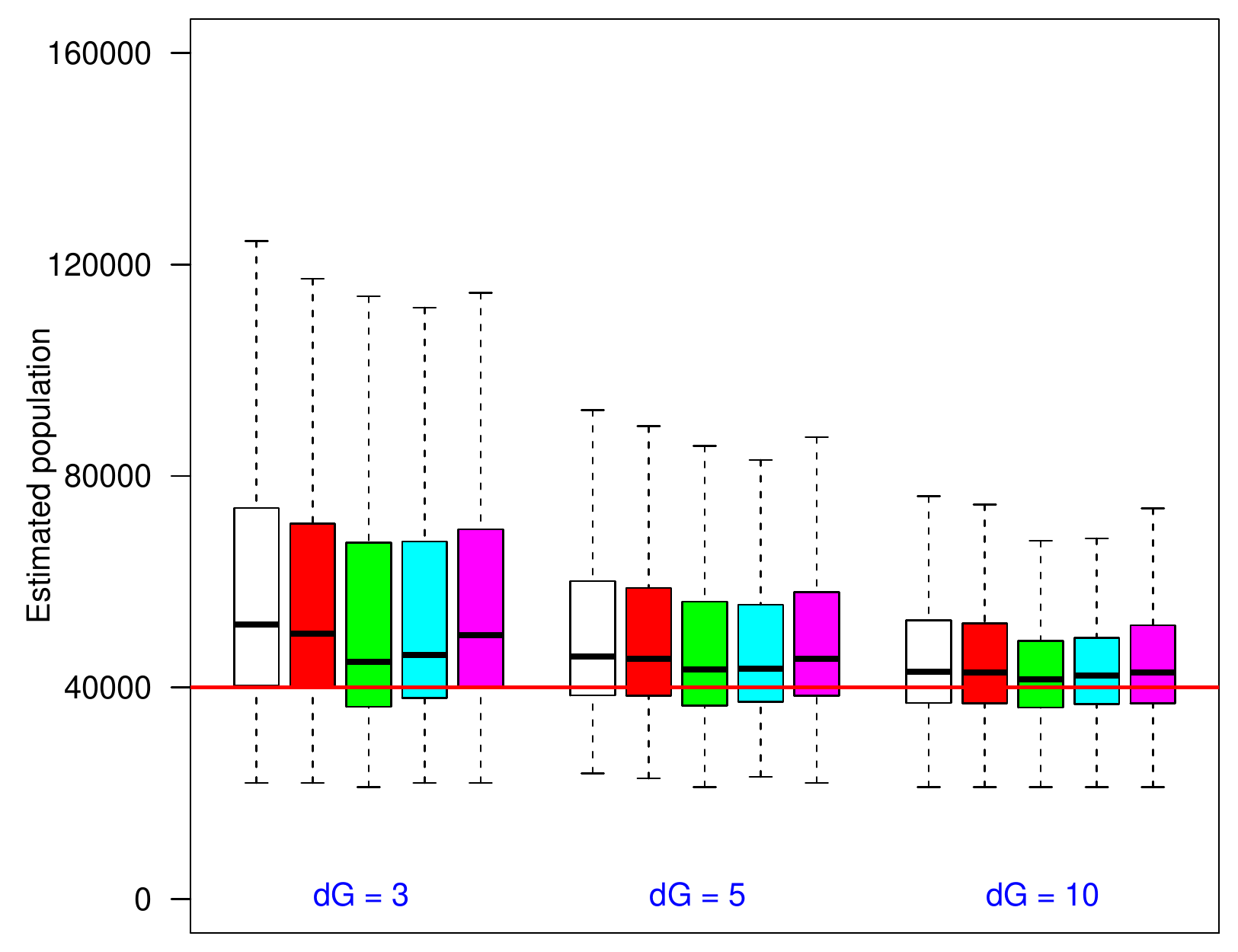}}
\end{tabular}
\caption{Estimator $n_1$ on uniform samples in populations of size $n=5\cdot 10^3$ to $40\cdot 10^3$. In each box, the thick line indicates the sample median; the top of the box is the median of the upper half of the estimated values (75\% quartile); the bottom of the box indicates the median of the lower half of the estimated values (25\% quartile; and the whiskers indicate the full range of estimated values. No (finite) outliers were removed.}
\label{results:n1}
\end{figure}

Figure \ref{results:n1} shows that as sample size increases, the medians of $n_1$ converge to the true population size.  For example, when $n=5\cdot 10^3$ and $r=250$, Exponential degree distribution graphs with $\lambda=3$ have a median $n_1$ value of 5663 (a 13.3\% offset from the true value of $n=5\cdot 10^3$).  In comparison, when $r=750$, the median for this family of graphs is 5204 (just 4.1\% offset from the true value).  As the sample size increases from $r=250$ to $r=750$, the error in the median estimate decreases by 9.2\%.  The benefit of increasing sample size diminishes as networks grow larger, however.  For example, for a network of size $n=40\cdot 10^3$, increasing the sample size from $r=250$ to $r=750$ causes the error in the median $n_1$ estimate to undergo only a 2\% change.

In addition, Figure \ref{results:n1} shows that as sample size increases, the interquartile range (IQR) of the estimates decrease. For example, when $n=5\cdot 10^3$ and $r=250$, Lognormal degree distribution graphs with $\lambda=10$ experience an interquartile range of 1950 in their $n_1$ estimates (35.9\% of the median).  In comparison, when $r=750$, the interquartile range for this family of graphs decreases to 1425 (a 26.9\% reduction).  The magnitude of this effect increases as networks grow larger.  For example, for a network of size $n=40\cdot 10^3$, increasing the sample size from $r=250$ to $r=750$ causes the interquartile range of the $n_1$ estimate to undergo a 48.6\% decrease.

\subsection{Evaluating $n_2$ on Synthetic Networks}
\label{sec:eval-n2}

The experiments in this and all subsequent section use respondent-driven samples.

\begin{assumption}
\label{def:rds-assumptions}
In all our experiments where RDS is used to generate samples, we take $|D|=7$ random seeds drawn uniformly at random from $V$.  Reflecting our experiences in the field \cite{coronado2017using}, we assumed that each subject has a 90\% chance of recruiting two subjects randomly from their ego network, and a 10\% chance of recruiting just one.  [Individuals with an ego network of size 1 will  recruit that one individual with 100\% probability, while individuals with an ego network of size 0 will recruit no one]. 
\end{assumption}

The $12$ graphs in Figure \ref{results:n2} present the performance of the $n_2$ estimator as the true population size $n$ is varied from $5\cdot 10^3$ to $40\cdot 10^3$ (vertical axis of the grid) and the size of the RDS sample is varied from $250$ to $750$ (horizontal axis of the grid).  In each of the $12$ graphs, the x-axis varies the average degree $\lambda$ from $3$ to $10$.  For each choice of $\lambda$, the medians and quartile ranges of $n_2$ are given for each of the $5$ graph families.  Each of these is determined by $900$ simulations ($30$ graphs times $30$ uniformly drawn samples in each graph).  

\begin{figure}[tbp]
\setlength{\belowcaptionskip}{12pt}
\centering
\begin{tabular}{ScScScSc}
\subcaptionbox{$n$ = $5\cdot 10^3$, $r=250$\label{1a}}{\includegraphics[width = 0.3\linewidth]{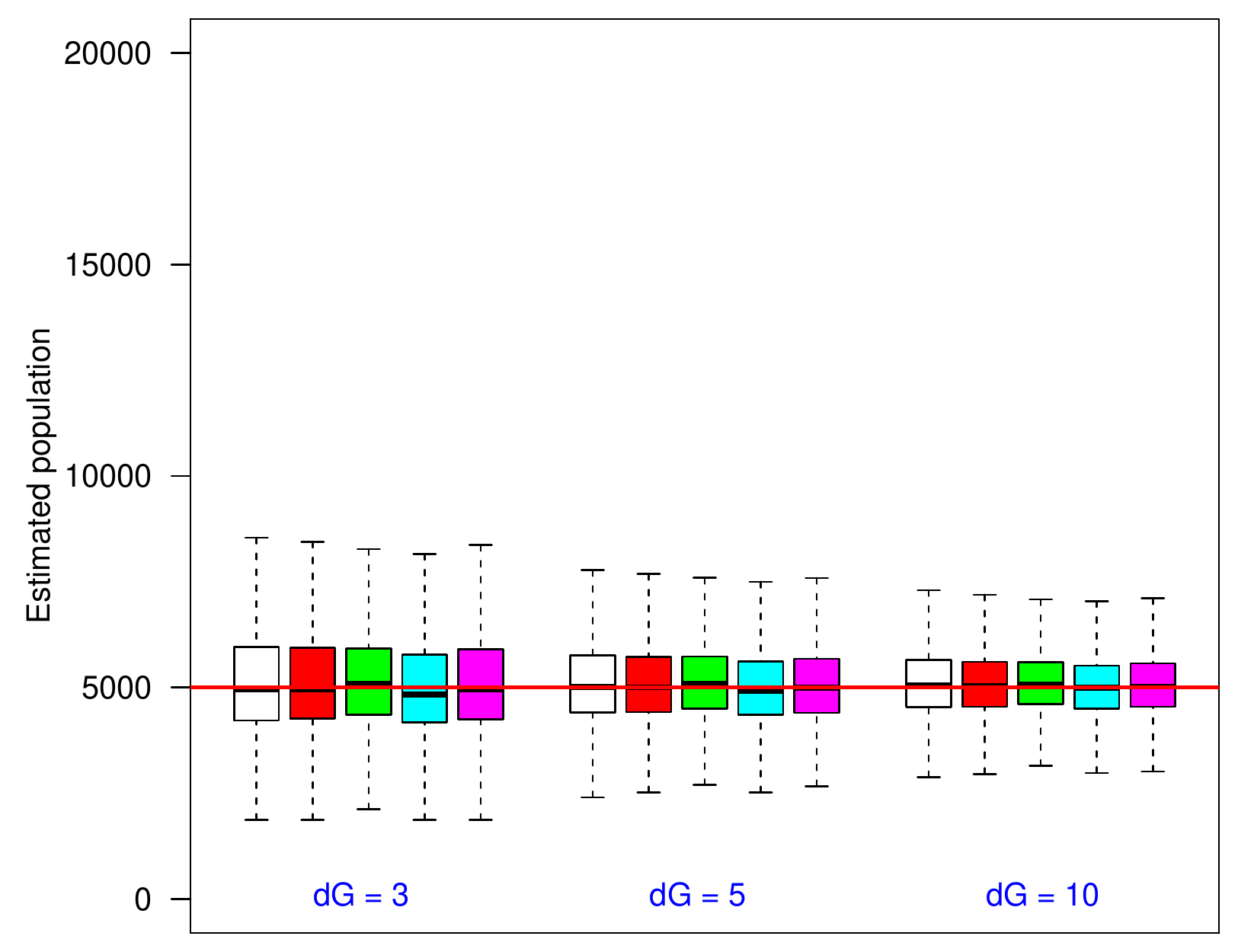}} &
\subcaptionbox{$n$ = $5\cdot 10^3$, $r=500$\label{2a}}{\includegraphics[width = 0.3\linewidth]{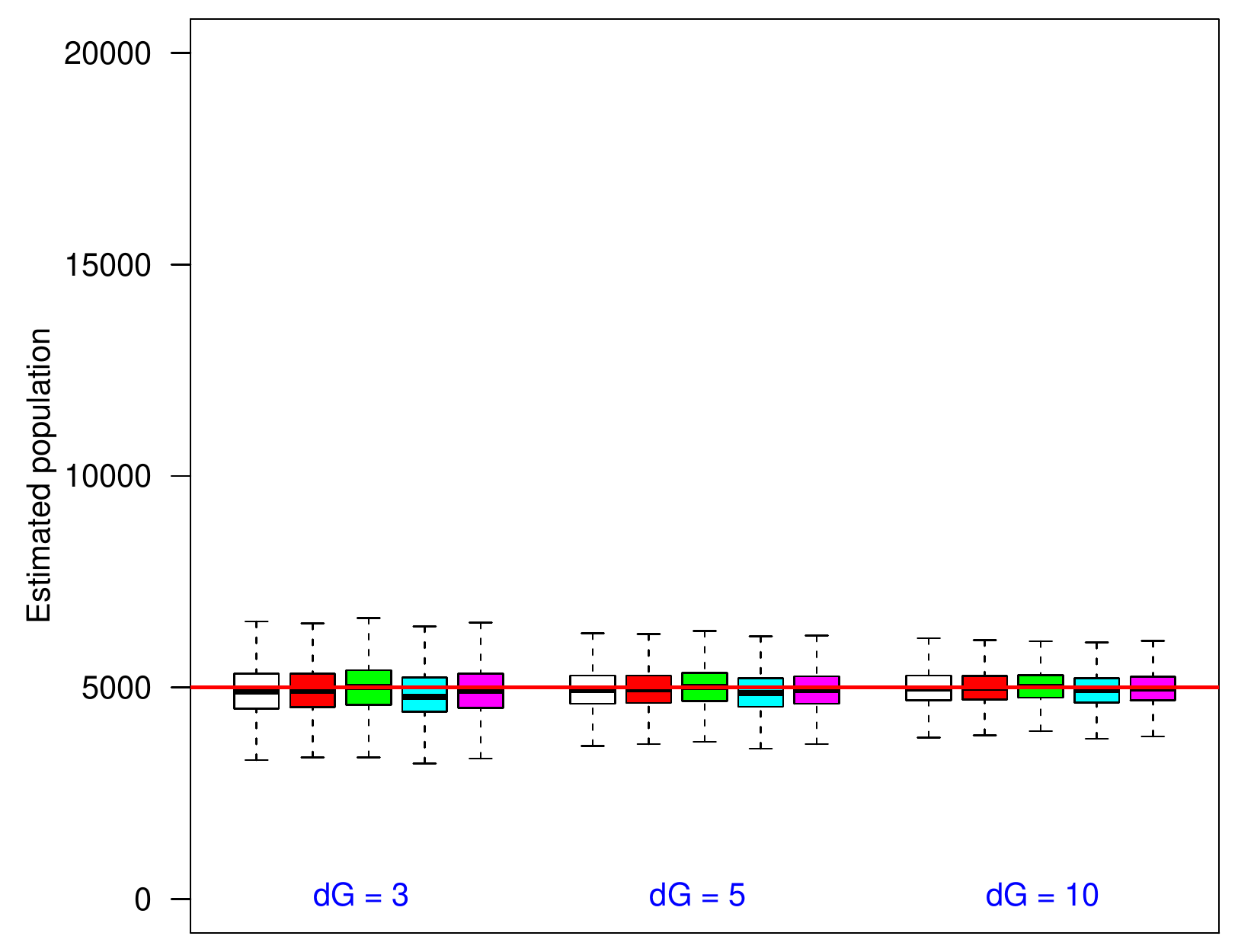}} &
\subcaptionbox{$n$ = $5\cdot 10^3$, $r=750$\label{3a}}{\includegraphics[width = 0.3\linewidth]{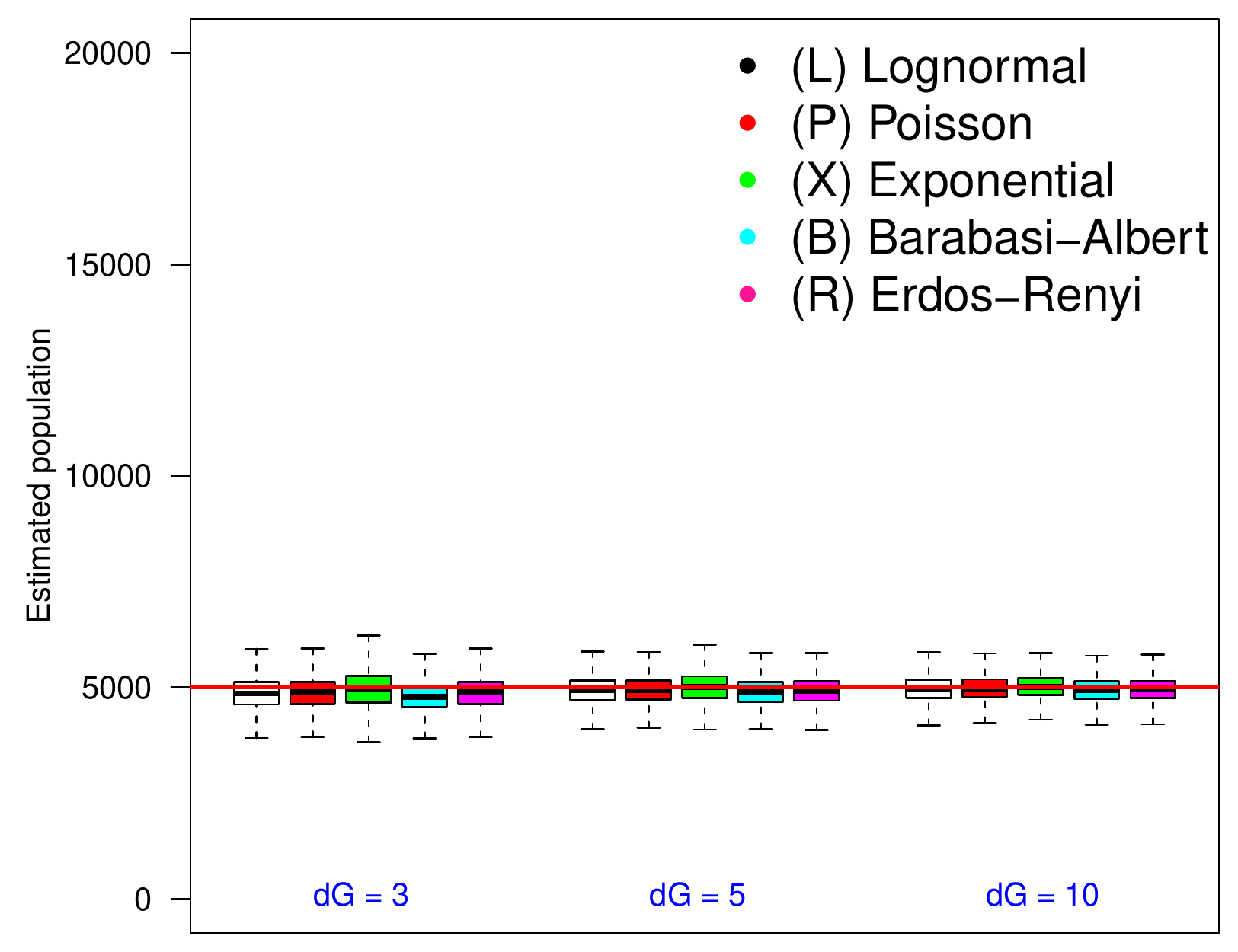}}\\[0pt]
\subcaptionbox{$n$ = $10\cdot 10^3$, $r=250$\label{1b}}{\includegraphics[width = 0.3\linewidth]{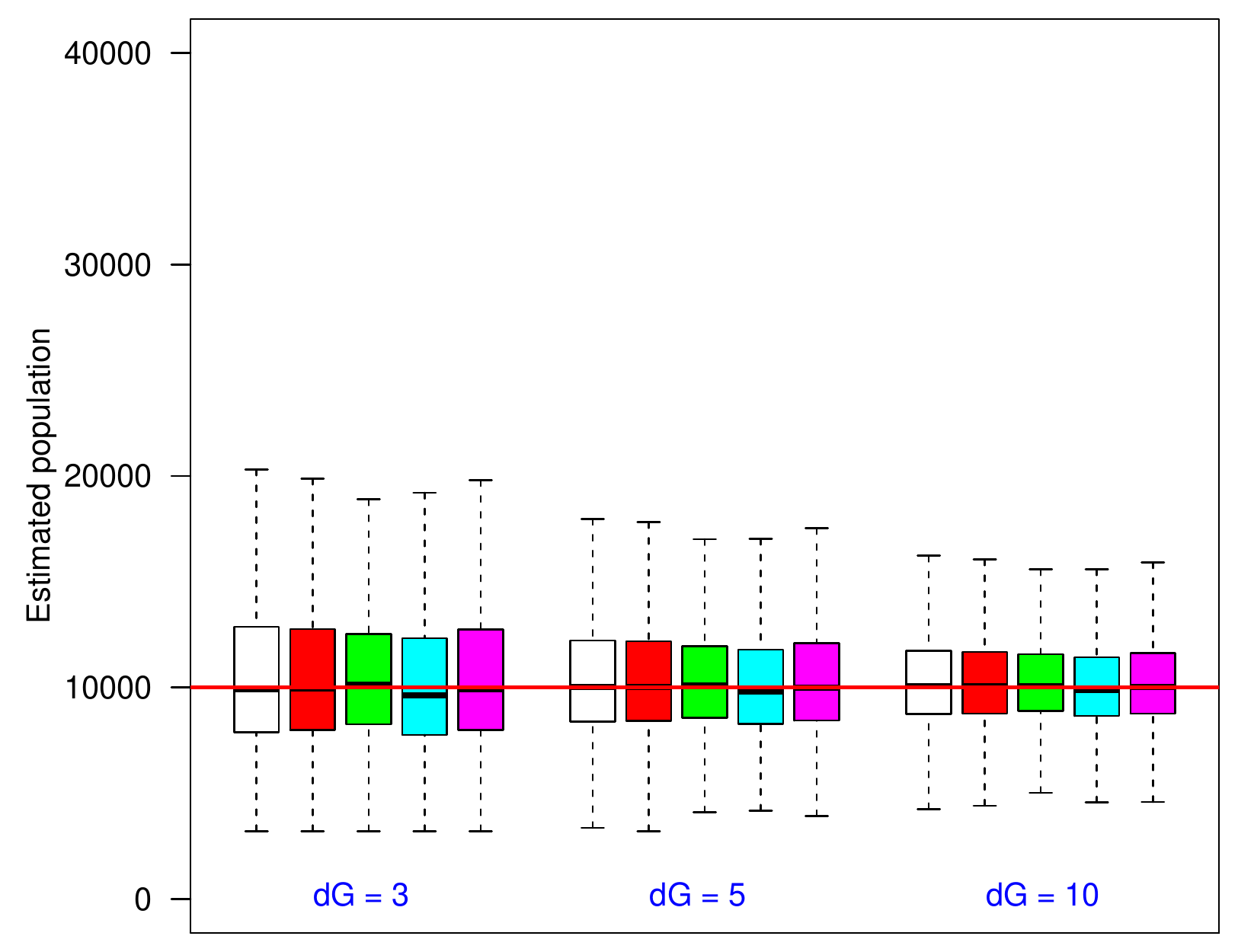}} &
\subcaptionbox{$n$ = $10\cdot 10^3$, $r=500$\label{2b}}{\includegraphics[width = 0.3\linewidth]{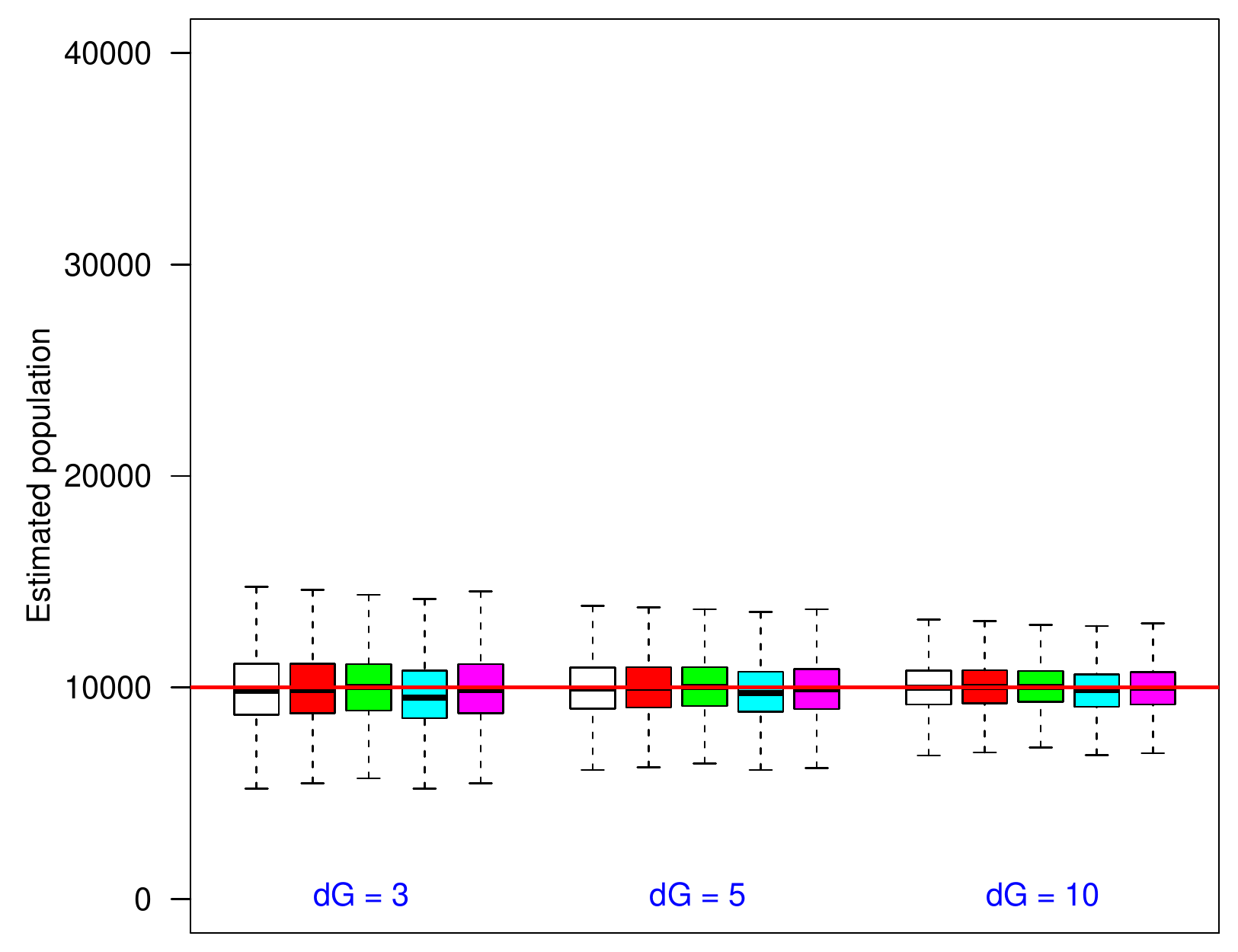}} &
\subcaptionbox{$n$ = $10\cdot 10^3$, $r=750$\label{3b}}{\includegraphics[width = 0.3\linewidth]{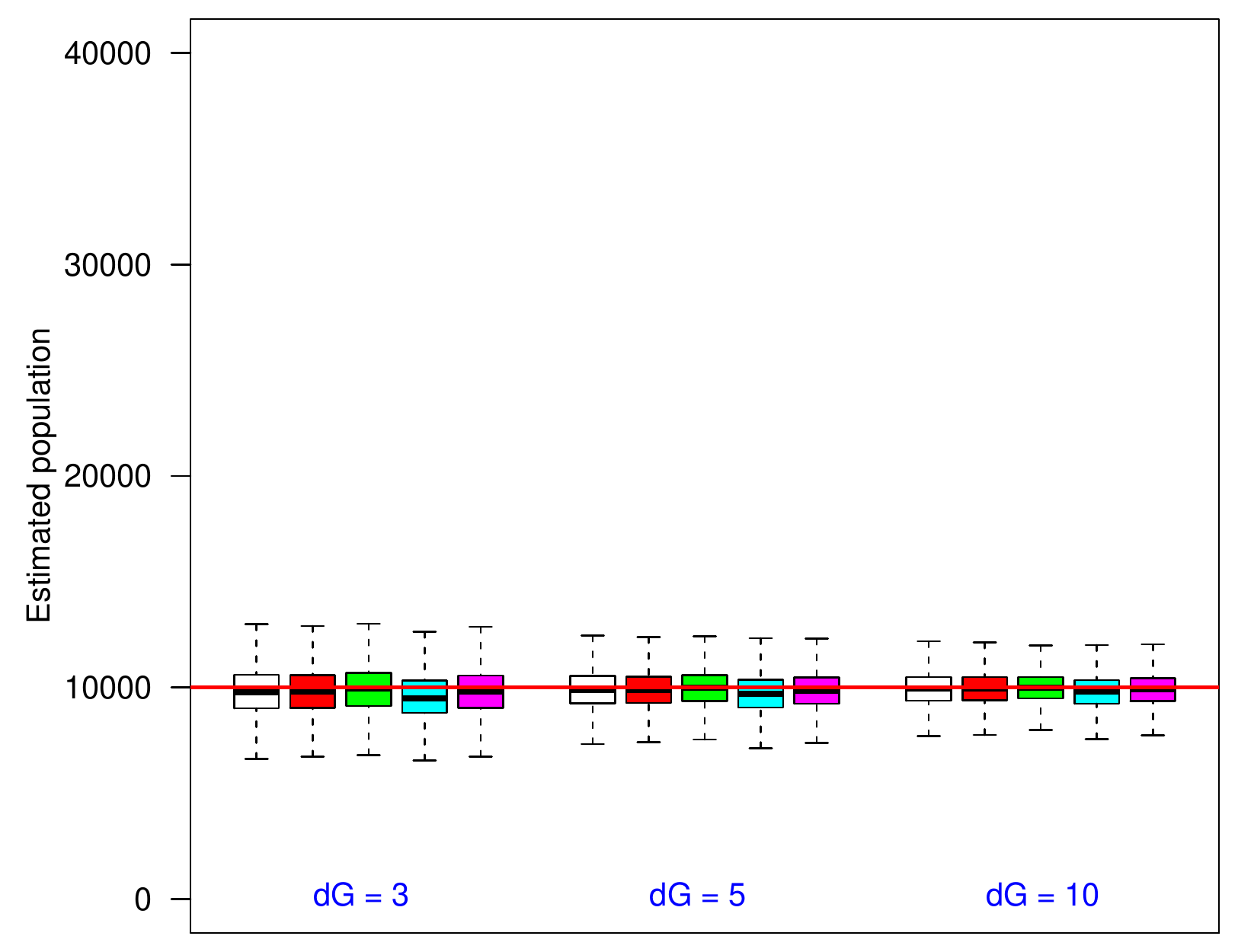}}\\
\subcaptionbox{$n$ = $20\cdot 10^3$, $r=250$\label{1c}}{\includegraphics[width = 0.3\linewidth]{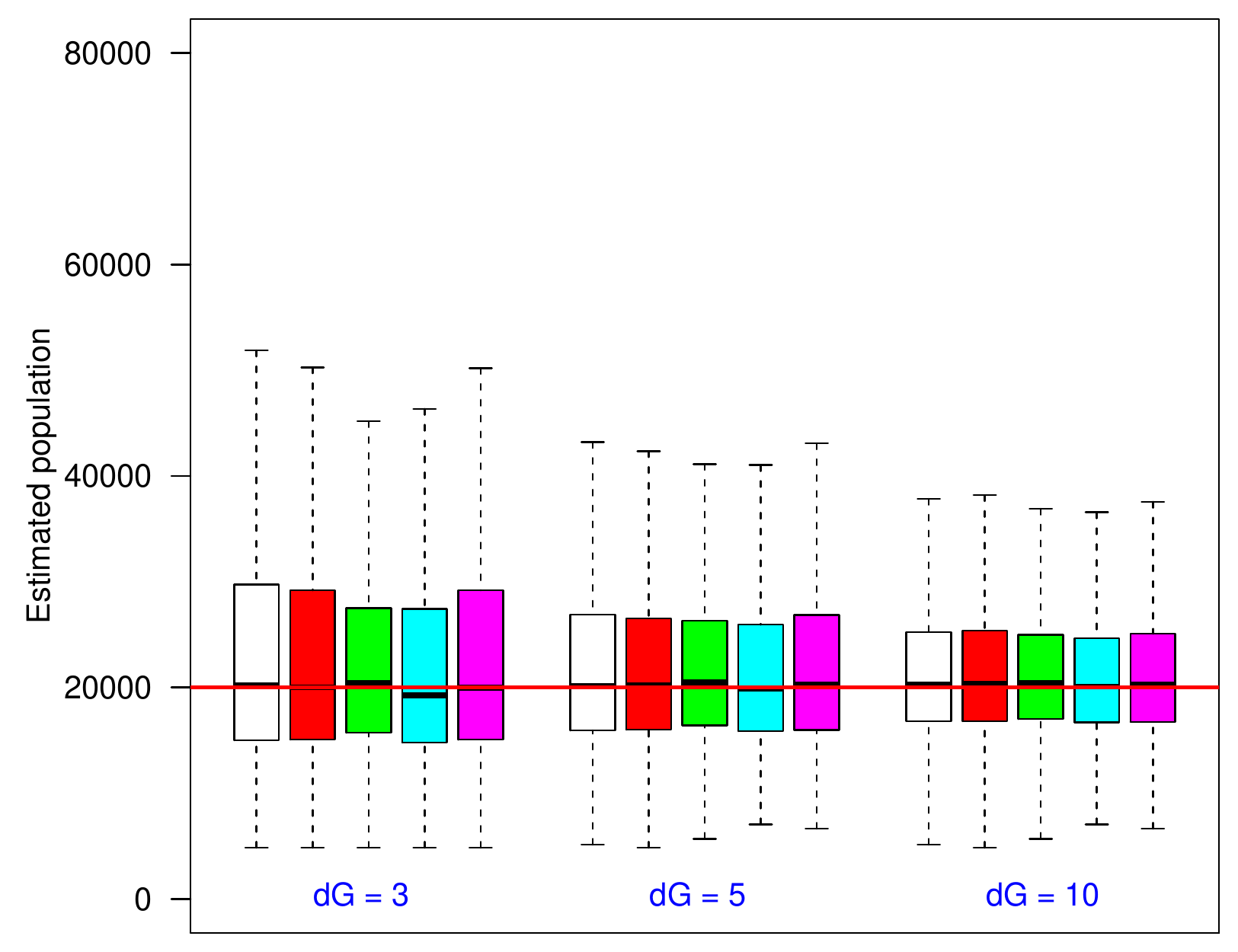}} &
\subcaptionbox{$n$ = $20\cdot 10^3$, $r=500$\label{2c}}{\includegraphics[width = 0.3\linewidth]{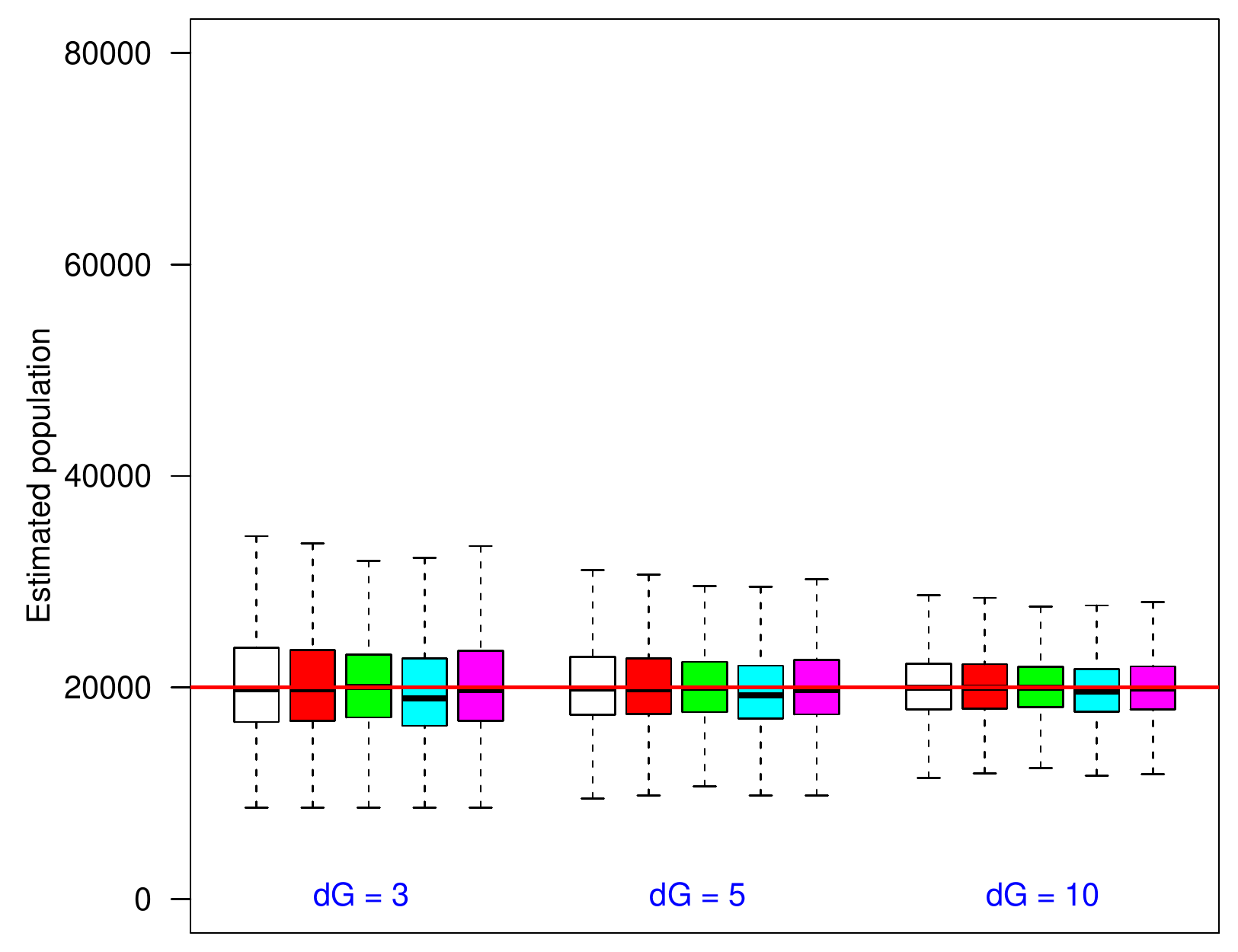}} &
\subcaptionbox{$n$ = $20\cdot 10^3$, $r=750$\label{3c}}{\includegraphics[width = 0.3\linewidth]{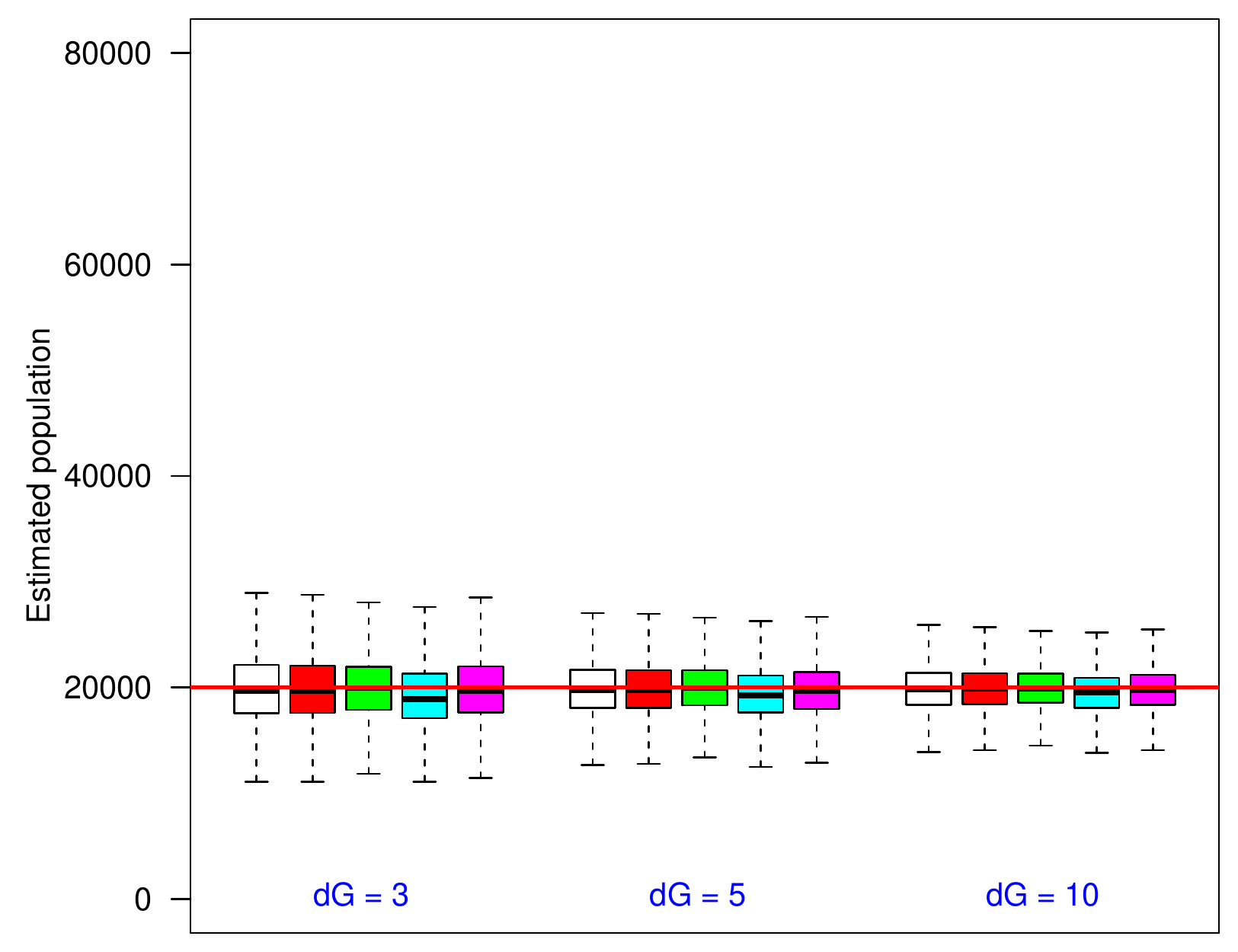}}\\[0pt]
\subcaptionbox{$n$ = $40\cdot 10^3$, $r=250$\label{1d}}{\includegraphics[width = 0.3\linewidth]{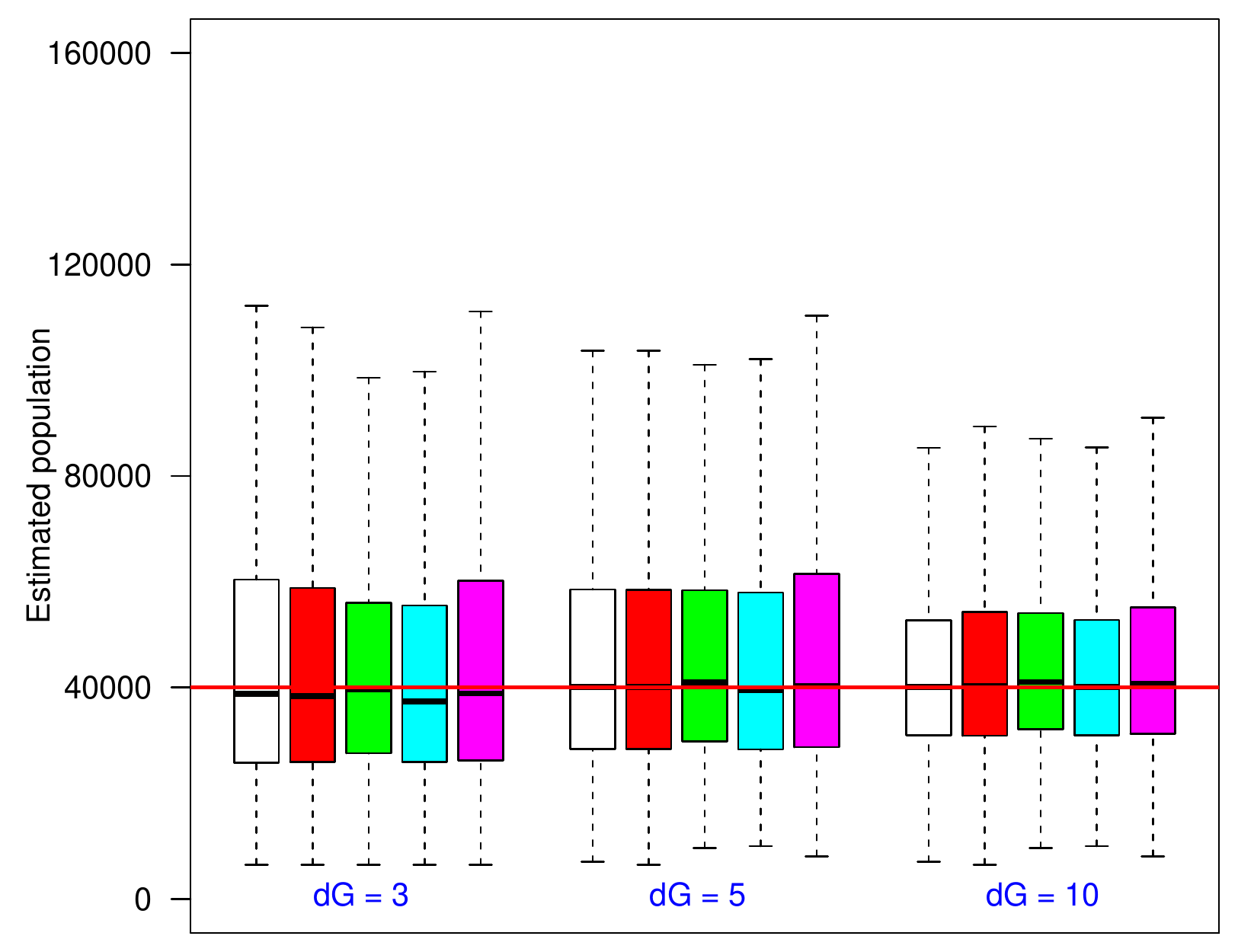}} &
\subcaptionbox{$n$ = $40\cdot 10^3$, $r=500$\label{2d}}{\includegraphics[width = 0.3\linewidth]{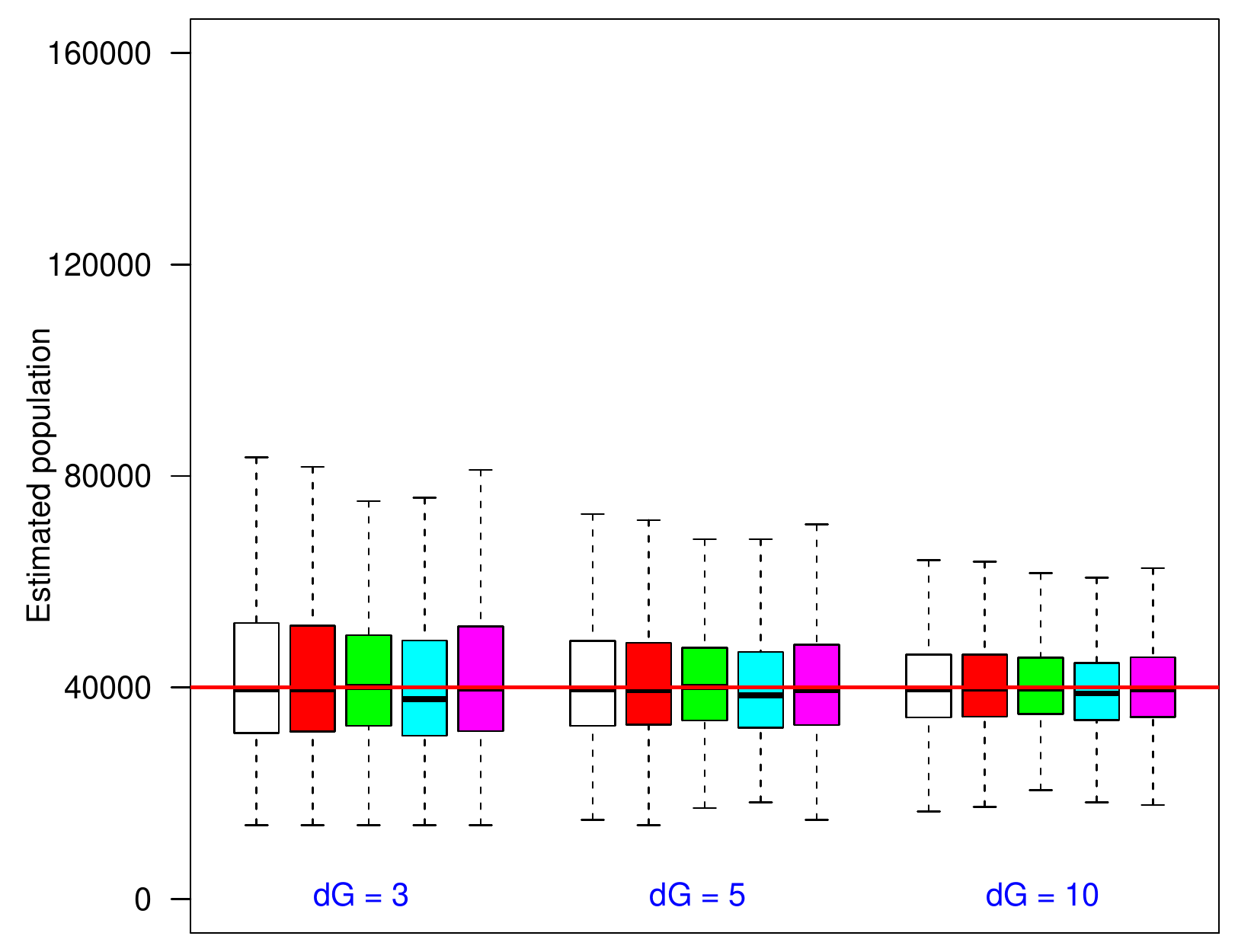}} &
\subcaptionbox{$n$ = $40\cdot 10^3$, $r=750$\label{3d}}{\includegraphics[width = 0.3\linewidth]{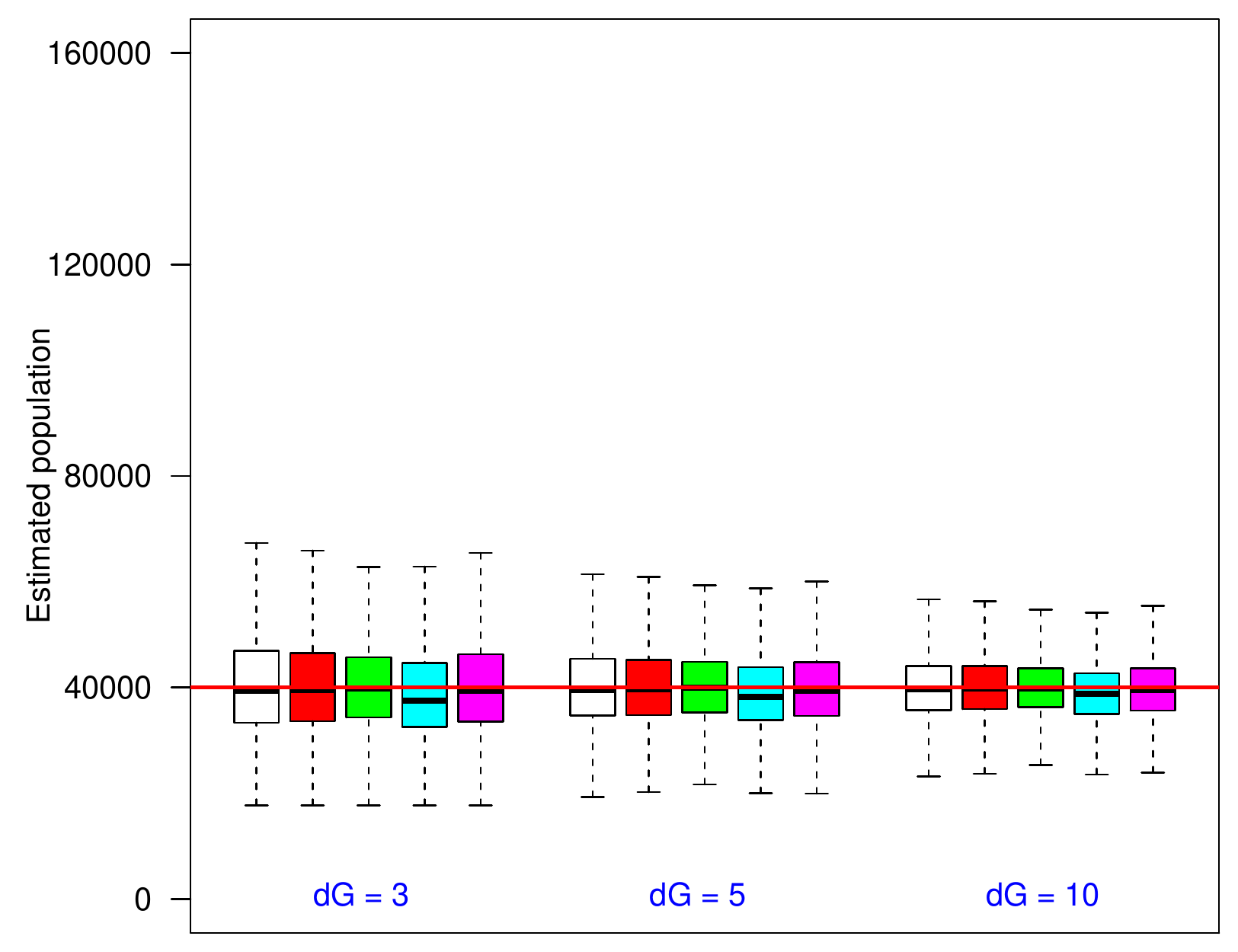}}
\end{tabular}
\setlength{\abovecaptionskip}{1cm}
\caption{Estimator $n_2$ on RDS samples in populations of size $n=5\cdot 10^3$ to $40\cdot 10^3$.  In each box, the thick line indicates the sample median; the top of the box is the median of the upper half of the estimated values (75\% quartile); the bottom of the box indicates the median of the lower half of the estimated values (25\% quartile; and the whiskers indicate the full range of estimated values. No (finite) outliers were removed.}
\label{results:n2}
\end{figure}

Figure \ref{results:n2} shows that the median of $n_2$ converges to the true population size across a range of topologies, RDS sample sizes, and overall populations. In addition, Figure \ref{results:n2} shows that as sample size increases, the interquartile difference decrease. For example, when $n=5\cdot 10^3$ and $r=250$, Poisson degree distribution graphs with $\lambda=3$ experience an interquartile range of 1676 in their $n_2$ estimates (33.8\% of the median).  In comparison, when $r=750$, the interquartile range for this family of graphs decreases to 524 (a 68.7\% reduction).  The magnitude of this effect decreases as networks grow larger, such that, for a network of size $n=40\cdot 10^3$, increasing the sample size from $r=250$ to $r=750$ causes the interquartile range of the $n_2$ estimate to undergo a 60.8\% decrease. However, the total range of estimates as a proportion of the median decreases as sample size increases, indicating decreasing sample-based variance (a key concern in RDS sampling \cite{verdery2017social}).

\subsection{Estimating Population Size via RDS Samples in the Presence of Clustering}
\label{sec:n3}
Beyond the oversampling of high degree nodes, RDS faces challenges when used in networks where network clustering is pronounced \cite{salganik_assessing_2011,verdery_new_2017}. While methods are available to assess the presence of clustering \cite{gile_diagnostics_2015}, and recent work has proposed new techniques to estimate and account for clustering from a single RDS sample \cite{verdery2017social}, the effects of this phenomenon on population size estimation from RDS samples is seldom discussed. The root of the problem lies in the fact that RDS walks necessarily sample network neighborhoods. Where neighbors show high levels of network transitivity, counts of common edges will produce high numbers of ``matches'' that appear in the denominator of both $n_1$ and $n_2$. This will bias the estimates of overall population size derived from these estimators toward underestimation of the total network size. To address this, we propose a modification of $n_2$ that discounts matching free ends discovered within a single RDS sampling tree and, for purposes of estimation, counts only those that occur across trees. The next Definition introduces formalisms necessary to make this precise.

\begin{definition}
\label{def:cross-seeds}
Let $G=(V,E)$, a let set $S \subseteq V$, and $H=(S,F)$ a subgraph on $S\subseteq V$ with edge set $F \subseteq E \cap (S \times S)$ obtained by respondent driven sampling from a set of seeds $D\subseteq S$ where $|D|>1$.  Define  the function $\gamma: S\rightarrow D$ associating each $u\in S$ with the unique seed $\gamma(u) \in D$ from which $u$ was discovered through a sequence of referrals.
For each $u\in S$, the component of $u$ is denoted 
\begin{align}
C_\gamma(u) &\coloneqq \{ v \;|\; \gamma(v) = \gamma(u) \} \subseteq S \label{def:C}
\intertext{while its complement is written $\widetilde{C}_\gamma(u) \coloneqq S \backslash C_\gamma(u)$.  Note that $C_\gamma(u) \cap \widetilde{C}_\gamma(u) = \emptyset$. For each seed $s\in D$, we define the {\em cross-seed matches} from the $C_\gamma(u)$ component (in $G$ modulo $H$) as the disjoint union (multiset)}
X(s, F, \gamma) &\coloneqq \coprod_{u\in C_\gamma(s)} \left( N(u,F) \cap \widetilde{C}_\gamma(s) \right) \subseteq V \label{def:X}
\intertext{whose cardinality is denoted}
\langle X(s, F, \gamma) \rangle  &\coloneqq \sum_{u\in C_\gamma(s)} \left| N(u,F) \cap \widetilde{C}_\gamma(s) \right|. \nonumber
\end{align}
\end{definition}

The next estimator $n_3$, provides a revised estimate $|V|$ from a respondent driven sample $S \subseteq V$, discounting matches that occr within the same RDS component.

\begin{definition}
\label{def:n3}
Given a graph $G=(V,E)$, a set $S \subseteq V$, and $H=(S,F)$ a subgraph on $S\subseteq V$ with edge set $F \subseteq E \cap (S \times S)$.  Take $D\subseteq S$ satisfying $|D|>1$ and
$$
s_1 \neq s_2 \implies C_\gamma(s_1) \cap C_\gamma(s_2) = \emptyset.
$$
Define
\begin{align}
\label{eq:n3}
n_3(S, F, D, \gamma) &\coloneqq \frac{\sum_{s\in D} \frac{d(\widetilde{C}_\gamma(s))-1}{\widetilde{d}(S)} \cdot |\widetilde{C}_\gamma(s)| \cdot \langle R(C_\gamma(s), F) \rangle}{\sum_{s\in D} \langle X(s, F, \gamma) \rangle}
\end{align}
\end{definition}

The next proposition gives sufficient conditions under which respondent-driven samples $S\subseteq V$ produce consistent estimates $n_3(T) \sim |V|$ when $|V|$ is large.

\begin{proposition}
\label{prop:n3}
For $n=1,2,\ldots$ let $G_n=(V_n,E_n)$ be a graph on $|V_n|=f(n)$ vertices obtained by configuration graph sampling via degree distribution ${\cal D}_n$, where $f(n)$ grows unboundedly.  Let $c_n \in (0,1]$, and take $S_n\subseteq V_n$ to be a subset of size $|S_n| = \lfloor c_n\cdot f(n) \rfloor$ selected using RDS sampling in $G_n$ from $|D_n|>1$ seeds.   Define the random variable
  \begin{align*}
\Delta_n &\coloneqq \frac{d(S_n)-1}{\widetilde{d}(S_n)}.
  \end{align*}
  Accepting Assumption \ref{assumption-harmonic}, if $c_n\cdot f(n)/D_n$ diverges as $n$ goes to infinity, while 
    \begin{align}
\Delta_{n}^2 \cdot c_n^2 \cdot d(V_n) \cdot \frac{|D_n|-1}{|D_n|} = \frac{{\left( d(S_n)-1 \right)}^2 \cdot c_n^2}{\widetilde{d}(S_n)} \cdot \frac{|D_n|-1}{|D_n|}\xrightarrow{\phantom{xxx}p\phantom{xxx}} \Theta_3 \label{n3:condition}
  \end{align}
for some finite constant $\Theta_3 > 0$, then $\frac{n_3(S_n, F_n, D_n, \gamma)}{f(n)}$ necessarily converges to $1$.
\end{proposition}
\begin{proof}
Since each seed $s\in D_n$ is chosen uniformly at random, and RDS-CAPTURE recruits  from all seeds {\em in parallel}, and $|S_n| = \lfloor c_n \cdot f(n) \rfloor$ diverges, for random $s \in D_n$, we know that
\begin{align}
|C_\gamma(s)| &\xrightarrow{\phantom{xxx}p\phantom{xxx}} \frac{1}{|D_n|}\cdot |S_n| = \frac{c_n \cdot f(n)}{|D_n|}\label{eq:compsize}\\[0.1in]
|\widetilde{C}_\gamma(s)| &\xrightarrow{\phantom{xxx}p\phantom{xxx}}\frac{|D_n|-1}{|D_n|}\cdot |S_n| = \frac{|D_n|-1}{|D_n|}\cdot c_n \cdot f(n)\label{eq:complement-component}\\[0.1in]
d({C}_\gamma(s)), \;\; d(\widetilde{C}_\gamma(s)) &\xrightarrow{\phantom{xxx}p\phantom{xxx}} d(S_n)\label{eq:degcomp}
\intertext{Combining (\ref{eq:compsize}) and (\ref{eq:degcomp}), we conclude}
\langle R(C_\gamma(s), F_n) \rangle &\xrightarrow{\phantom{xxx}p\phantom{xxx}} \frac{\langle R(S_n, F_n)\rangle}{|D_n|}.\label{eq:reports-component}
\intertext{Sufficient reasoning about the configuration graph construction process tells us}
\langle X(s, F_n, \gamma) \rangle &\xrightarrow{\phantom{xxx}p\phantom{xxx}} \frac{1}{|D_n|} \cdot \langle M(S_n,F_n) \rangle\cdot \frac{|D_n|-1}{|D_n|}.
\end{align}
Define the following random variables, closely related to (\ref{def:rv1-prop2}) and (\ref{def:rv2-prop2}) of Proposition \ref{prop:n2}:
  \begin{align*}
R_n^{\circ} &\coloneqq \sum_{s\in D_n} \frac{d(\widetilde{C}_\gamma(s))-1}{\widetilde{d}(S)} \cdot |\widetilde{C}_\gamma(s)| \cdot \langle R(C_\gamma(s), F_n) \rangle \;/\; f(n)\\[0.1in]
M_n^{\circ} &\coloneqq \sum_{s\in D_n} \langle X(s, F_n, \gamma) \rangle \;/\; f(n).
  \end{align*}
As $n$ tends to infinity
\begin{align*}
R_n^{\circ} &\xrightarrow{\phantom{xxx}p\phantom{xxx}}  \frac{d(S_n)-1}{\widetilde{d}(S)} \left( \frac{|D_n|-1}{|D_n|} \cdot c_n \cdot f(n) \right) \cdot R_n^*(S_n, F_n)\\
M_n^{\circ} &\xrightarrow{\phantom{xxx}p\phantom{xxx}} \frac{|D_n|-1}{|D_n|} \cdot M_n^*(S_n, F_n)
\intertext{where $R_n^*(S_n, F_n)\xrightarrow{\phantom{xxx}p\phantom{xxx}} \Delta_n \cdot c_n \cdot d(V_n)$ as noted in (\ref{def:rv1-prop2}), while $M_n^*(S_n, F_n)\xrightarrow{\phantom{xxx}p\phantom{xxx}} \Delta_n^2 \cdot c_n^2 \cdot d(V_n)$ as noted in (\ref{def:rv2-prop2}).  Thus}
R_n^{\circ} &\xrightarrow{\phantom{xxx}p\phantom{xxx}} \Delta_n^2 \cdot c_n^2 \cdot d(V_n) \cdot \frac{|D_n|-1}{|D_n|} \cdot f(n) = \Theta_3 \cdot f(n)\\
M_n^{\circ} &\xrightarrow{\phantom{xxx}p\phantom{xxx}} \Delta_n^2 \cdot c_n^2 \cdot d(V_n) \cdot \frac{|D_n|-1}{|D_n|} = \Theta_3
\end{align*}
By Slutsky’s theorem \cite{slutsky1925uber}, it follows that
\begin{eqnarray}
\frac{n_3(S_n, F_n, D_n, \gamma)}{f(n)} = \frac{\frac{1}{f(n)}\cdot R_n^{\circ}}{M_n^{\circ}} &\xrightarrow{\phantom{xxx}d\phantom{xxx}} \frac{\plim_{n\rightarrow \infty} \frac{1}{f(n)}\cdot R_n^{\circ}}{\plim_{n\rightarrow \infty} M_n^{\circ}} = \frac{\Theta_3}{\Theta_3} = 1
\end{eqnarray}
\end{proof}

\subsection{Evaluating $n_3$ on Synthetic Networks}
\label{sec:eval-n3}

Prior to examining the performance of $n_3$ on empirical networks, we first look at its performance on the synthetic networks used to evaluate $n_1$ and $n_2$. The experiments shown in Figure \ref{results:n3} follow the framework described in Section \ref{sec:exp-framework} and use respondent-driven samples, each obtained via RDS process operating as specified in Assumption \ref{def:rds-assumptions}.  

The $12$ graphs in Figure \ref{results:n3} present the performance of the $n_3$ estimator as the true population size $n$ is varied from $5\cdot 10^3$ to $40\cdot 10^3$ (vertical axis of the grid) and the size of the RDS sample is varied from $250$ to $750$ (horizontal axis of the grid).  In each of the $12$ graphs, the x-axis varies the average degree $\lambda$ from $3$ to $10$.  For each choice of $\lambda$, the medians and quartile ranges of $n_3$ are given for each of the $5$ graph families.  Each of these is determined by $900$ simulations ($30$ graphs times $30$ uniformly drawn samples in each graph).

\begin{figure}[tbp]
\setlength{\belowcaptionskip}{12pt}
\centering
\begin{tabular}{ScScScSc}
\subcaptionbox{$n$ = $5\cdot 10^3$, $r=250$\label{1a}}{\includegraphics[width = 0.3\linewidth]{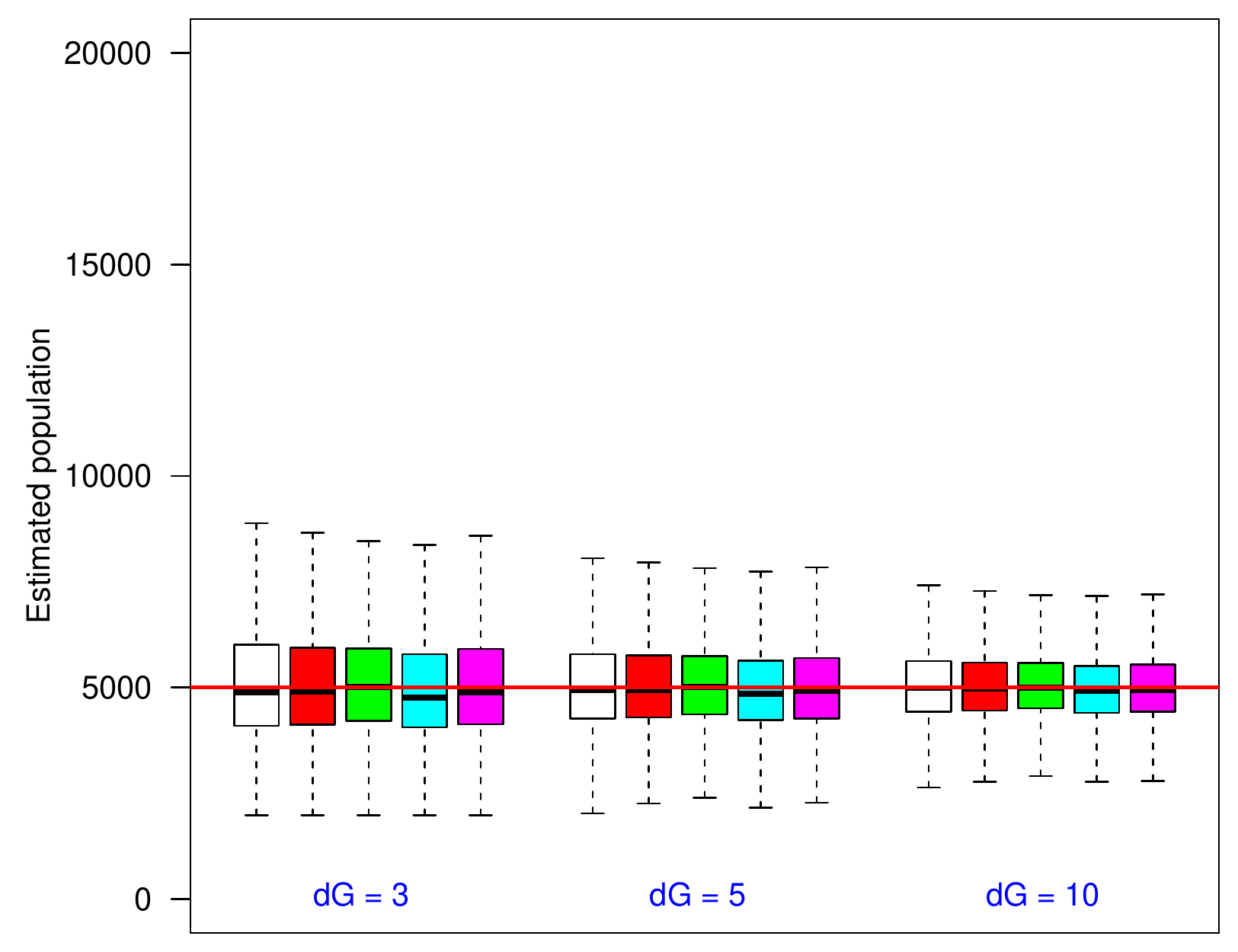}} &
\subcaptionbox{$n$ = $5\cdot 10^3$, $r=500$\label{2a}}{\includegraphics[width = 0.3\linewidth]{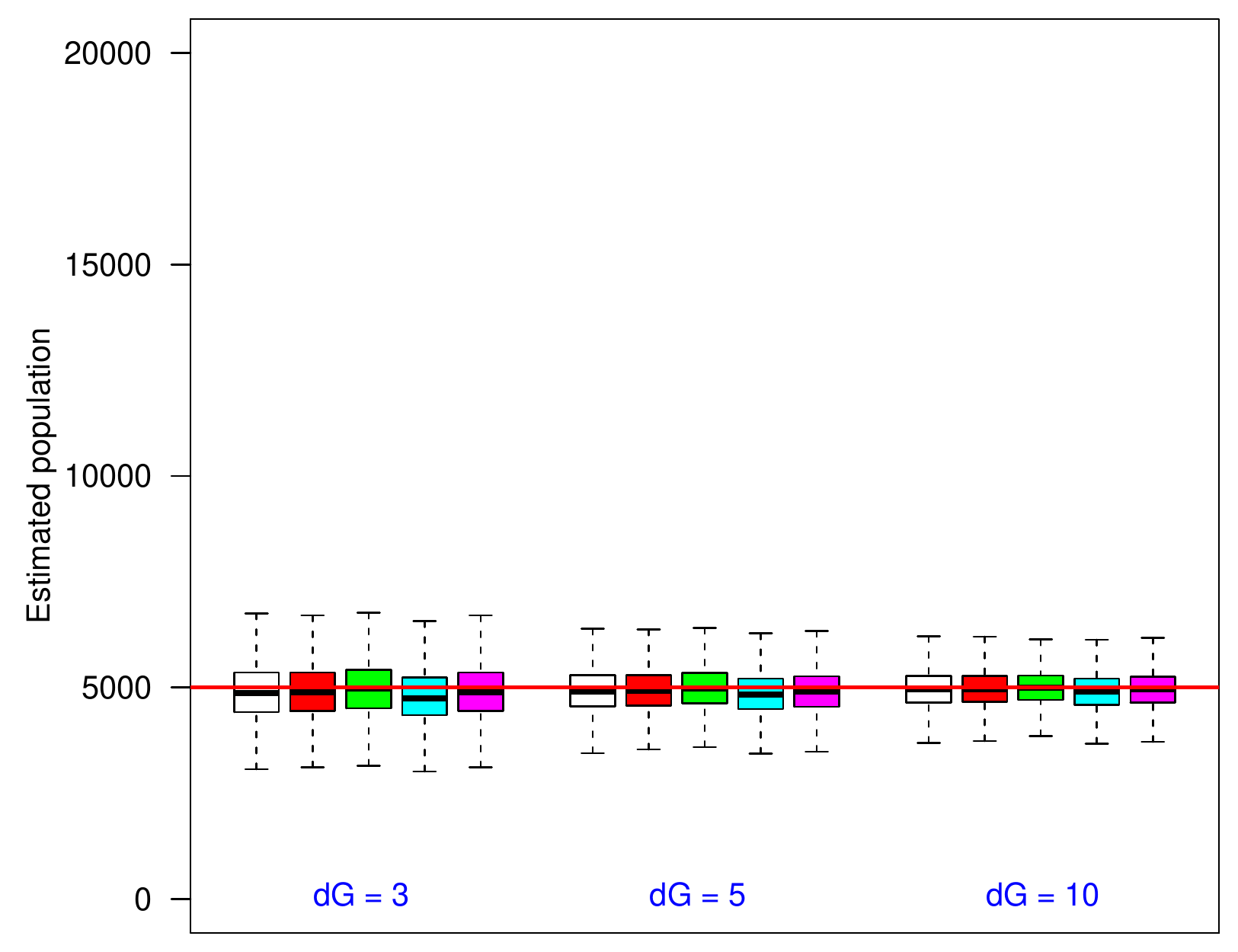}} &
\subcaptionbox{$n$ = $5\cdot 10^3$, $r=750$\label{3a}}{\includegraphics[width = 0.3\linewidth]{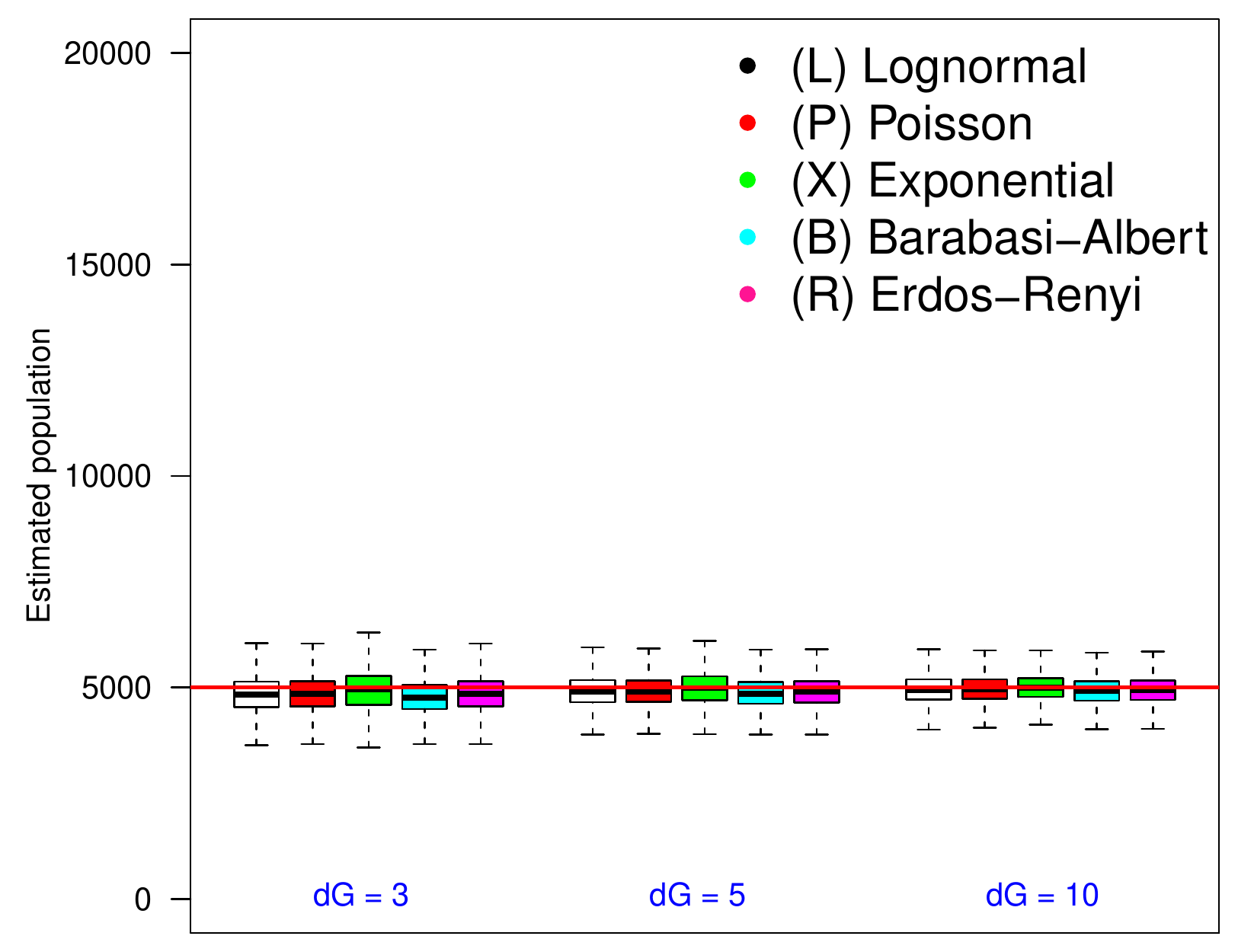}}\\[0pt]
\subcaptionbox{$n$ = $10\cdot 10^3$, $r=250$\label{1b}}{\includegraphics[width = 0.3\linewidth]{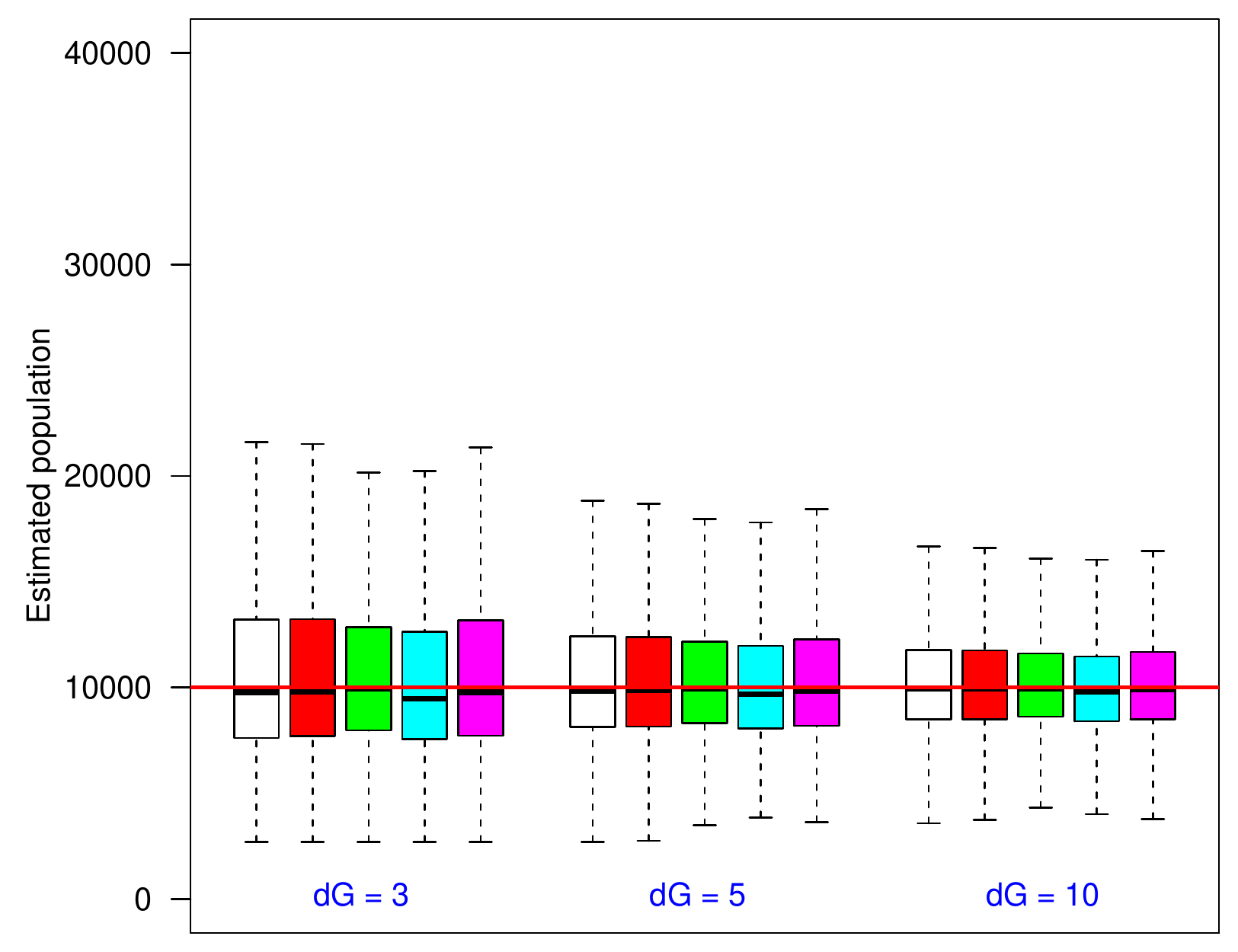}} &
\subcaptionbox{$n$ = $10\cdot 10^3$, $r=500$\label{2b}}{\includegraphics[width = 0.3\linewidth]{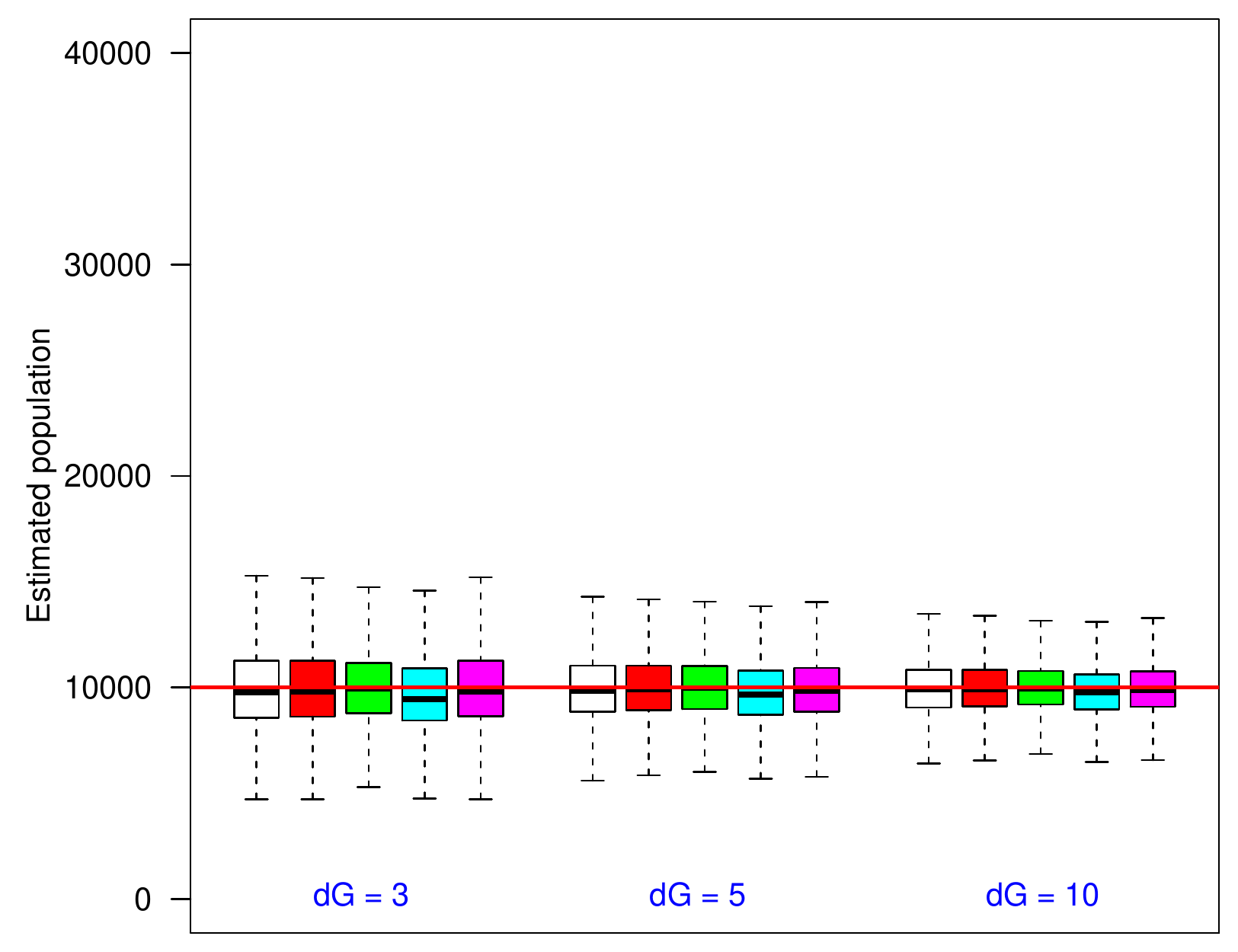}} &
\subcaptionbox{$n$ = $10\cdot 10^3$, $r=750$\label{3b}}{\includegraphics[width = 0.3\linewidth]{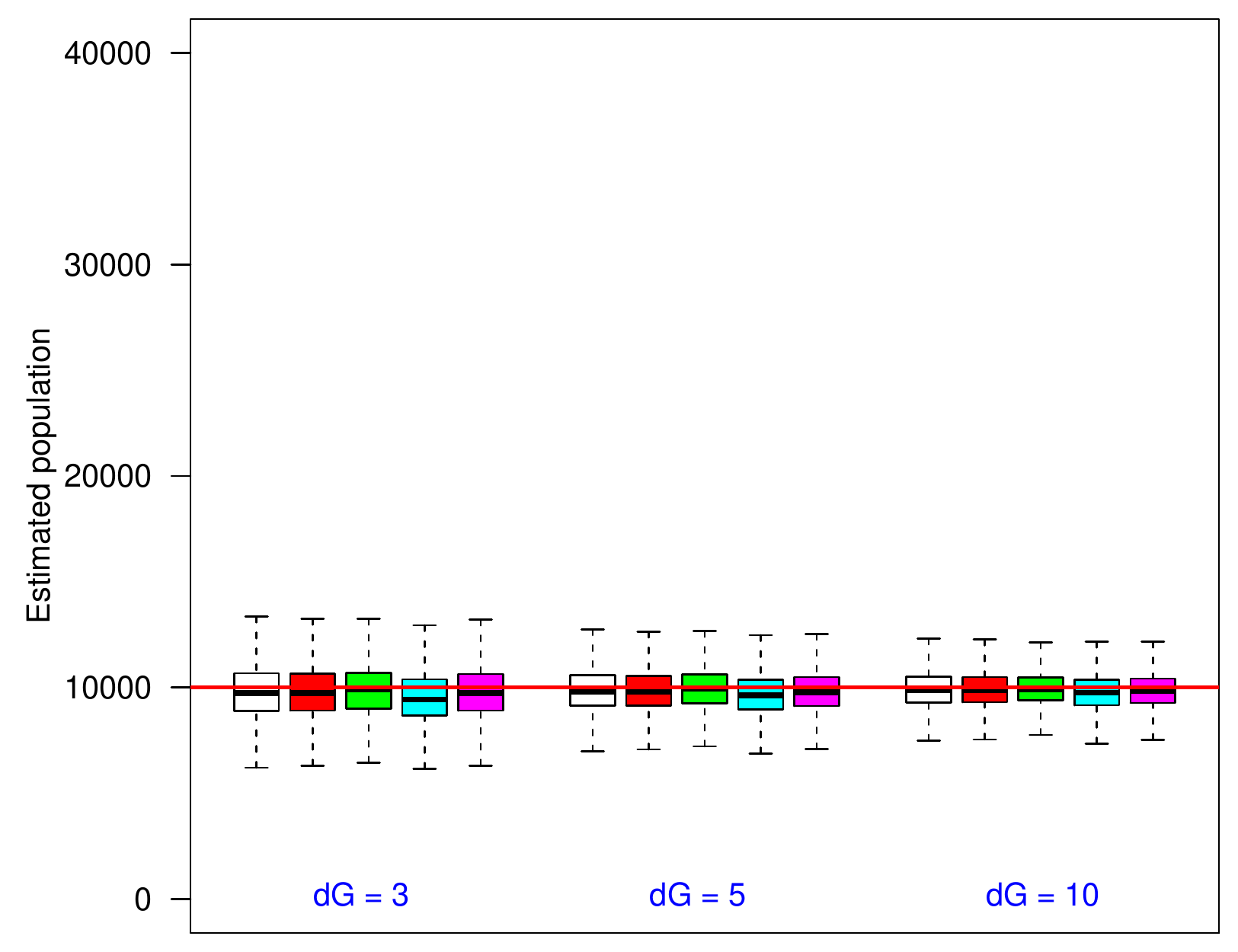}}\\
\subcaptionbox{$n$ = $20\cdot 10^3$, $r=250$\label{1c}}{\includegraphics[width = 0.3\linewidth]{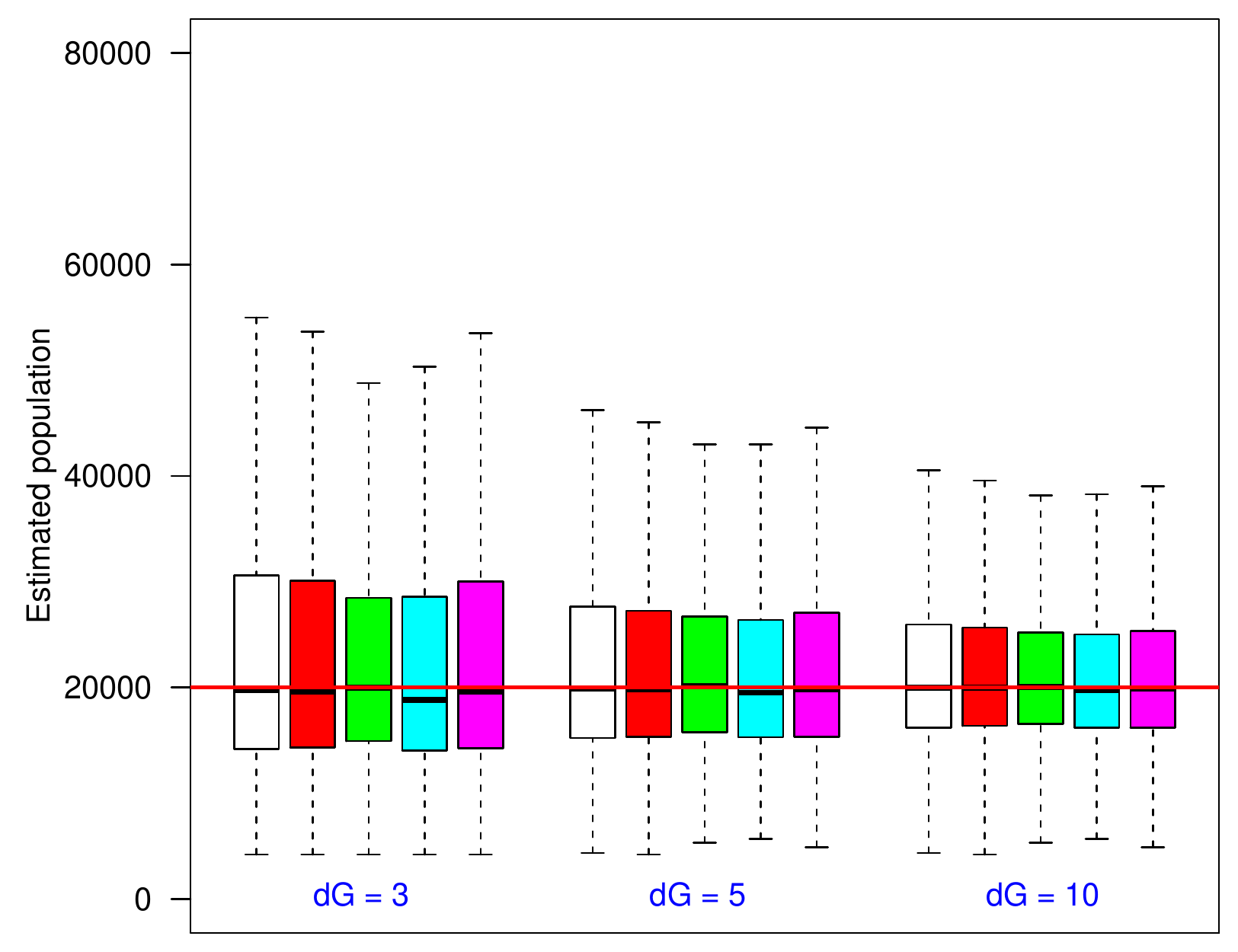}} &
\subcaptionbox{$n$ = $20\cdot 10^3$, $r=500$\label{2c}}{\includegraphics[width = 0.3\linewidth]{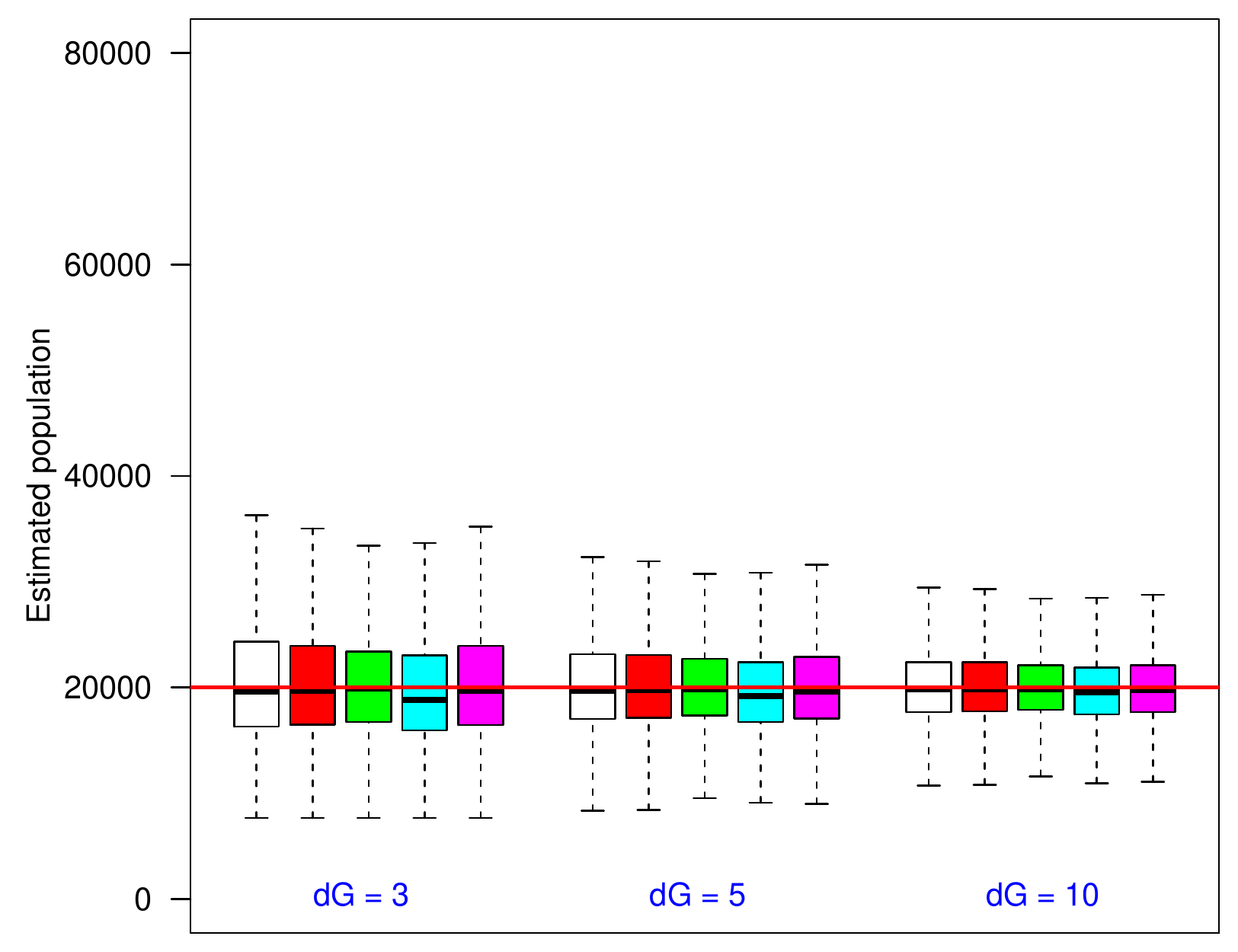}} &
\subcaptionbox{$n$ = $20\cdot 10^3$, $r=750$\label{3c}}{\includegraphics[width = 0.3\linewidth]{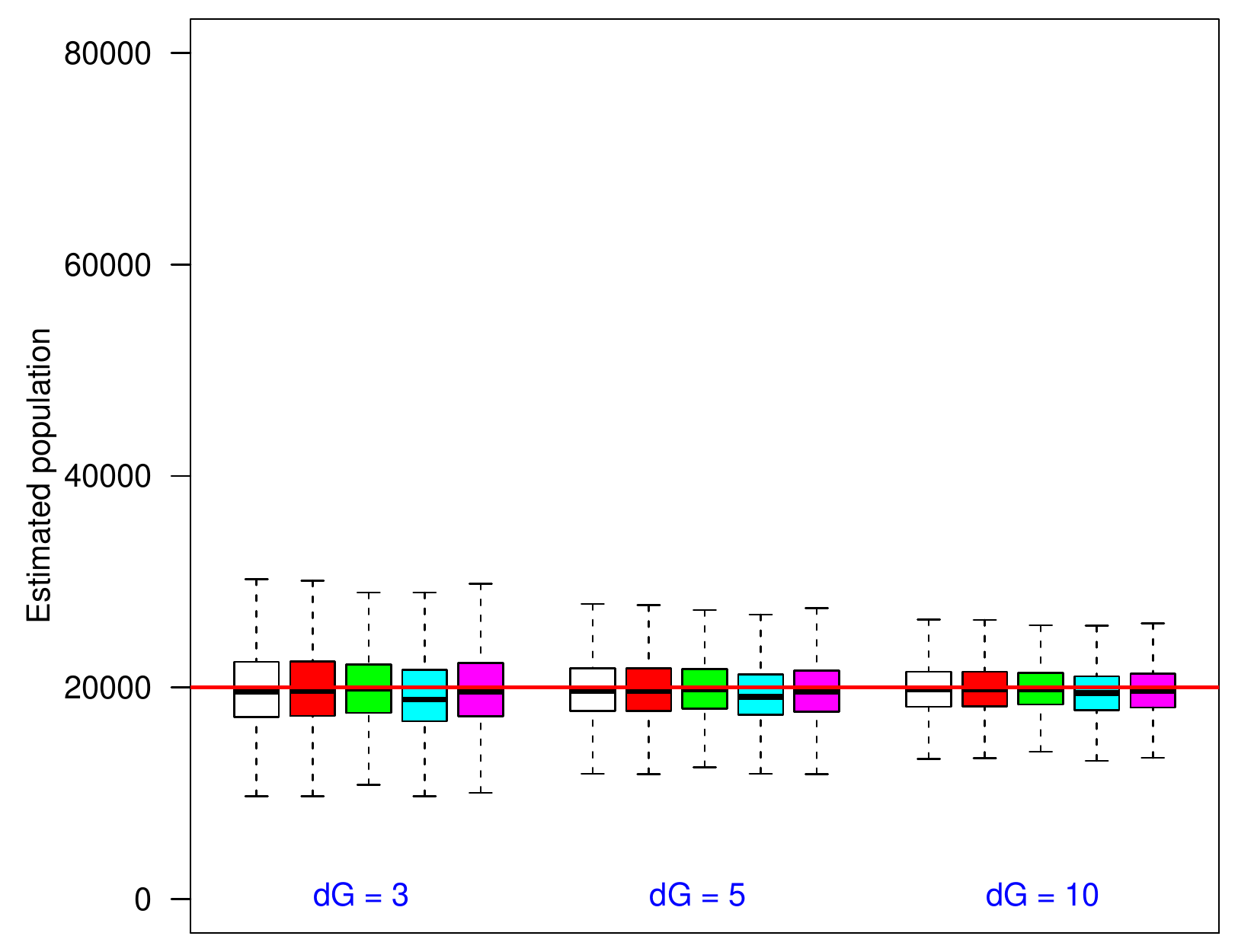}}\\[0pt]
\subcaptionbox{$n$ = $40\cdot 10^3$, $r=250$\label{1d}}{\includegraphics[width = 0.3\linewidth]{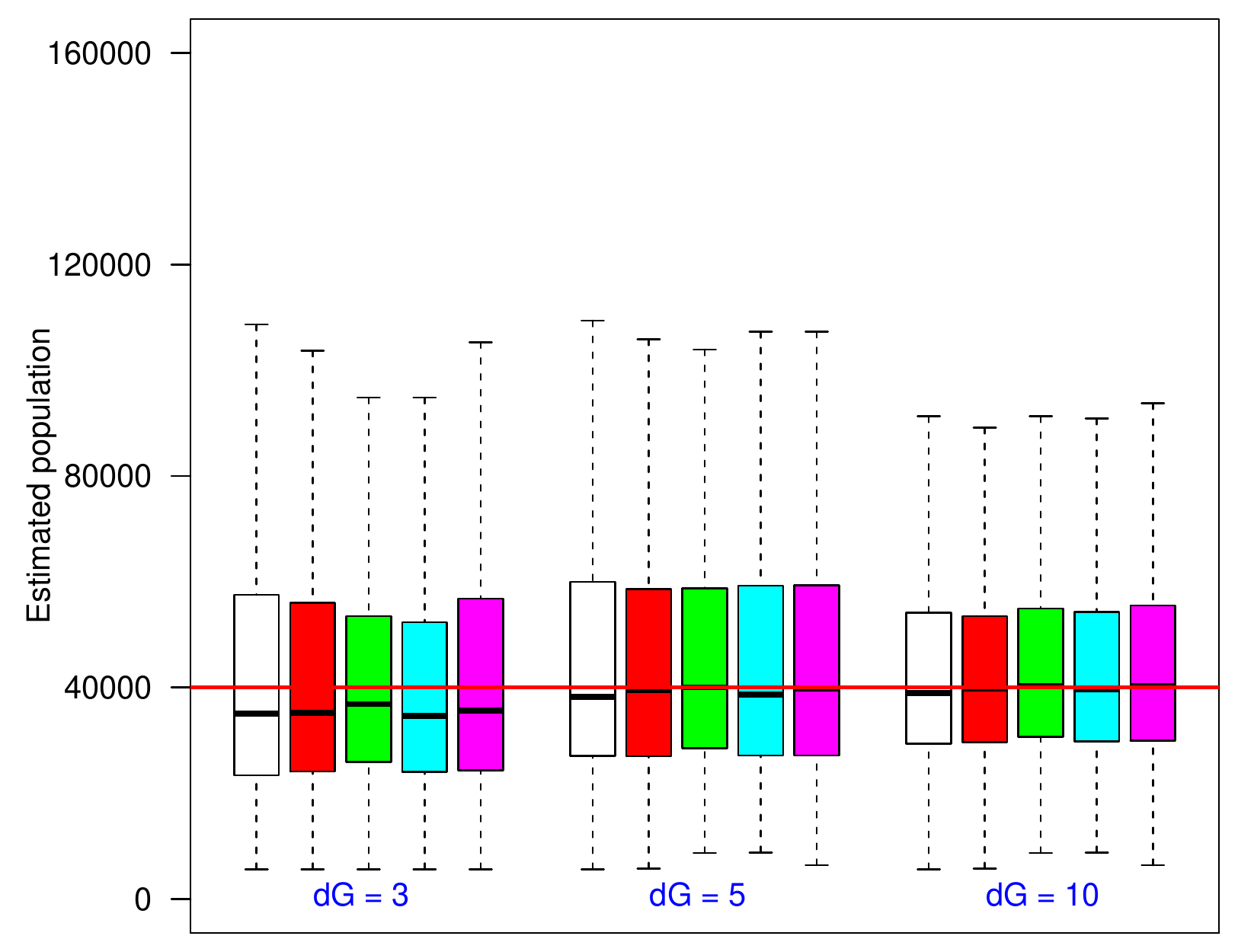}} &
\subcaptionbox{$n$ = $40\cdot 10^3$, $r=500$\label{2d}}{\includegraphics[width = 0.3\linewidth]{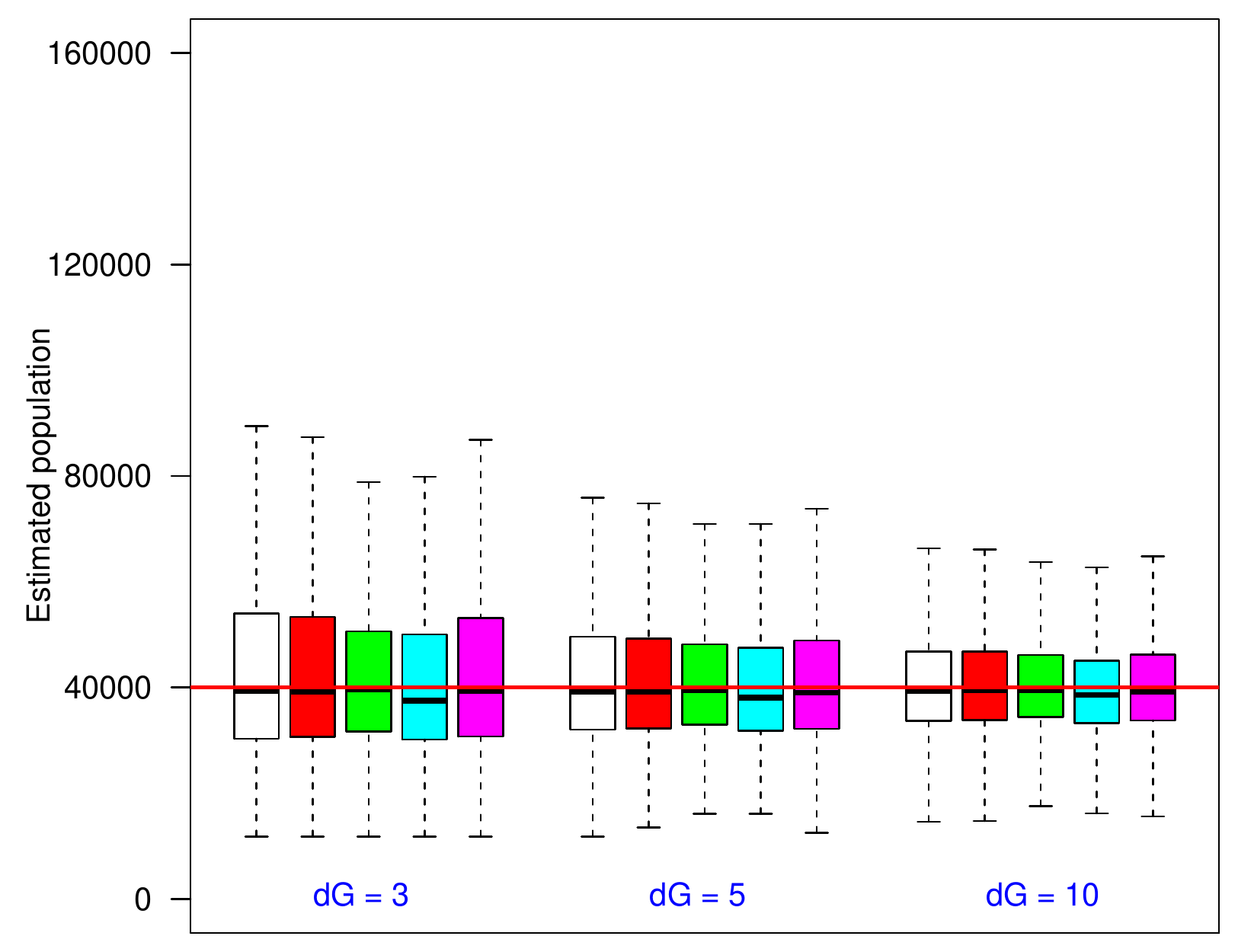}} &
\subcaptionbox{$n$ = $40\cdot 10^3$, $r=750$\label{3d}}{\includegraphics[width = 0.3\linewidth]{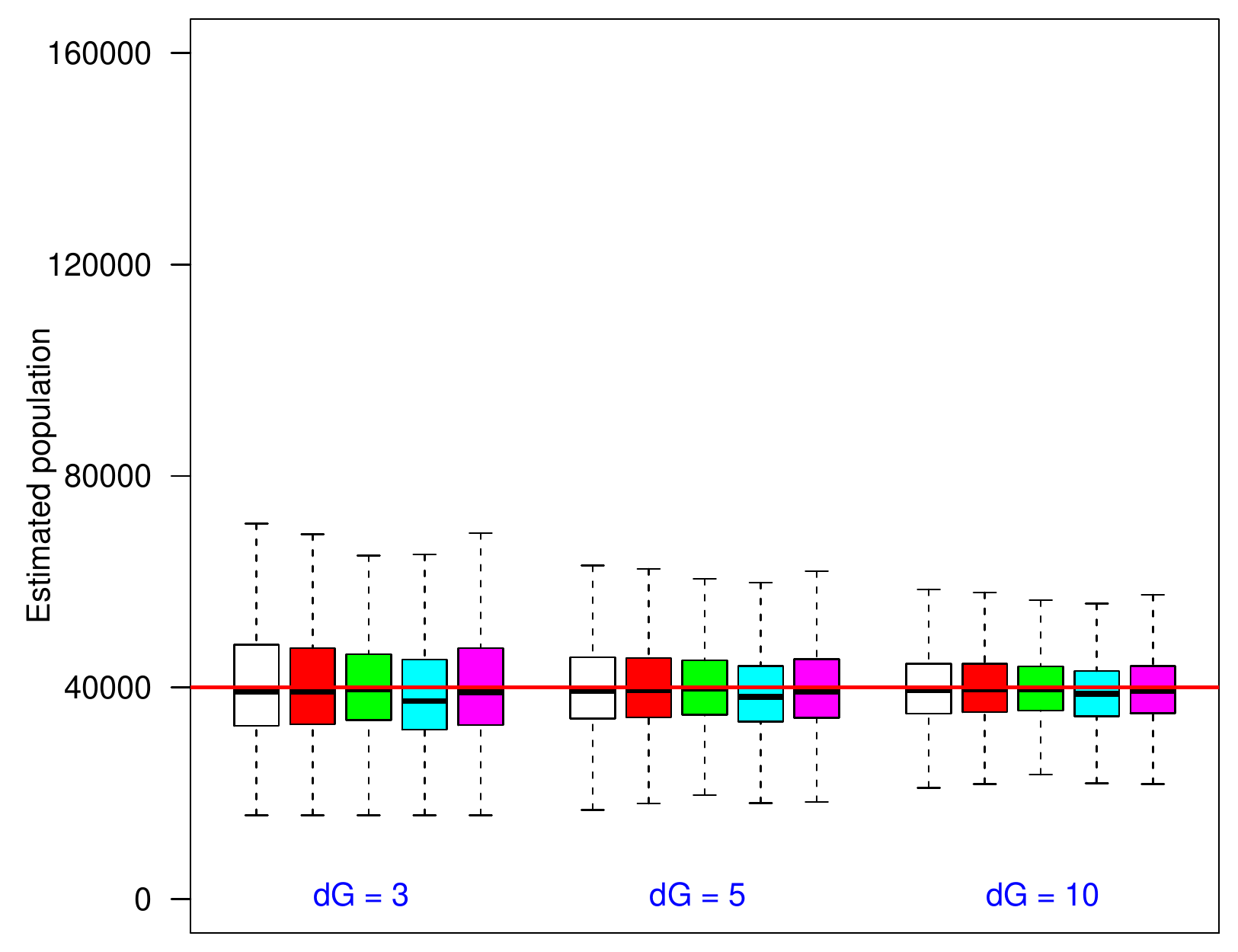}}
\end{tabular}
\setlength{\abovecaptionskip}{1cm}
\caption{Estimator $n_3$ on RDS samples in populations of size $n=5\cdot 10^3$ to $40\cdot 10^3$. In each box, the thick line indicates the sample median; the top of the box is the median of the upper half of the estimated values (75\% quartile); the bottom of the box indicates the median of the lower half of the estimated values (25\% quartile; and the whiskers indicate the full range of estimated values. No (finite) outliers were removed.}
\label{results:n3}
\end{figure}

Figure \ref{results:n3} shows that the median of $n_3$ converge to the true population size, much like the performance of $n_2$. In all the networks, the medians of $n_3$ are all very close to the their true network populations, regardless the sample sizes, population sizes and types of graph. In addition, Figure \ref{results:n3} shows that as sample size increases, the interquartile ranges of the estimates decrease. For example, when $n=5\cdot 10^3$ and $r=250$, Lognormal degree distribution graphs with $\lambda=3$ experience a interquartile range of 1915 in their $n_3$ estimates (39.1\% of the median).  In comparison, when $r=750$, the interquartile range for this family of graphs decreases to 604 (a 68.5\% reduction).  The magnitude of this effect decreases as networks grow larger.  For example for a network of size $n=40\cdot 10^3$, increasing the sample size from $r=250$ to $r=750$ causes the interquartile range of the $n_3$ estimate to undergo a (still sizable) 55.0\% decrease.


\section{Estimating Population Size while Ensuring Anonymity}
\label{sec:anonymity}

Significant obstacles arise in the direct application of estimators $n_1, n_2, n_3$ (see (\ref{eq:n1}), (\ref{eq:n2}), and (\ref{eq:n3}), respectively).  In many circumstances where RDS is used, researchers are often required to measure the sizes of stigmatized networked populations (e.g. people who inject drugs, sex workers, individuals engaged in specific types of illegal activity, etc.) and within social communities that naturally seek to remain ``unidentified''.  In these circumstances, the membership of sets $S$ and $R(S,F)$ is often not explicitly knowable because individuals are reluctant to unambiguously identify themselves or their social network peers.  

To formalize and accommodate notions of privacy required under such circumstances within the estimation procedures described above, we assume that each individual in $V = \{v_1, v_2, \ldots, v_{|V|}\}$  has a unique ID; for simplicity we take the ID of $v_i\in V$ to be the integer $i$ (for $i=1, \ldots, |V|$).   
Towards ensuring anonymity, we imagine a {\em hashing} \cite{CARTER1979} function $\psi:V\rightarrow \Omega$ that assigns each individual's ID to a code in $\Omega$.  We thus follow the general framework of Privatized Network Sampling (PNS) design \cite{fellowsThesis}, mimicking the hash functions of telefunken-type \cite{TELEFUNKEN2012}.  

By taking $\psi$ to be a random (not necessarily 1-to-1) function that is difficult to invert, subjects are convinced that disclosing the hash code of an individual does not unambiguously identify the individual themselves, and so preserves their privacy.

\begin{assumption}
\label{def-hash-assumptions}
Suppose $V$ is a set of individuals obtained via RDS referral tree $F$.  While each $v_i\in V$ is unwilling to disclose their own ID $i$, and is secretive about the IDs of their peers $\{ j \;|\; v_j\in N(v_i, \emptyset)\}$, they are readily willing to reveal (a) the own hash code $\psi(v_i)$; (b) the (multiset of) hash codes of their peers (outside the referral tree $F$):
\begin{eqnarray}
\label{def:N-hash}
N_u^\psi(S, F) &\coloneqq& \coprod\limits_{\substack{v \in N(u)\\ (u,v) \not\in F}} \{ \psi(v) \} \subseteq \Omega
  \end{eqnarray}
and (c) their own network size $d(v_i)=\langle N_u^\psi(S, F) \rangle$, excluding the referral tree $F$.
\end{assumption}

\begin{assumption}
\label{def-hash-assumptions}
To simplify our analysis, throughout what follows, we will assume $\psi$ is a function chosen uniformly at random from the space of all functions from $V\rightarrow \Omega$.  We will refer to such a $\psi$ as a ``random hash function'' from $V$ to $\Omega$.  The action of $\psi$ on the $V$ is illustrated in Figure \ref{fig:psi-concept}.
\end{assumption}

\begin{figure}[tbp]
\setlength{\belowcaptionskip}{12pt}
\centering
\begin{tabular}{Sc}
\includegraphics[width = 3.0in]{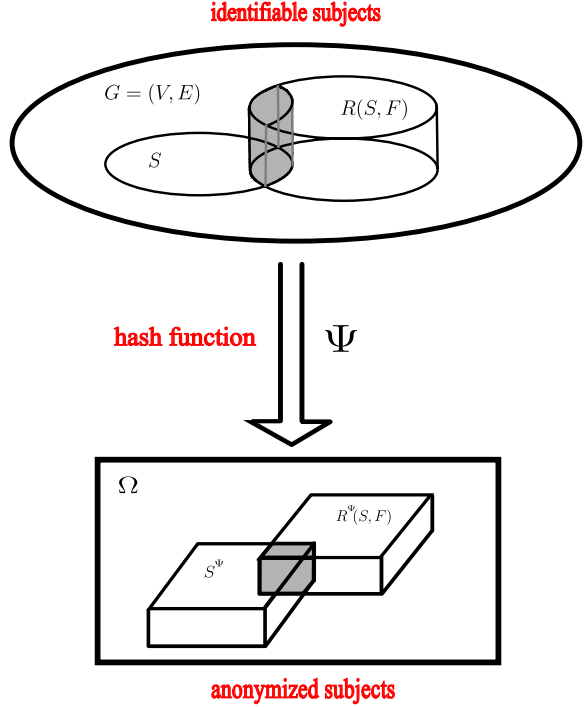}
\end{tabular}
\caption{The action of $\psi$ on $V$}
\label{fig:psi-concept}
\end{figure}

In practice, $\psi(v)$ might be an obtained by amalgamating a well-defined tuple of characteristics of $v$ which are known to $v$'s friends (e.g. $v$'s gender, phone number, hair color, approximate age, racial category, etc.) and then encoding this using a cryptographic function.  A related coding technique was used in our earlier work on estimating the size of the methamphetamine using population in New York City, where it was referred to as the ``telefunken'' code \cite{TELEFUNKEN2012}.  

\subsection{Revised Estimators for use with Privatized Network Sampling (PNS) design}
\label{sec:n2-n3-psi}

We begin by ``lifting'' the terms introduced in the earlier Definition \ref{def:basics}, to the PNS framework \cite{fellowsThesis}.

\begin{definition}
\label{def:basics-hashing}
Let $G=(V,E)$ be a graph, and $\psi:V\rightarrow \Omega$ a random hash function.  Let $H=(S,F)$ be a subgraph on $S\subseteq V$ with edge set $F \subseteq E \cap (S \times S)$.  The (multiset of) hash codes of the subjects is
   \begin{align}
   S^\psi &\coloneqq \{ \psi(v) \;|\; v \in S \} \subseteq \Omega \label{def:S-hash}
\intertext{The {\em $\psi$-free ends} of $S$ (in $G$ modulo $H$) are taken to be the disjoint union (multiset)}
R^\psi(S,F) &\coloneqq \coprod_{u\in S} N^\psi(u,F) \subseteq \Omega \label{def:R-hash}
\intertext{and the {\em $\psi$-matches} of (in $G$ modulo $H$) are taken to be the disjoint union (multiset)}
M^\psi(S,F) &\coloneqq \coprod_{u\in S} \left( N^\psi(u,F) \cap S^\psi \right) \subseteq \Omega. \label{def:M-hash}
\intertext{We denote their respective multiset cardinalities as}
 \langle R^\psi(S,F) \rangle &\coloneqq \sum_{u\in S} \left| N^\psi(u,F) \right| \nonumber\\
 \langle M^\psi(S,F) \rangle &\coloneqq \sum_{u\in S} \left| N^\psi(u,F) \cap S^\psi \right| \nonumber.
  \end{align}
  The reader may wish to compare expressions (\ref{def:N-hash}), (\ref{def:R-hash}), and (\ref{def:M-hash}) with the non-hashed analogues in Definition \ref{def:basics}'s expressions (\ref{def:N}), (\ref{def:R}), and (\ref{def:M}).
\end{definition}

The next Lemma is foundational and justifies the proposed revised estimates $n_1^\psi$, $n_2^\psi$, and $n_3^\psi$, which will be presented subsequently.

\begin{lemma}
Let $G=(V,E)$ a graph with $|V|=n'$, sampled from the space of all $n'$-vertex graphs by configuration sampling with respect to degree distribution ${\cal D}$.  Let $S\subseteq V$ be an RDS sample collected as a subgraph $H=(S,F)$ be  with edge set $F \subseteq E \cap (S \times S)$.   Let $c \coloneqq |S|/|V|$, where $c\ll 1$. Accepting Assumption~\ref{assumption-harmonic}, take $\psi:V\rightarrow \Omega$ to be a random hash function.
\begin{enumerate}
\item Suppose $u \in S$ reports its own code $x\coloneqq \psi(u)$, the code $y\coloneqq \psi(v)$ of one of its neighbors $v \in N_u(S,F)$.  If $w \in \psi^{-1}(y) \cap S$ is selected uniformly at random, and $w$ has degree $d(w)$, then
\begin{align*}
Prob(w=v) &= \frac{1}{\frac{n'-1}{|\Omega|}\frac{\widetilde{d}(S)}{(d(w) - 1)}+1} 
\end{align*}
\item  For each  code $y \in \Omega$, over the space of all random hash functions,
$$ E[\langle M^\psi(S,F) \rangle] = \hat{m}(y,n')$$
where
\begin{align*}
\hat{m}(y,n') &\coloneqq \sum_{w \in \psi^{-1}(y) \cap S} \frac{1}{\frac{n'-1}{|\Omega|}\frac{\widetilde{d}(S)}{(d(w) - 1)}+1} \\
\hat{m}(n') &\coloneqq \sum_{y \in M^\psi(S,F)} \hat{m}(y,n')
  \end{align*}
\end{enumerate}
\end{lemma}
\begin{proof}
(1) Because $\psi$ is a random function, for any $z\in \Omega$  
\begin{align*}
E[ |\psi^{-1}(z)| ] &= \frac{n'}{|\Omega|}. 
\end{align*}
The expected total number of free ends incident to some vertex in the set $\psi^{-1}(y) \backslash \{w\}$ is 
$$
\frac{(n'-1)(1-c)}{|\Omega|} \cdot \widetilde{d}(S) + \frac{(n'-1)c}{|\Omega|} \cdot \left(\widetilde{d}(S) - 1\right)
$$
and since $w\in S$, the expected number of free ends incident to $w$ is $d(w) - 1$.
So 
\begin{align*}
Prob(w=v) &= \frac{d(w)-1}{\frac{(n'-1)(1-c)}{|\Omega|} \cdot \widetilde{d}(S) + \frac{(n'-1)c}{|\Omega|} \cdot \left(\widetilde{d}(S) - 1\right) + (d(w)-1)}
\end{align*}
dividing through by $d(w)-1$, and considering $c \sim 0$, the Lemma is proved.
Assertion (2) follows from (1) by linearity of expectation.
\end{proof}

\begin{definition}
\label{def:n2-hashing}
Given a graph $G=(V,E)$, and $\psi:V\rightarrow \Omega$ a random hash function.  Fix $S \subseteq V$, and $H=(S,F)$ a subgraph on $S\subseteq V$ with edge set $F \subseteq E \cap (S \times S)$.  We define
\begin{align}
\label{eq:n2-hashing}
n_2^\psi(S,F) &\coloneqq RootOf\left[ f_2^\psi(n',S,F) - n' = 0,\;\; n' \right]
\intertext{where}
f_2^\psi(n',S,F) &\coloneqq \frac{\frac{d(S)-1}{\widetilde{d}(S)} \cdot \langle S^\psi \rangle\cdot \langle R^\psi(S,F) \rangle }{\hat{m}(n')} \nonumber
\end{align}
\end{definition}

\begin{definition}
\label{def:cross-seeds-hashing}
Given a graph $G=(V,E)$, a set $S \subseteq V$, and $H=(S,F)$ a subgraph on $S\subseteq V$ with edge set $F \subseteq E \cap (S \times S)$.  Let $D\subseteq S$ satisfying $|D|>1$ and
$$
s_1 \neq s_2 \implies C_\gamma(s_1) \cap C_\gamma(s_2) = \emptyset.
$$ 
Take $\gamma: S\rightarrow D$ as described in Definition \ref{def:cross-seeds}.  The (multiset of) hash codes of vertices in the component of $u$ are denoted
\begin{align}
C_\gamma^\psi(u) &\coloneqq \{ \psi(v) \;|\; v \in C_\gamma(u) \} \subseteq S^\psi \label{def:C-hash}
\intertext{while the codes of the complement set (inside $S$) is written as} 
\widetilde{C}_\gamma^\psi(u) &\coloneqq \{ \psi(v) \;|\; v \in \widetilde{C}_\gamma(u) \} \subseteq S^\psi. \nonumber
\intertext{Note that $C_\gamma^\psi(u) \cap \widetilde{C}_\gamma^\psi(u)$ may be non-empty. For each seed $s\in D$, we define the {\em cross-seed $\psi$-matches} from $C_\gamma^\psi(s)$ in $G$ modulo $H$ as the disjoint union (multiset)}
X^\psi(s, F, \gamma) &\coloneqq \coprod_{u\in C_\gamma(s)} \left( N^\psi(u,F) \cap \widetilde{C}_\gamma^\psi(s) \right) \subseteq \Omega. \label{def:X-hash}\\
\intertext{The reader may wish to compare expressions (\ref{def:C-hash}) and (\ref{def:X-hash}) with the non-hashed analogues in Definition \ref{def:cross-seeds}'s expressions (\ref{def:C}) and (\ref{def:X}).  We also define}
\tilde{x}(y,s,\gamma,n') &\coloneqq \sum_{w \in \psi^{-1}(y) \cap \widetilde{C}_\gamma(s)} \frac{1}{\frac{n'-1}{|\Omega|}\frac{\widetilde{d}(S)}{(d(w) - 1)}+1} \nonumber \\
\hat{x}(s,F,\gamma,n') &\coloneqq \sum_{y \in X^\psi(s,F,\gamma)} \tilde{x}(y,s,\gamma,n') \nonumber
  \end{align}
\end{definition}

\begin{definition}
\label{def:n3-psi}
Given a graph $G=(V,E)$, a set $S \subseteq V$, and $H=(S,F)$ a subgraph on $S\subseteq V$ with edge set $F \subseteq E \cap (S \times S)$.  We define
\begin{align}
\label{eq:n3-hashing}
n_3^\psi(S,F) &\coloneqq RootOf\left[ f_3^\psi(n',S,F, D, \gamma) - n' = 0,\;\; n' \right]
\intertext{where}
f_3^\psi(n', S, F, D, \gamma) &\coloneqq \frac{\sum_{s\in D} \frac{d(\widetilde{C}_\gamma(s))-1}{\widetilde{d}(S)} \cdot \langle \widetilde{C}_\gamma^\psi(s) \rangle \cdot \langle R^\psi(C_\gamma(s), F) \rangle}{\sum_{s\in D} \hat{x}(s,F,\gamma,n')} \nonumber
\end{align}
\end{definition}

\subsection{Evaluating $n^{\psi}_2$ on Synthetic Networks}
\label{sec:eval-n2-psi}

The experiments discussed here follow the framework used in prior experiments described above. Samples are derived using the RDS process operating as specified in Assumption \ref{def:rds-assumptions}.  The hash space size used for the encoding of each agent's identity was varied from $|\Omega|=2\cdot 10^3$ to $256 \cdot 10^3$.

The $12$ graphs in Figure \ref{results:n2_psi} present the performance of the $n^{\psi}_2$ estimator as the true population size $n$ is varied from $5\cdot 10^3$ to $40\cdot 10^3$ (vertical axis of the grid), the sample size is fixed to $r = 500$ and the hash space size was varied from $|\Omega|=2\cdot 10^3$ to $256 \cdot 10^3$ (horizontal axis of the grid).  In each of the $12$ graphs, the x-axis varies the average degree $\lambda$ from $3$ to $10$.  For each choice of $\lambda$, the medians and quartile ranges of $n^{\psi}_2$ are given for each of the $5$ graph families.  Each of these is determined by $900$ simulations ($30$ graphs times $30$ uniformly drawn samples in each graph).

\begin{figure}[p]
\setlength{\belowcaptionskip}{12pt}
\centering
{\small 
\begin{tabular}{ScScScSc}
\subcaptionbox{$n$ = $5\cdot 10^3$, $|\Omega| = 2\cdot 10^3$\label{1a}}{\includegraphics[width = 0.3\linewidth]{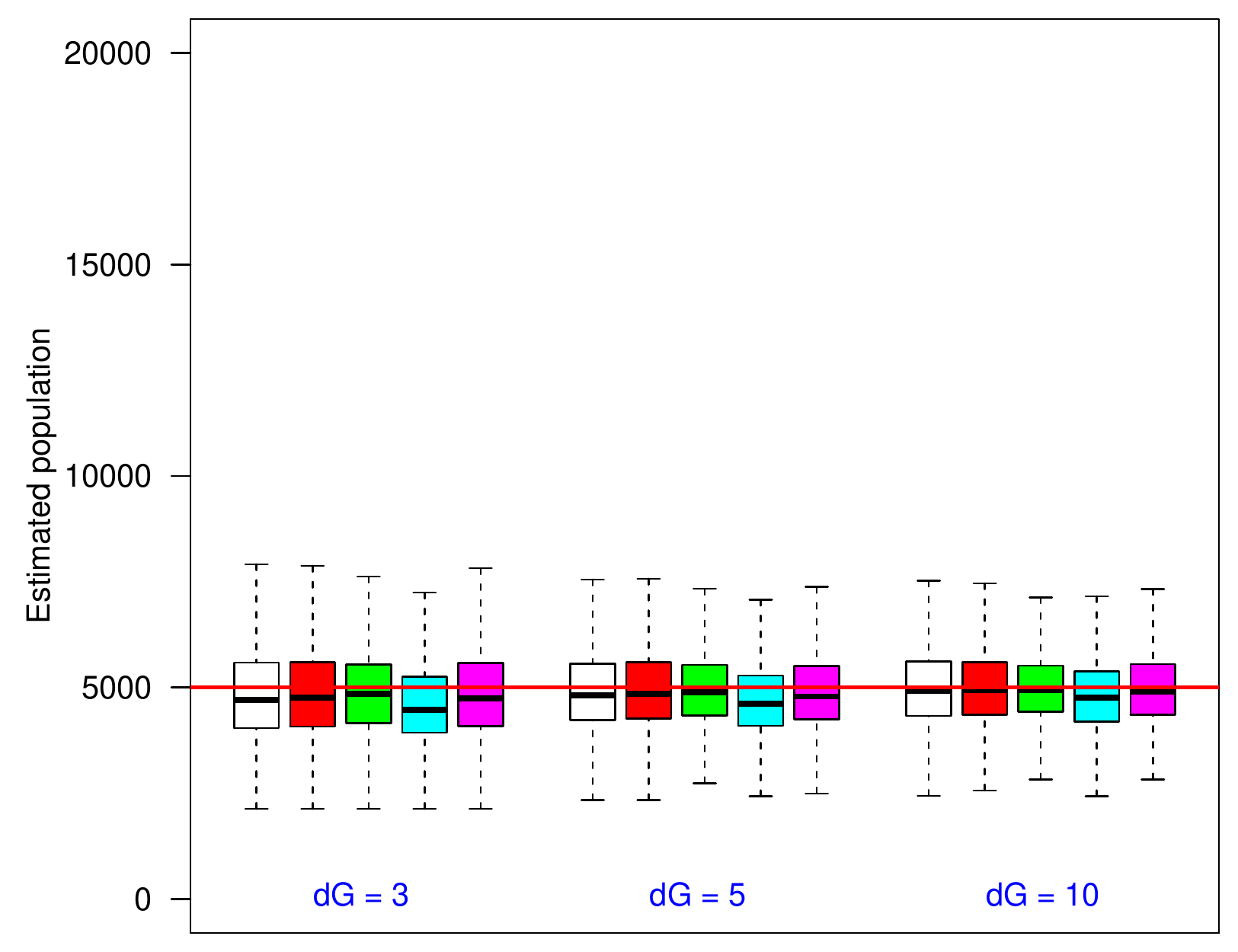}} &
\subcaptionbox{$n$ = $5\cdot 10^3$, $|\Omega| = 32\cdot 10^3$\label{2a}}{\includegraphics[width = 0.3\linewidth]{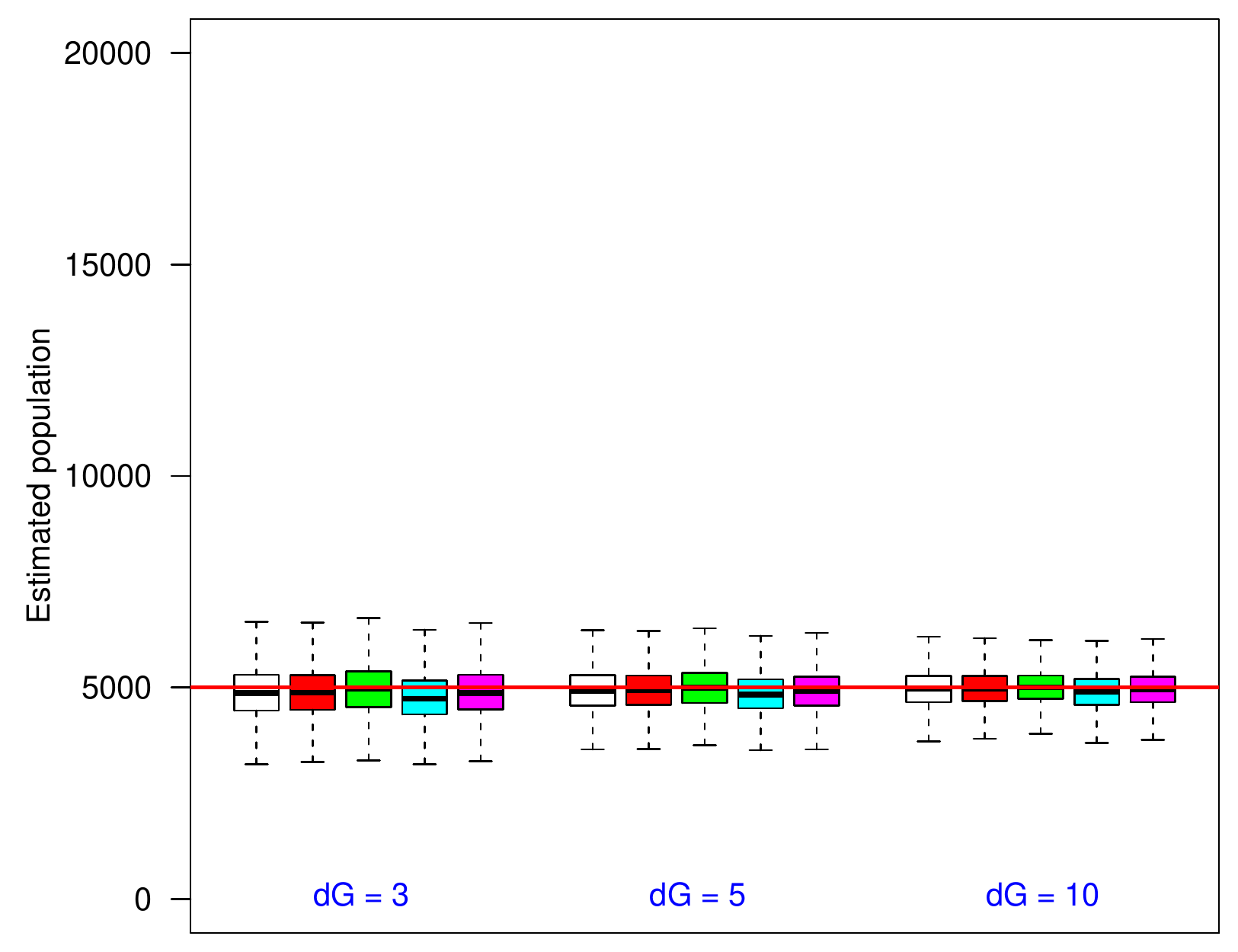}} &
\subcaptionbox{$n$ = $5\cdot 10^3$, $H = 256 \cdot 10^3$\label{3a}}{\includegraphics[width = 0.3\linewidth]{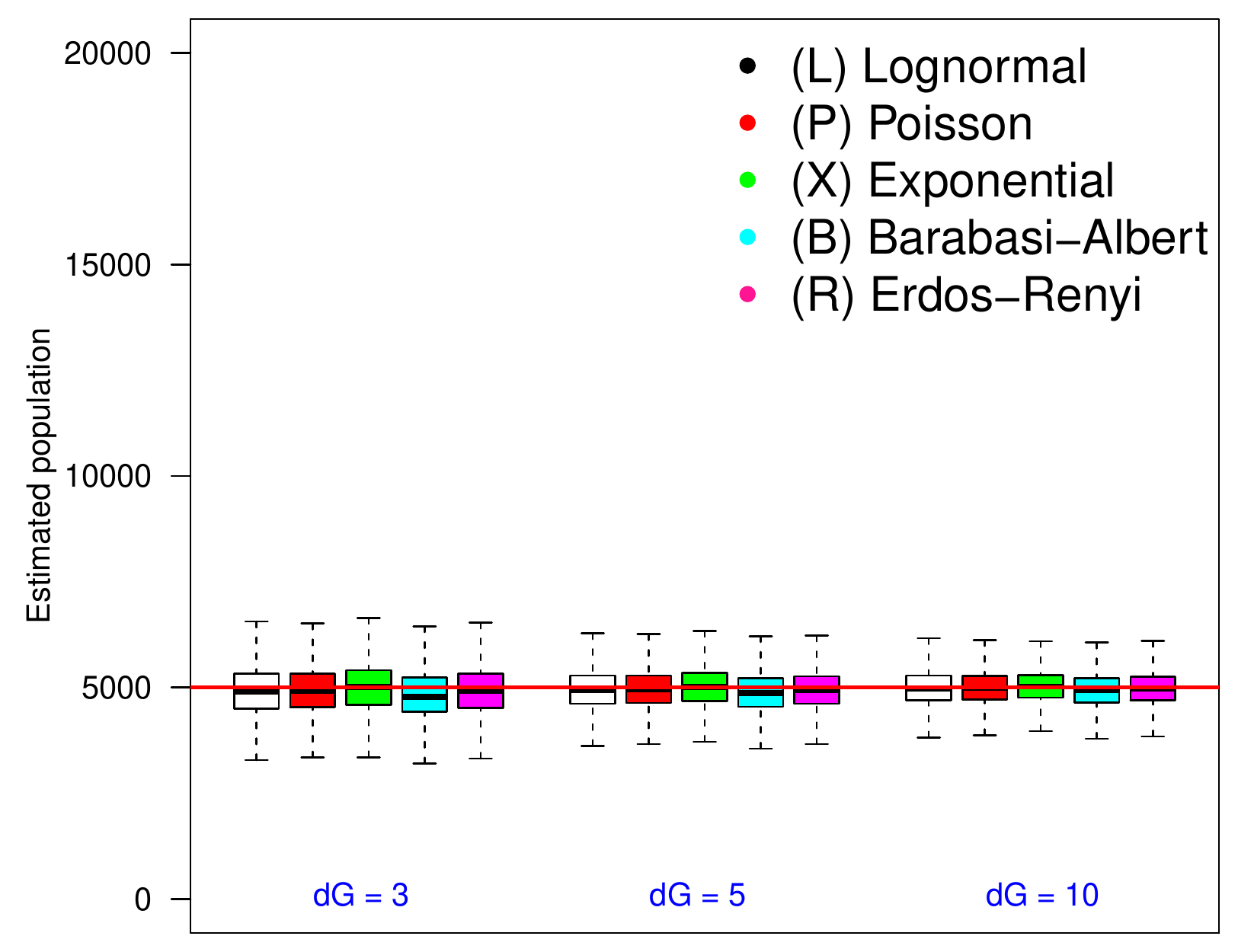}}\\
\subcaptionbox{$n$ = $10\cdot 10^3$, $|\Omega| = 2\cdot 10^3$\label{1b}}{\includegraphics[width = 0.3\linewidth]{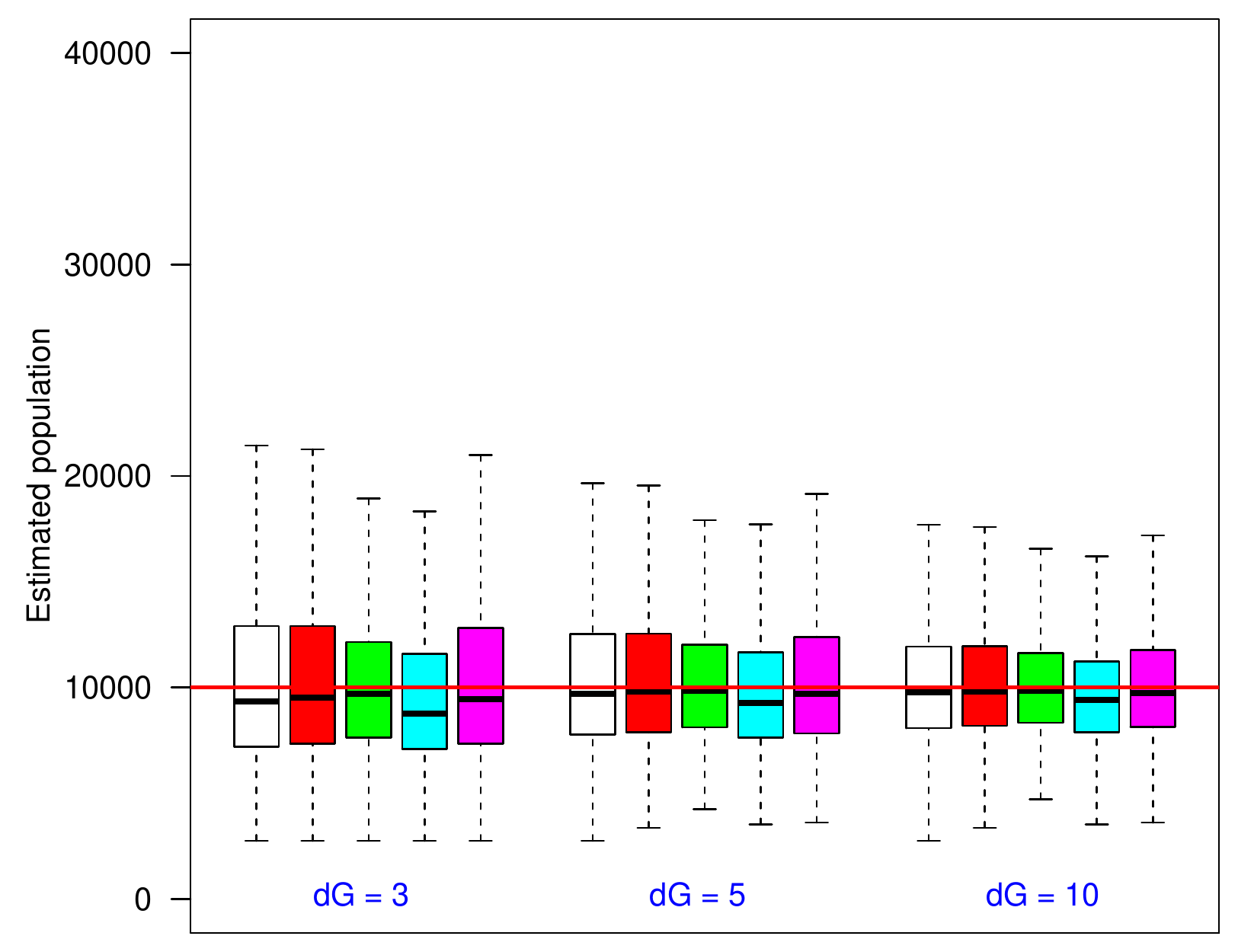}} &
\subcaptionbox{$n$ = $10\cdot 10^3$, $|\Omega| = 32\cdot 10^3$\label{2b}}{\includegraphics[width = 0.3\linewidth]{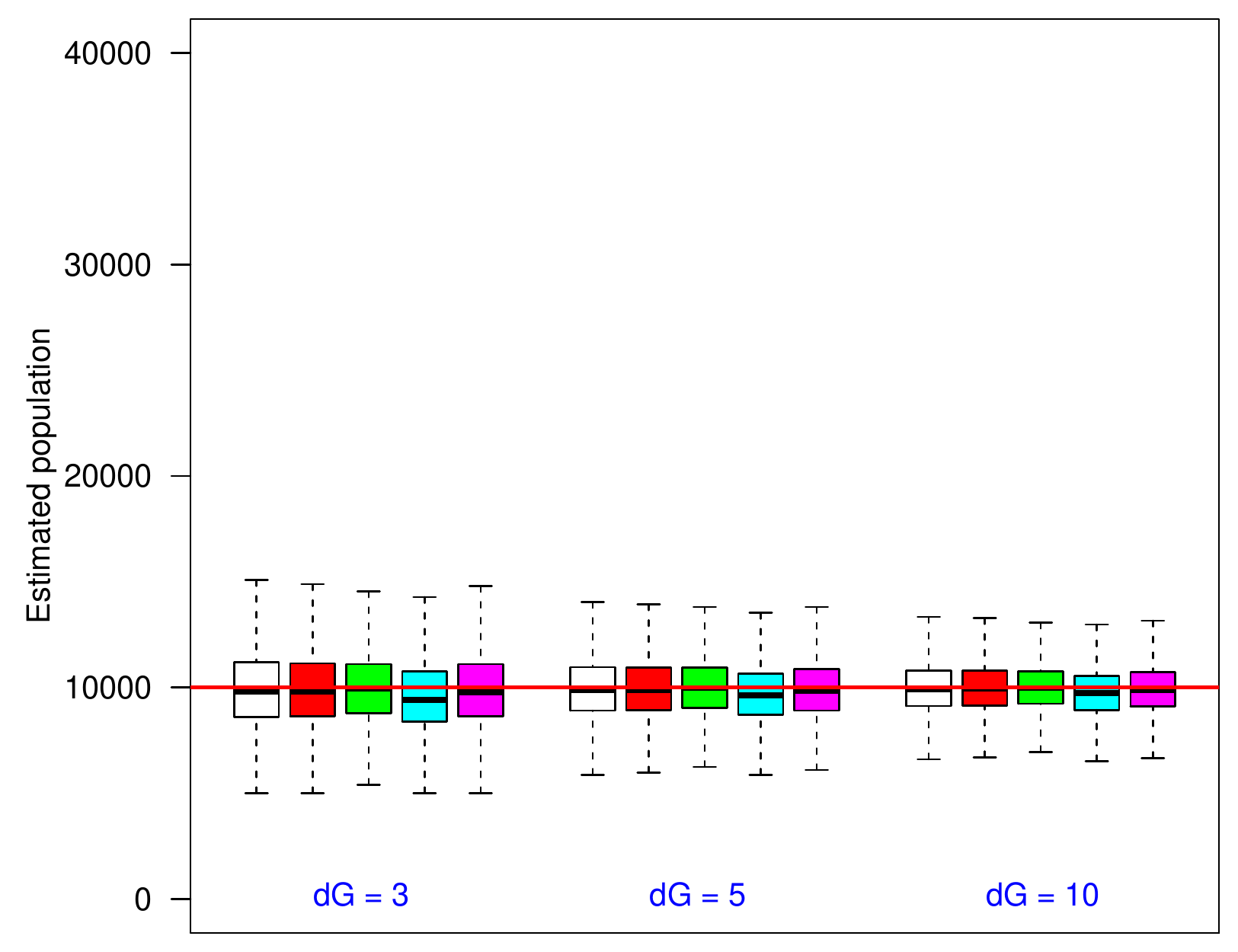}} &
\subcaptionbox{$n$ = $10\cdot 10^3$, $|\Omega| = 256 \cdot 10^3$\label{3b}}{\includegraphics[width = 0.3\linewidth]{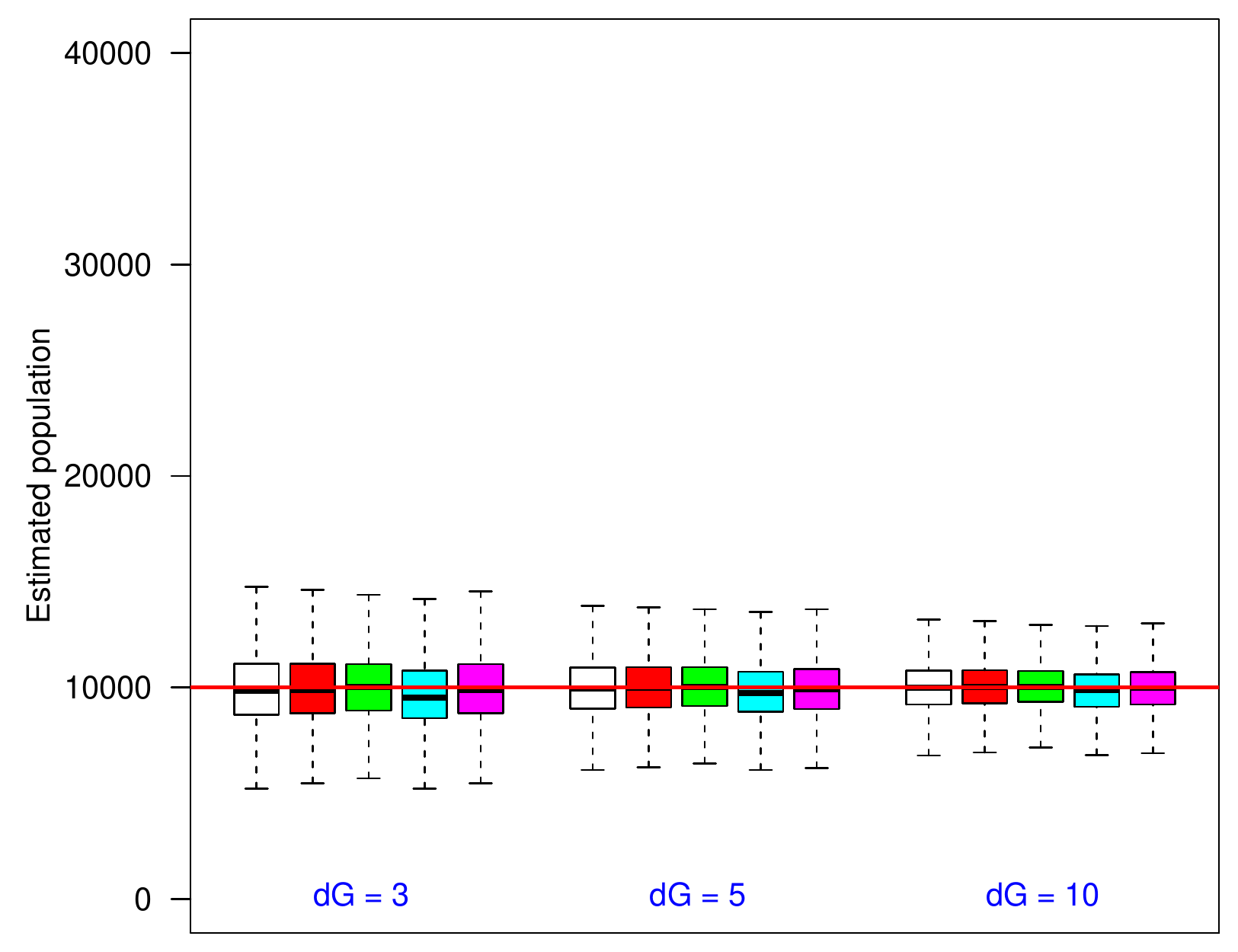}}\\
\subcaptionbox{$n$ = $20\cdot 10^3$, $|\Omega| = 2\cdot 10^3$\label{1c}}{\includegraphics[width = 0.3\linewidth]{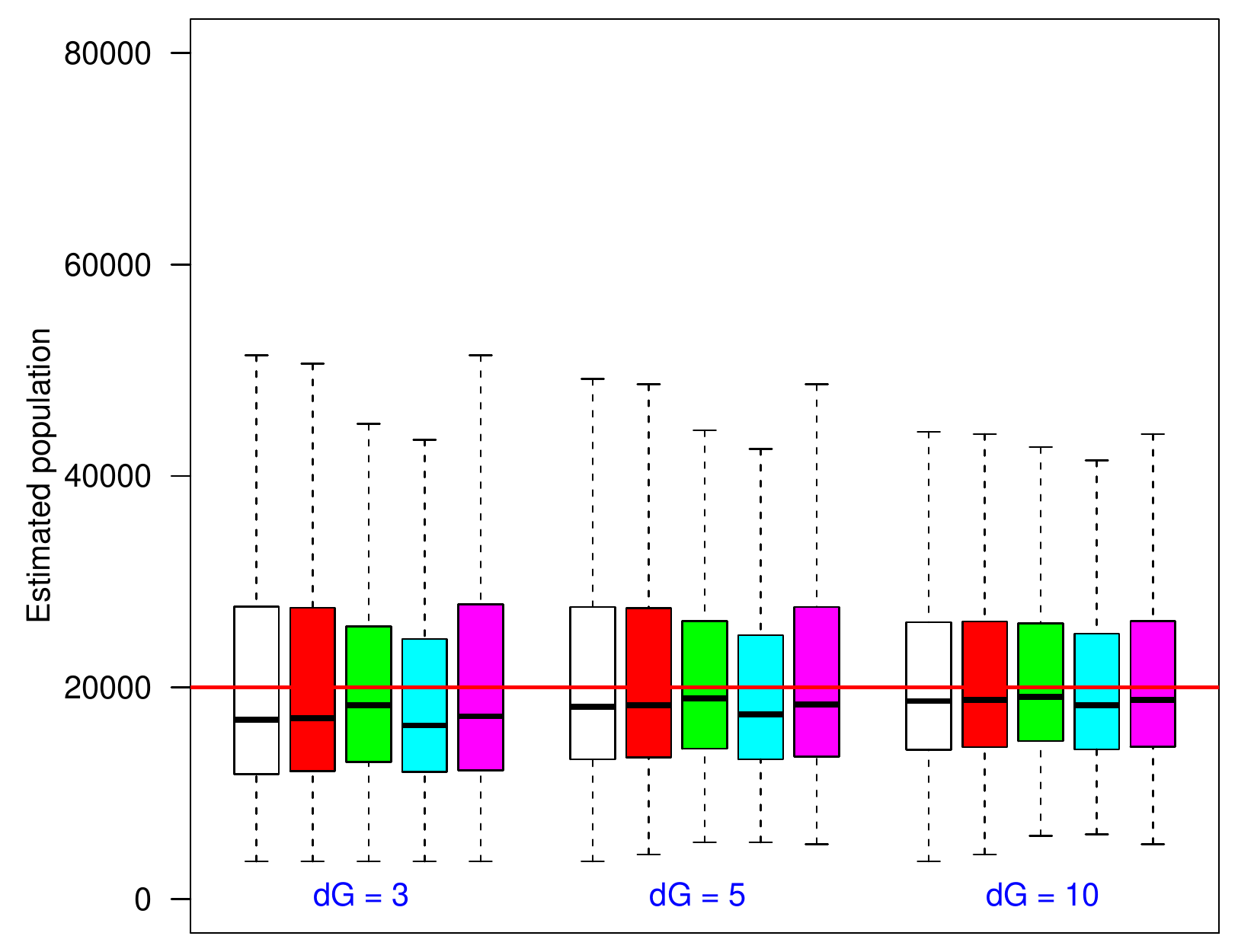}} &
\subcaptionbox{$n$ = $20\cdot 10^3$, $|\Omega| = 32\cdot 10^3$\label{2c}}{\includegraphics[width = 0.3\linewidth]{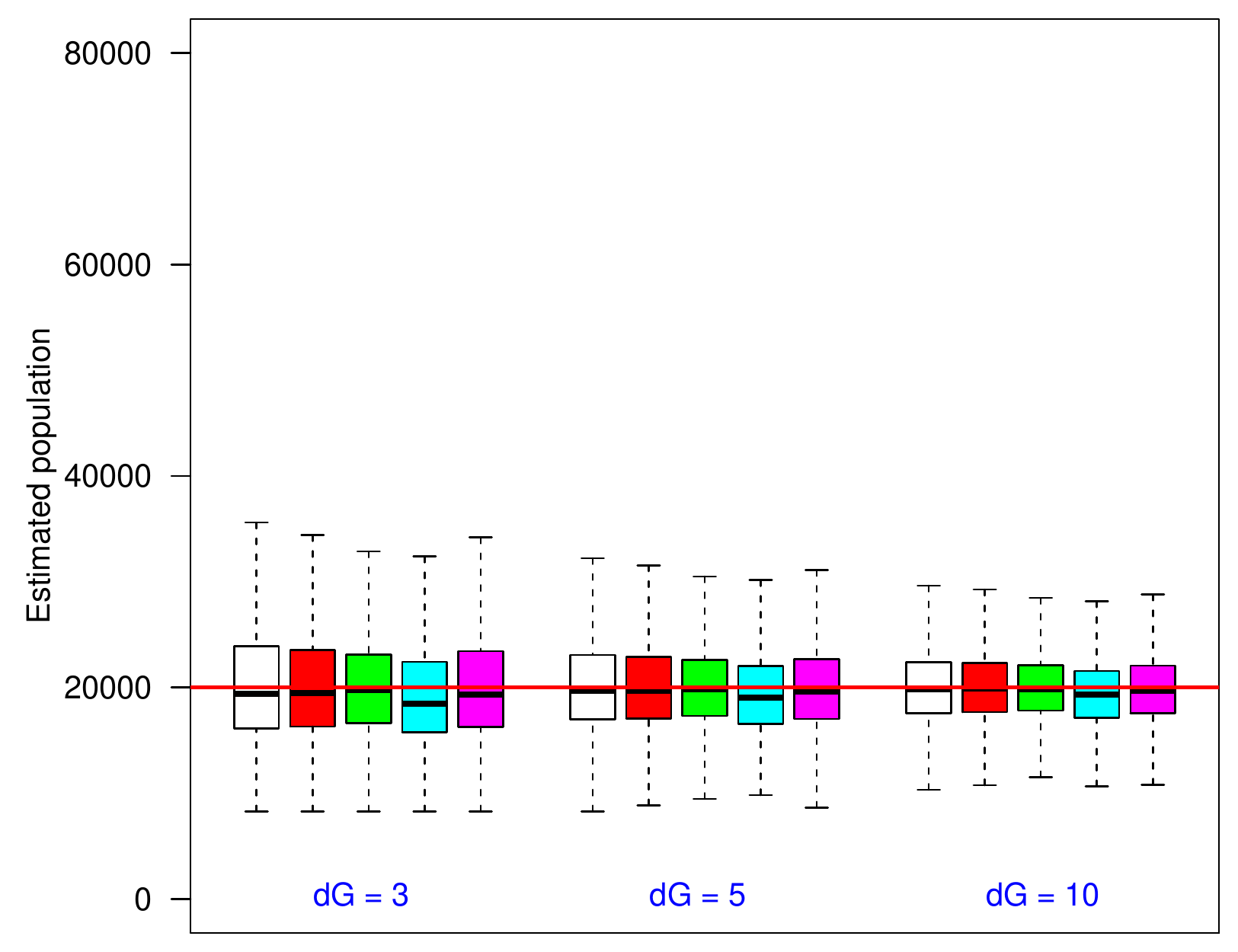}} &
\subcaptionbox{$n$ = $20\cdot 10^3$, $|\Omega| = 256 \cdot 10^3$\label{3c}}{\includegraphics[width = 0.3\linewidth]{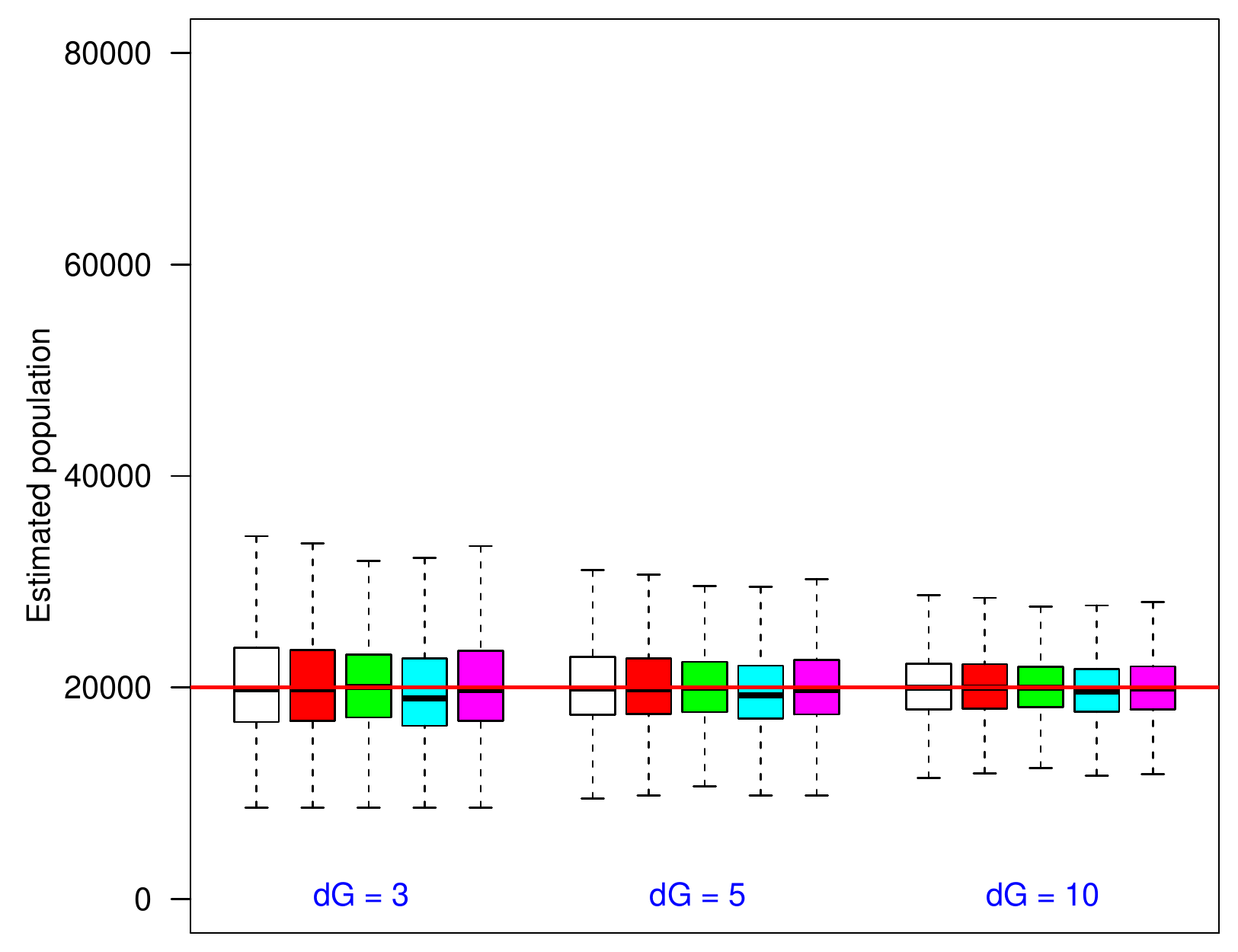}}\\[0pt]
\subcaptionbox{$n$ = $40\cdot 10^3$, $|\Omega| = 2\cdot 10^3$\label{1d}}{\includegraphics[width = 0.3\linewidth]{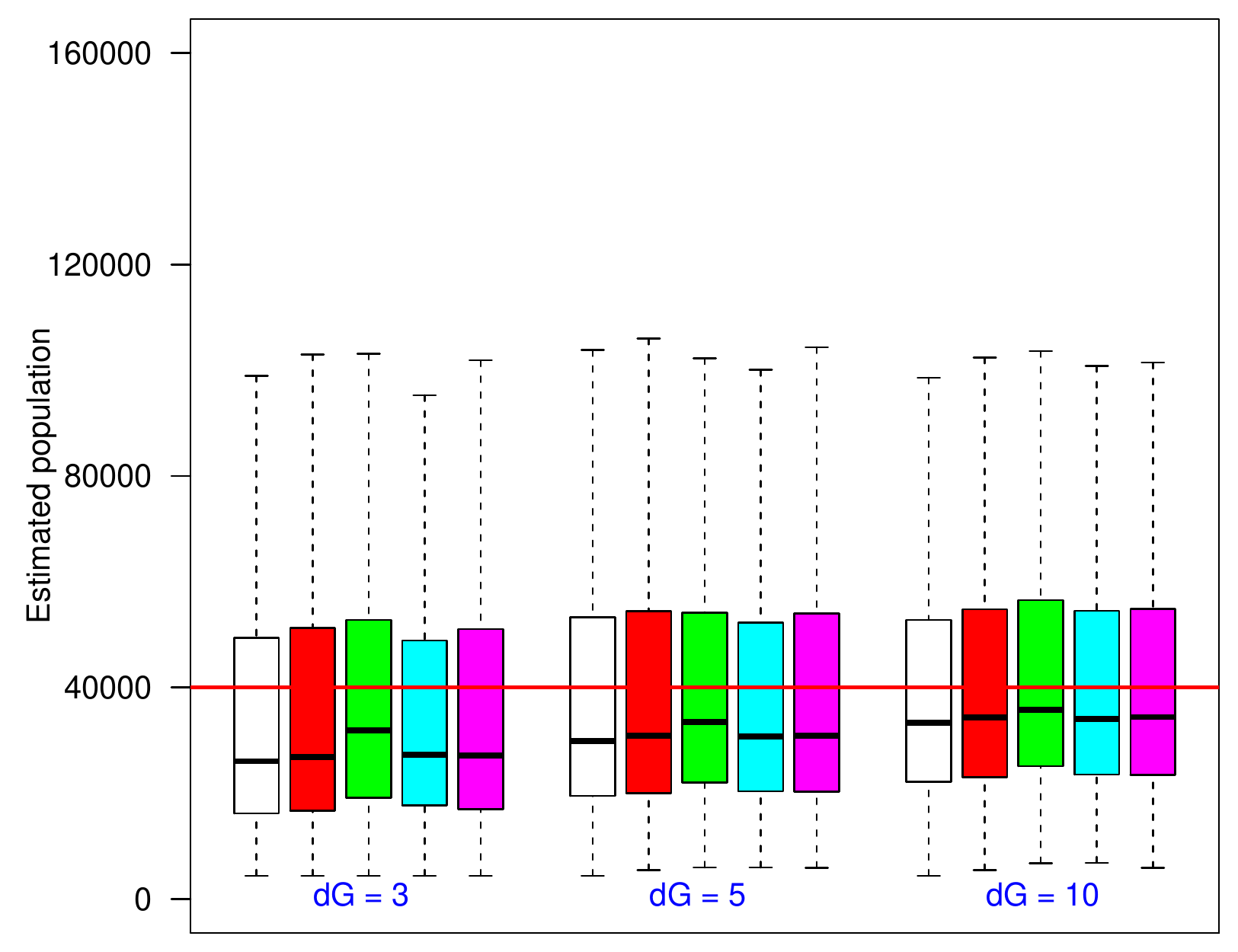}} &
\subcaptionbox{$n$ = $40\cdot 10^3$, $|\Omega| = 32\cdot 10^3$\label{2d}}{\includegraphics[width = 0.3\linewidth]{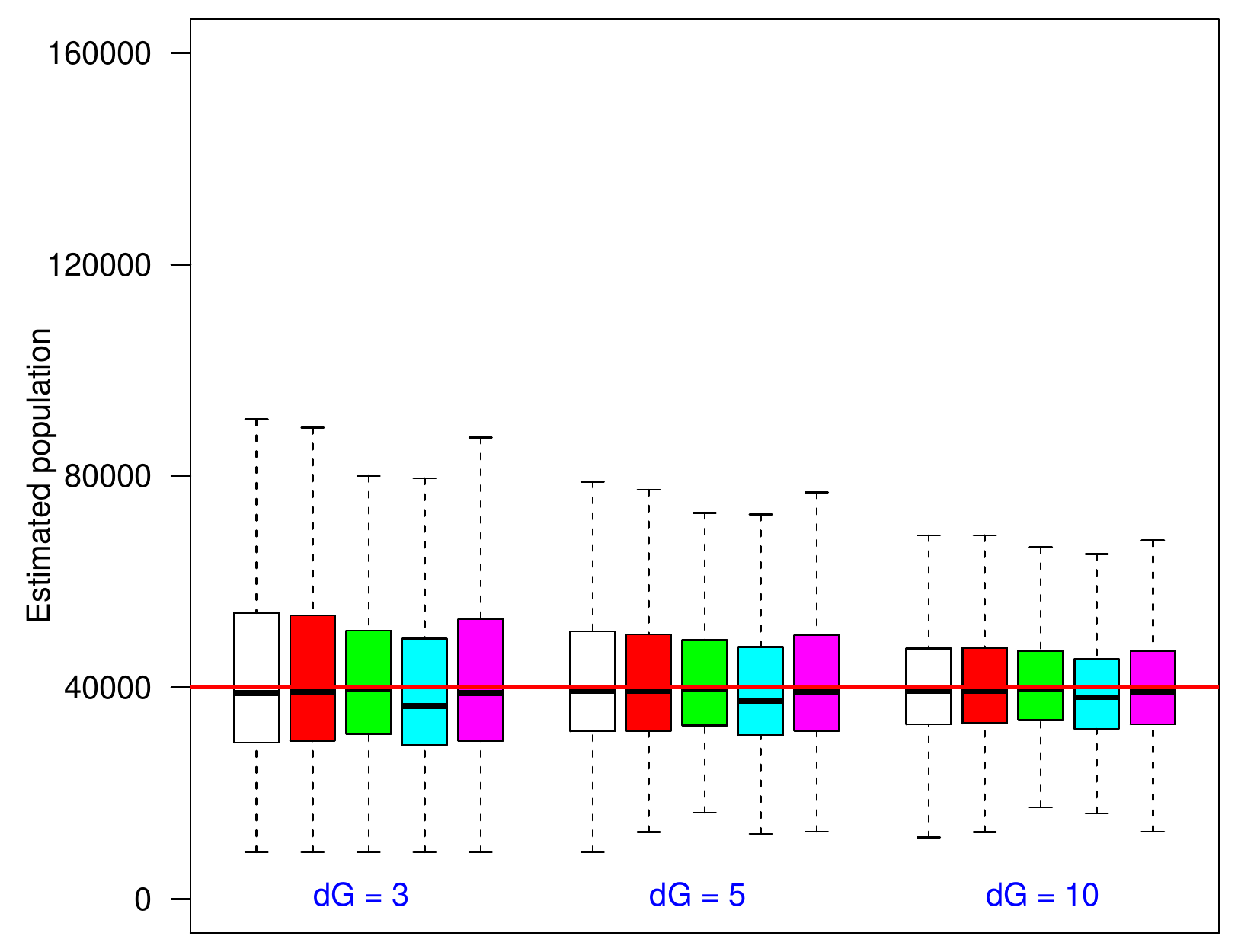}} &
\subcaptionbox{$n$ = $40\cdot 10^3$, $|\Omega| = 256 \cdot 10^3$\label{3d}}{\includegraphics[width = 0.3\linewidth]{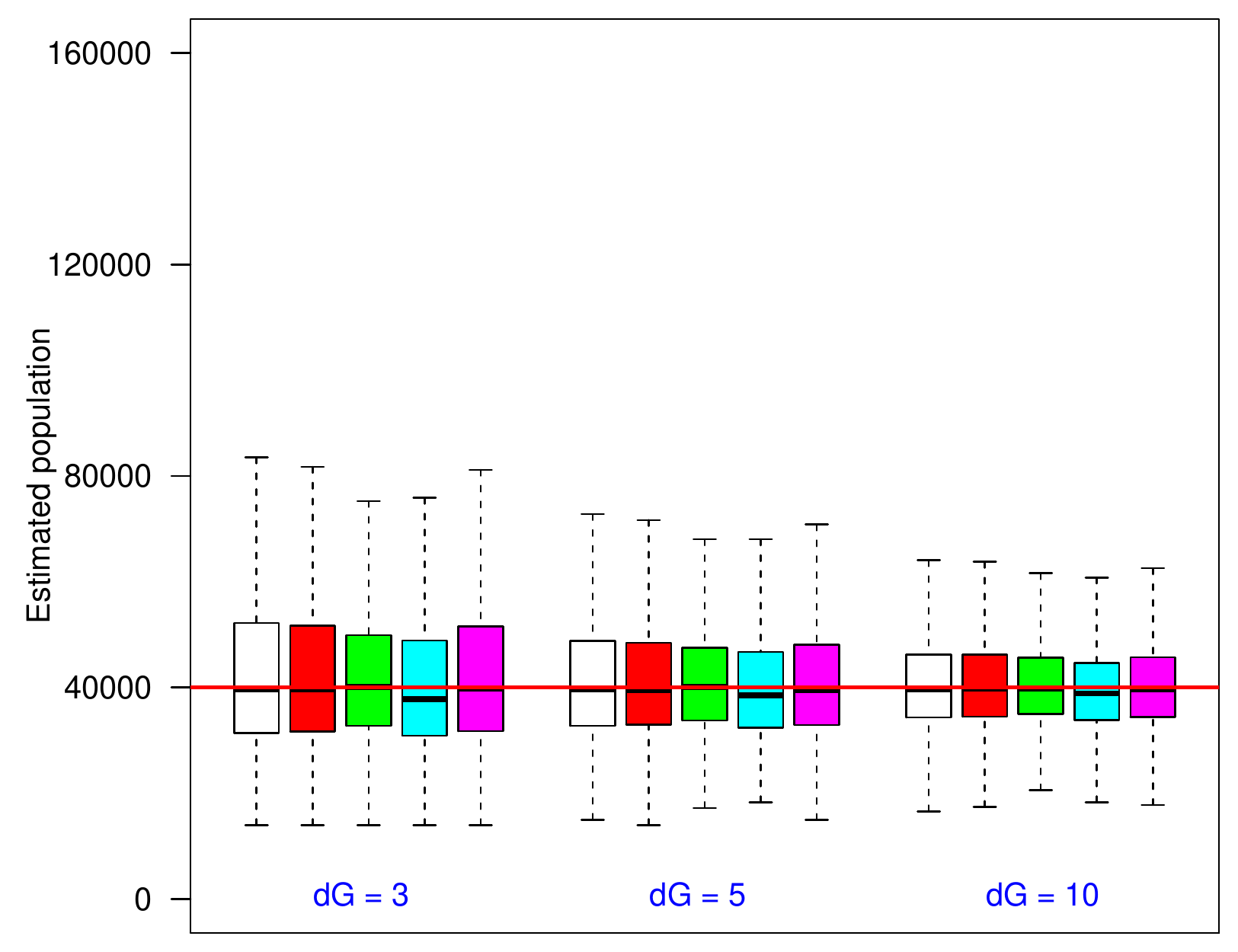}}
\end{tabular}
}
\caption{Estimator $n^{\psi}_2$ on RDS samples of size $r = 500$ with $|\Omega| =2\cdot 10^3$ to $256 \cdot 10^3$. In each box, the thick line indicates the sample median; the top of the box is the median of the upper half of the estimated values (75\% quartile); the bottom of the box indicates the median of the lower half of the estimated values (25\% quartile; and the whiskers indicate the full range of estimated values. No (finite) outliers were removed.}
\label{results:n2_psi}
\end{figure}

Figure \ref{results:n2_psi} shows that as hash space size increases, the median of $n^{\psi}_2$ converge to the true population size.  For example, when $n=5\cdot 10^3$ and $|\Omega| =2\cdot 10^3$, Lognormal degree distribution graphs with $\lambda=3$ have a median $n^{\psi}_2$ value of 4705 (a 5.9\% offset from the true value of $n=5\cdot 10^3$).  In comparison, when $|\Omega| =256 \cdot 10^3$, the median value for this family of graphs is 4901 (just 2.0\% offset from the true value).  As the hash space size increases from $|\Omega| =2\cdot 10^3$ to $|\Omega| =256 \cdot 10^3$, the error in the median estimate decreases by 3.9\%.  The magnitude of this phenomenon increases as networks grow larger.  For example for a network of size $n=40\cdot 10^3$, increasing the hash space size from $|\Omega| =2\cdot 10^3$ to $|\Omega| =256 \cdot 10^3$ causes the error in the median $n^{\psi}_2$ estimate to undergo a 33.9\% change.

In addition, Figure \ref{results:n2_psi} shows that as hash space size increases, the interquartile ranges of the estimates decrease.  For example, when $n=5\cdot 10^3$ and $|\Omega| =2\cdot 10^3$, Poisson degree distribution graphs with $\lambda=3$ experience a interquartile range of 1522 in their $n^{\psi}_2$ estimates (32.0\% of the median).  In comparison, when $|\Omega| =256 \cdot 10^3$, the interquartile range for this family of graphs decreases to 793 (a 47.9\% reduction).  The magnitude of this effect increases as networks grow larger.  For example for a network of size $n=40\cdot 10^3$, increasing the hash space size from $|\Omega| =2\cdot 10^3$ to $|\Omega| =256 \cdot 10^3$ causes the interquartile range of the $n^{\psi}_2$ estimate to undergo a 42.1\% decrease.

\subsection{Evaluating $n^{\psi}_3$ on Synthetic Networks}
\label{sec:eval-n3-psi}

A second set of experiments shows the performance of the $n^{\psi}_3$ performance under identical hashing conditions used to test $n^{\psi}_2$. These experiments also follow the framework described in Section \ref{sec:exp-framework} and use samples derived from an RDS process operating as specified in Assumption \ref{def:rds-assumptions}.  The hash space size was varied from $|\Omega|=2\cdot 10^3$ to $256 \cdot 10^3$.

The $12$ graphs in Figure \ref{results:n3_psi} present the performance of the $n^{\psi}_3$ estimator as the true population size $n$ is varied from $5\cdot 10^3$ to $40\cdot 10^3$ (vertical axis of the grid), the sample size is fixed to $r = 500$ and the hash space size was varied from $|\Omega|=2\cdot 10^3$ to $256 \cdot 10^3$ (horizontal axis of the grid).  In each of the $12$ graphs, the x-axis varies the average degree $\lambda$ from $3$ to $10$.  For each choice of $\lambda$, the medians and quartile ranges of $n^{\psi}_3$ are given for each of the $5$ graph families.  Each of these is determined by $900$ simulations ($30$ graphs times $30$ uniformly drawn samples in each graph).

\begin{figure}[p]
\setlength{\belowcaptionskip}{12pt}
\centering
\begin{tabular}{ScScScSc}
\subcaptionbox{$n$ = $5\cdot 10^3$, $|\Omega| = 2\cdot 10^3$\label{1a}}{\includegraphics[width = 0.3\linewidth]{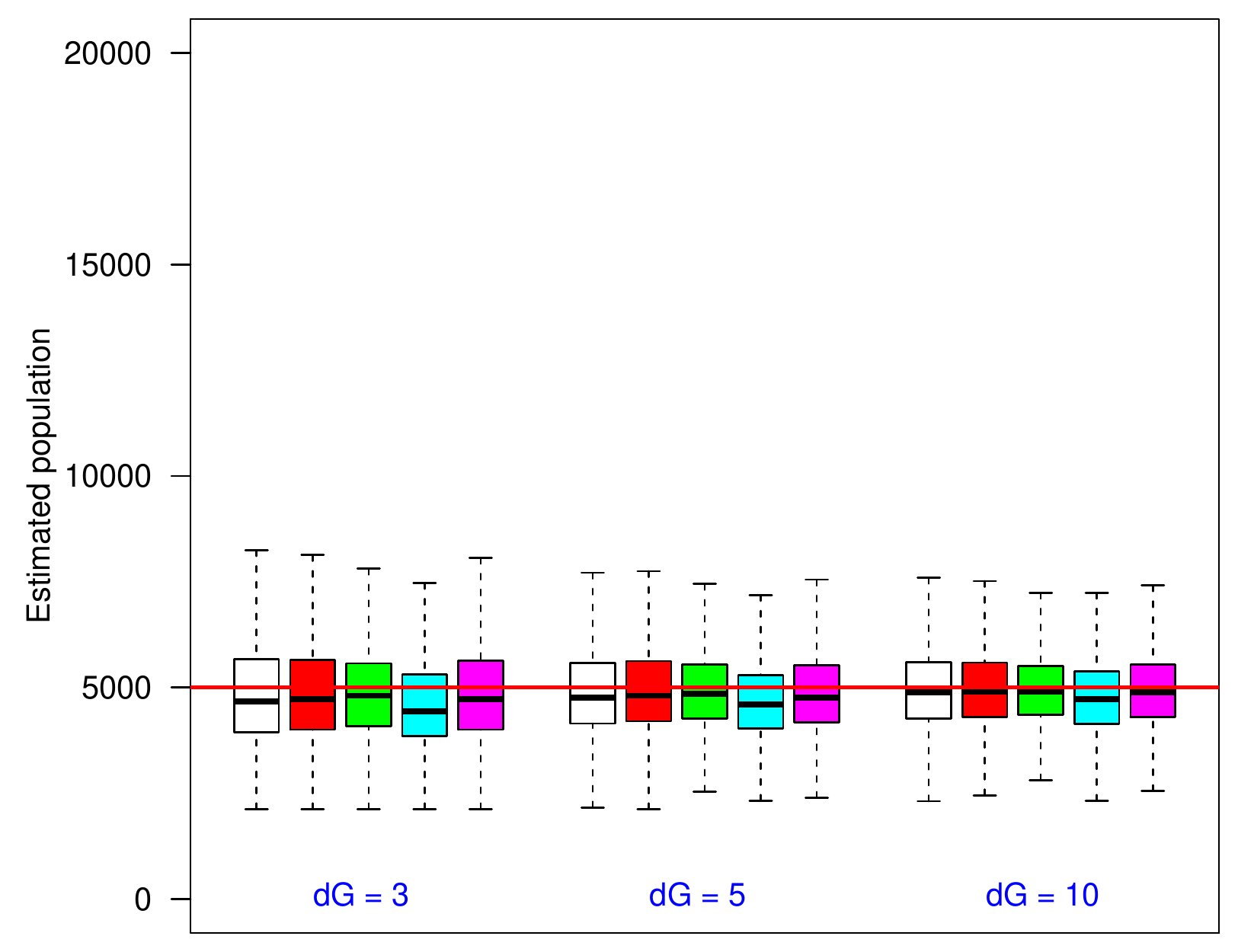}} &
\subcaptionbox{$n$ = $5\cdot 10^3$, $|\Omega| = 32\cdot 10^3$\label{2a}}{\includegraphics[width = 0.3\linewidth]{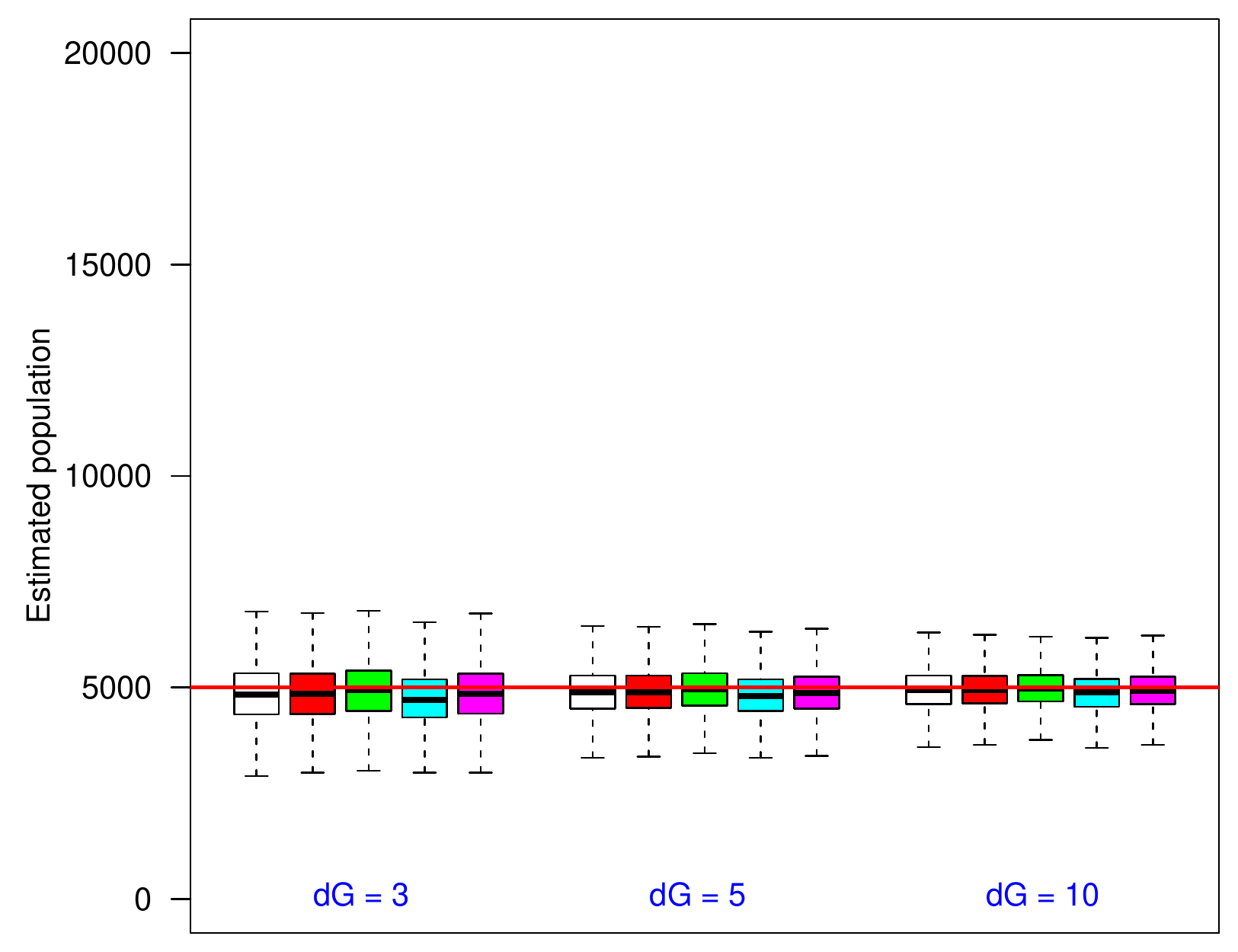}} &
\subcaptionbox{$n$ = $5\cdot 10^3$, $|\Omega| = 256 \cdot 10^3$\label{3a}}{\includegraphics[width = 0.3\linewidth]{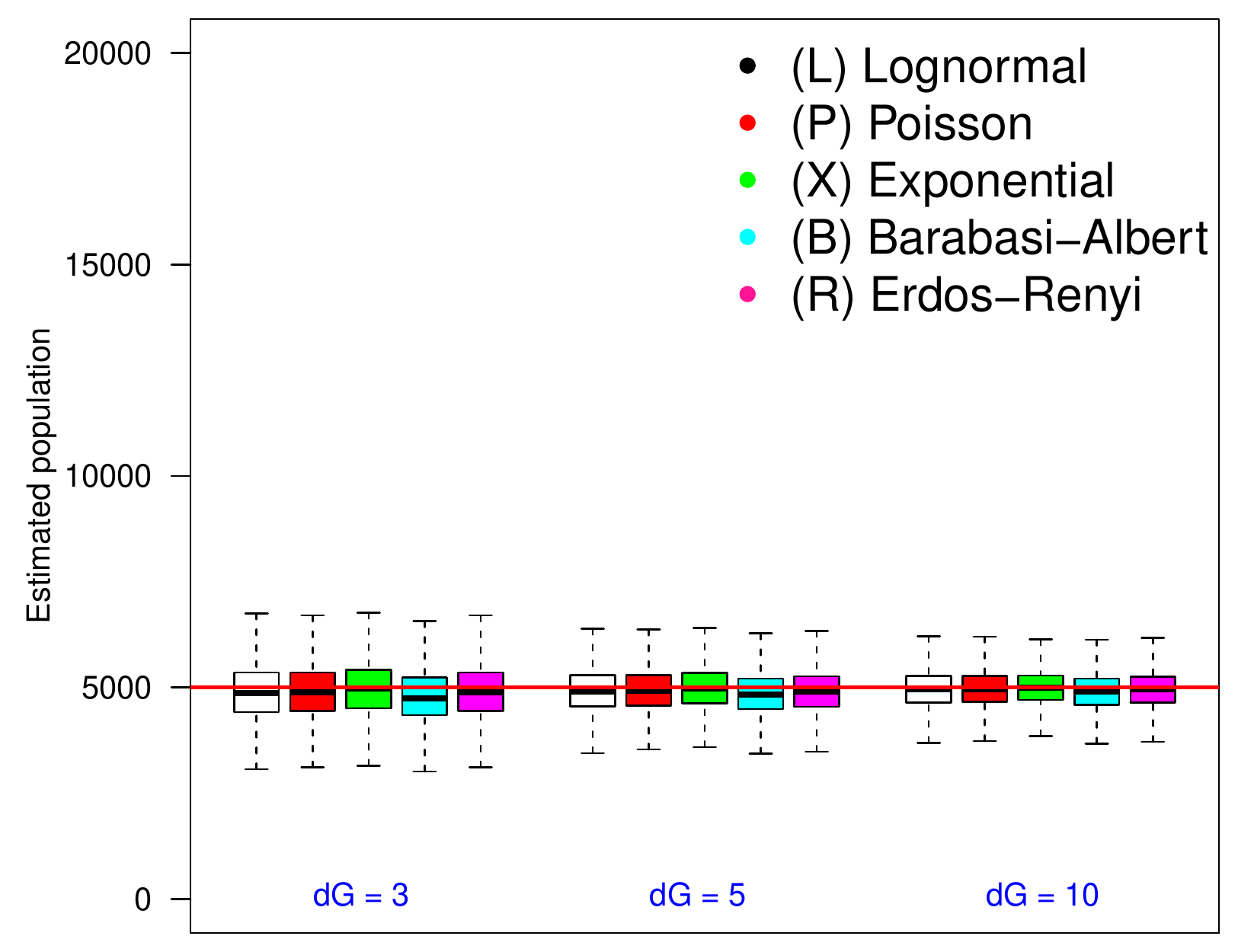}}\\[0pt]
\subcaptionbox{$n$ = $10\cdot 10^3$, $|\Omega| = 2\cdot 10^3$\label{1b}}{\includegraphics[width = 0.3\linewidth]{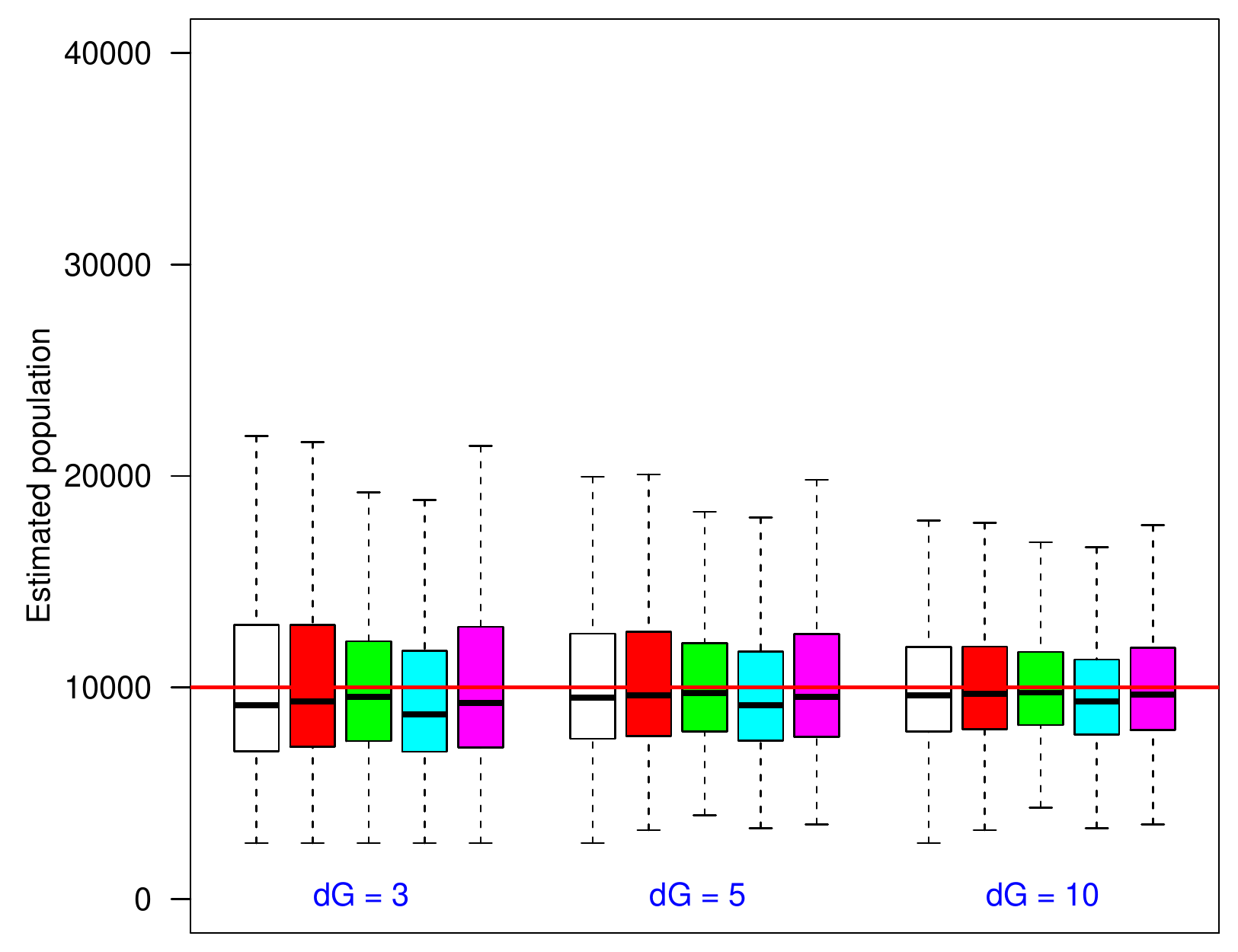}} &
\subcaptionbox{$n$ = $10\cdot 10^3$, $|\Omega| = 32\cdot 10^3$\label{2b}}{\includegraphics[width = 0.3\linewidth]{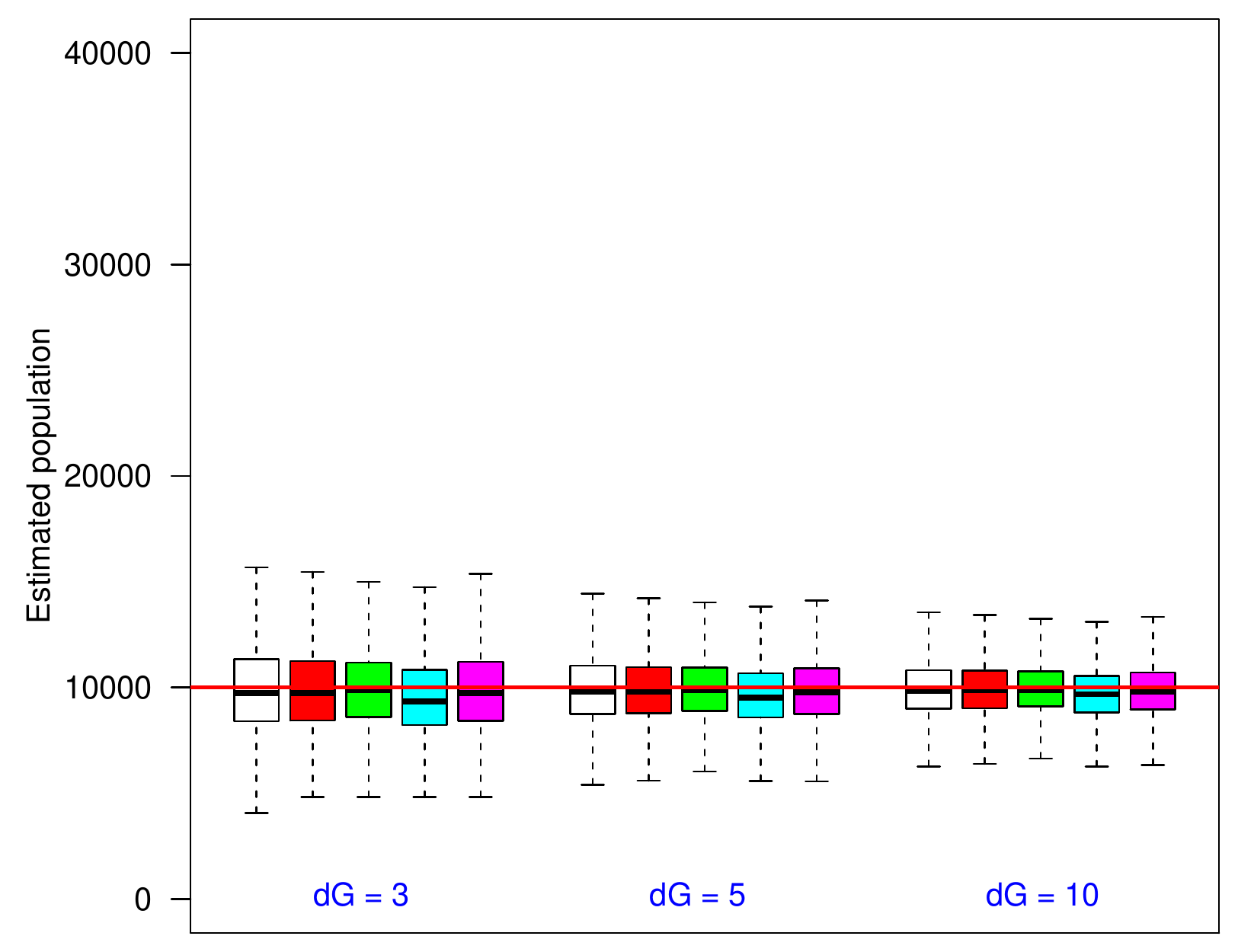}} &
\subcaptionbox{$n$ = $10\cdot 10^3$, $|\Omega| = 256 \cdot 10^3$\label{3b}}{\includegraphics[width = 0.3\linewidth]{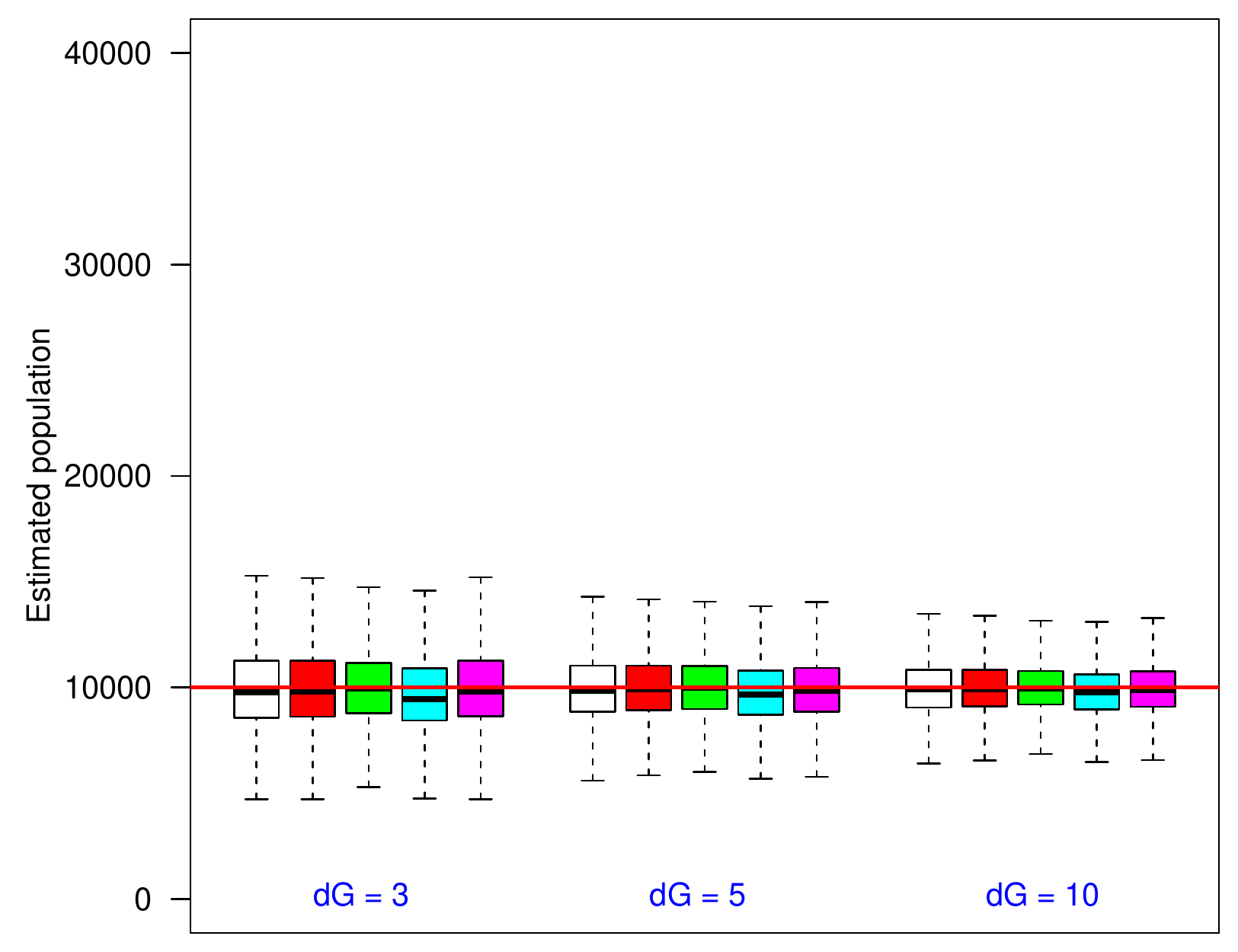}}\\
\subcaptionbox{$n$ = $20\cdot 10^3$, $|\Omega| = 2\cdot 10^3$\label{1c}}{\includegraphics[width = 0.3\linewidth]{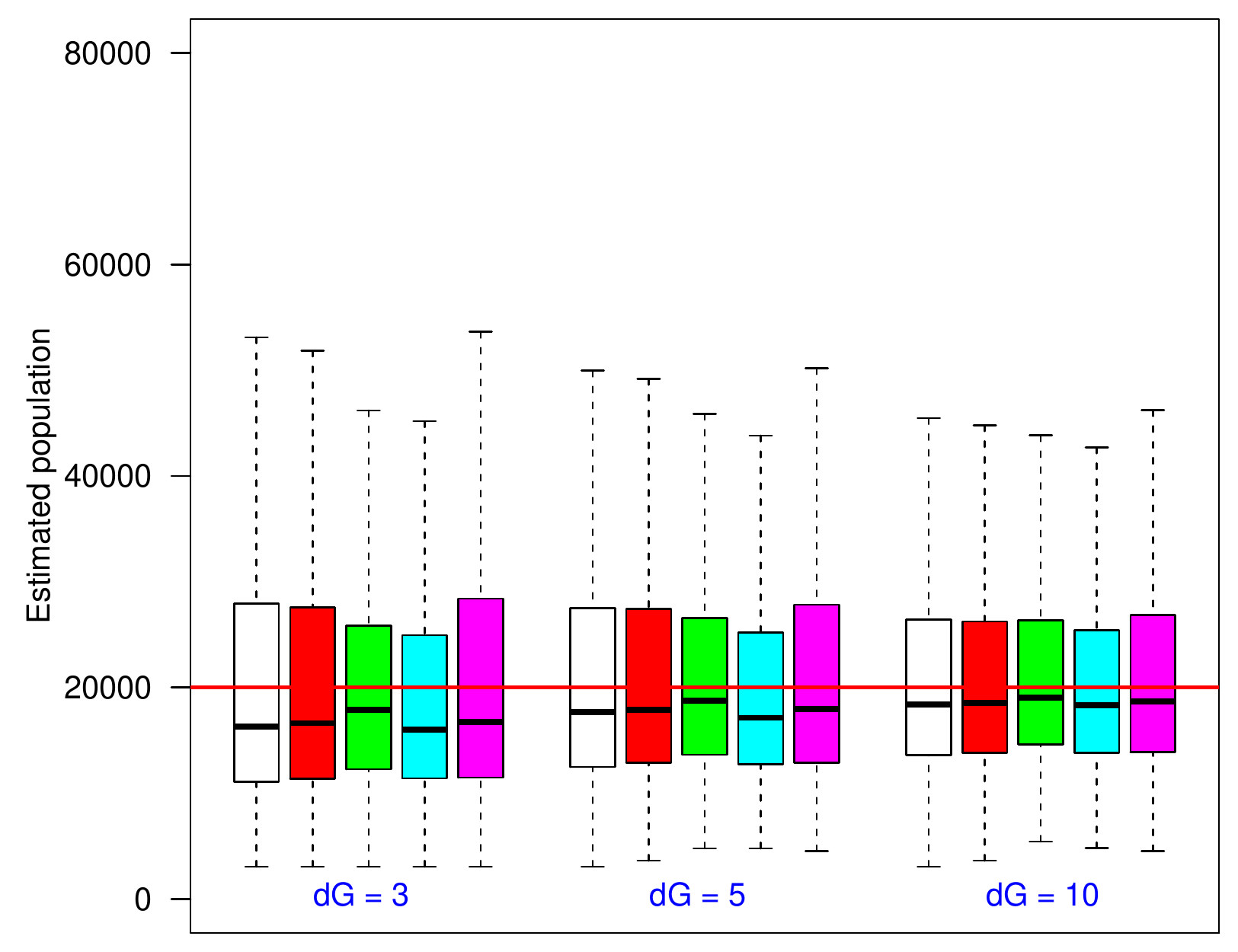}} &
\subcaptionbox{$n$ = $20\cdot 10^3$, $|\Omega| = 32\cdot 10^3$\label{2c}}{\includegraphics[width = 0.3\linewidth]{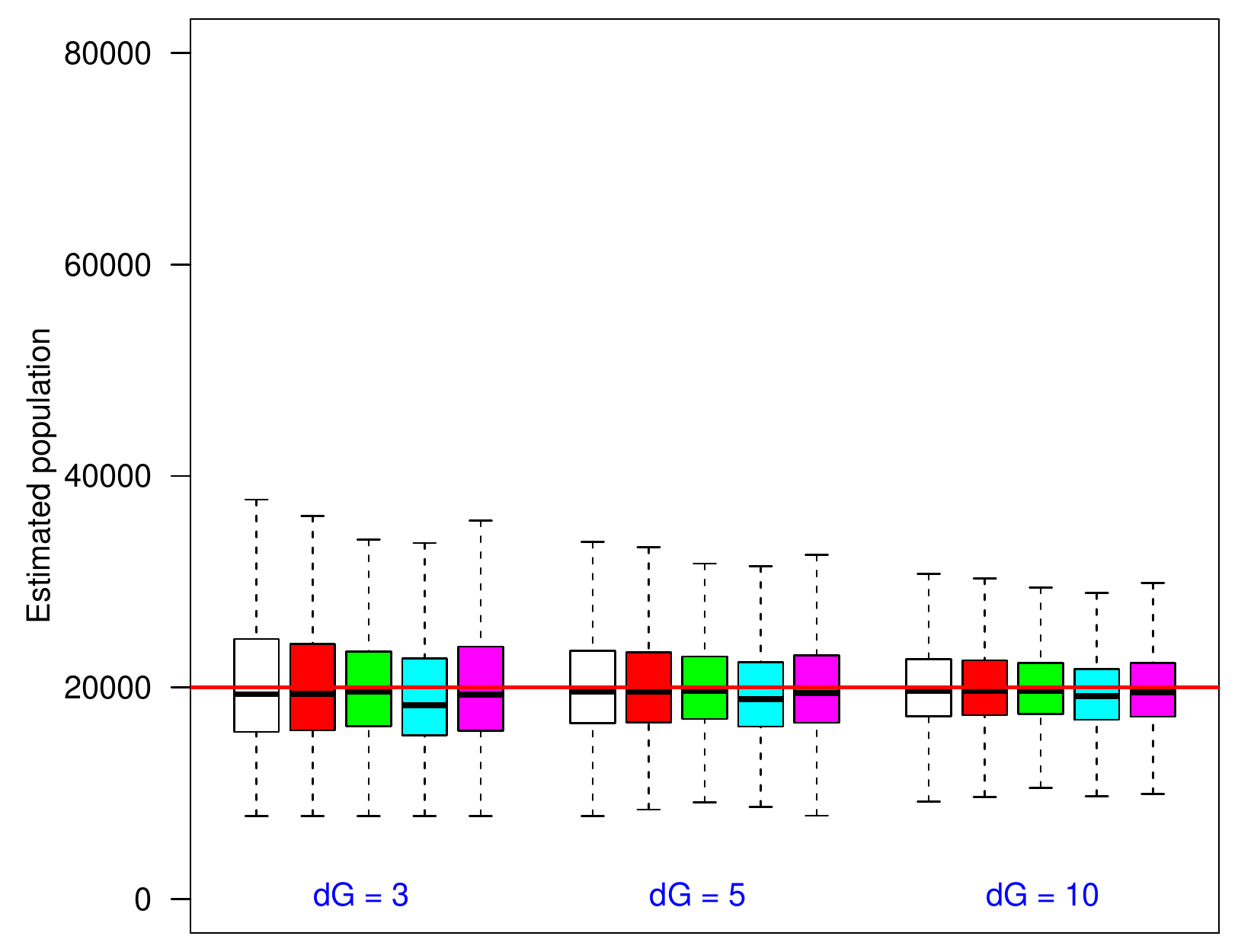}} &
\subcaptionbox{$n$ = $20\cdot 10^3$, $|\Omega| = 256 \cdot 10^3$\label{3c}}{\includegraphics[width = 0.3\linewidth]{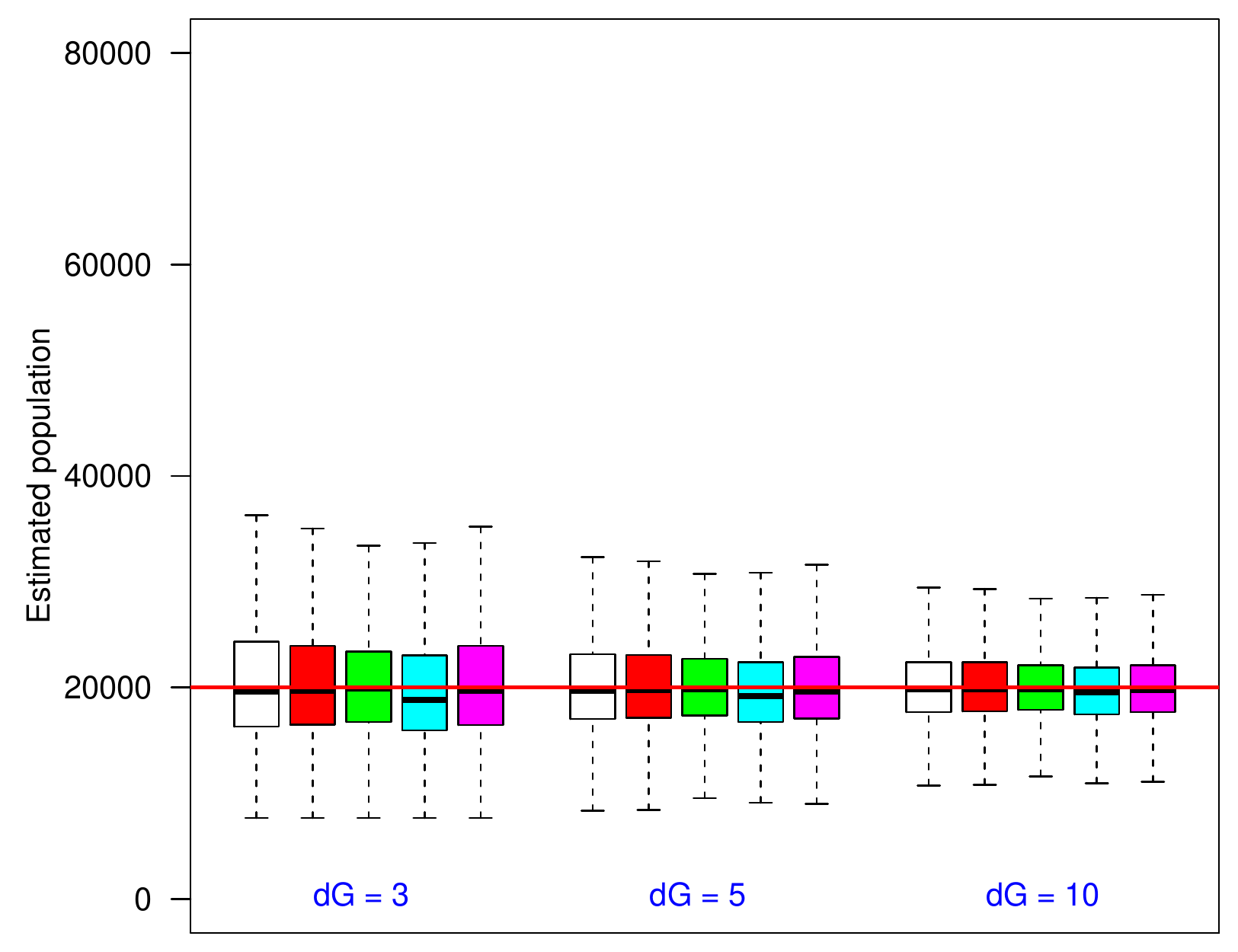}}\\
\subcaptionbox{$n$ = $40\cdot 10^3$, $|\Omega| = 2\cdot 10^3$\label{1d}}{\includegraphics[width = 0.3\linewidth]{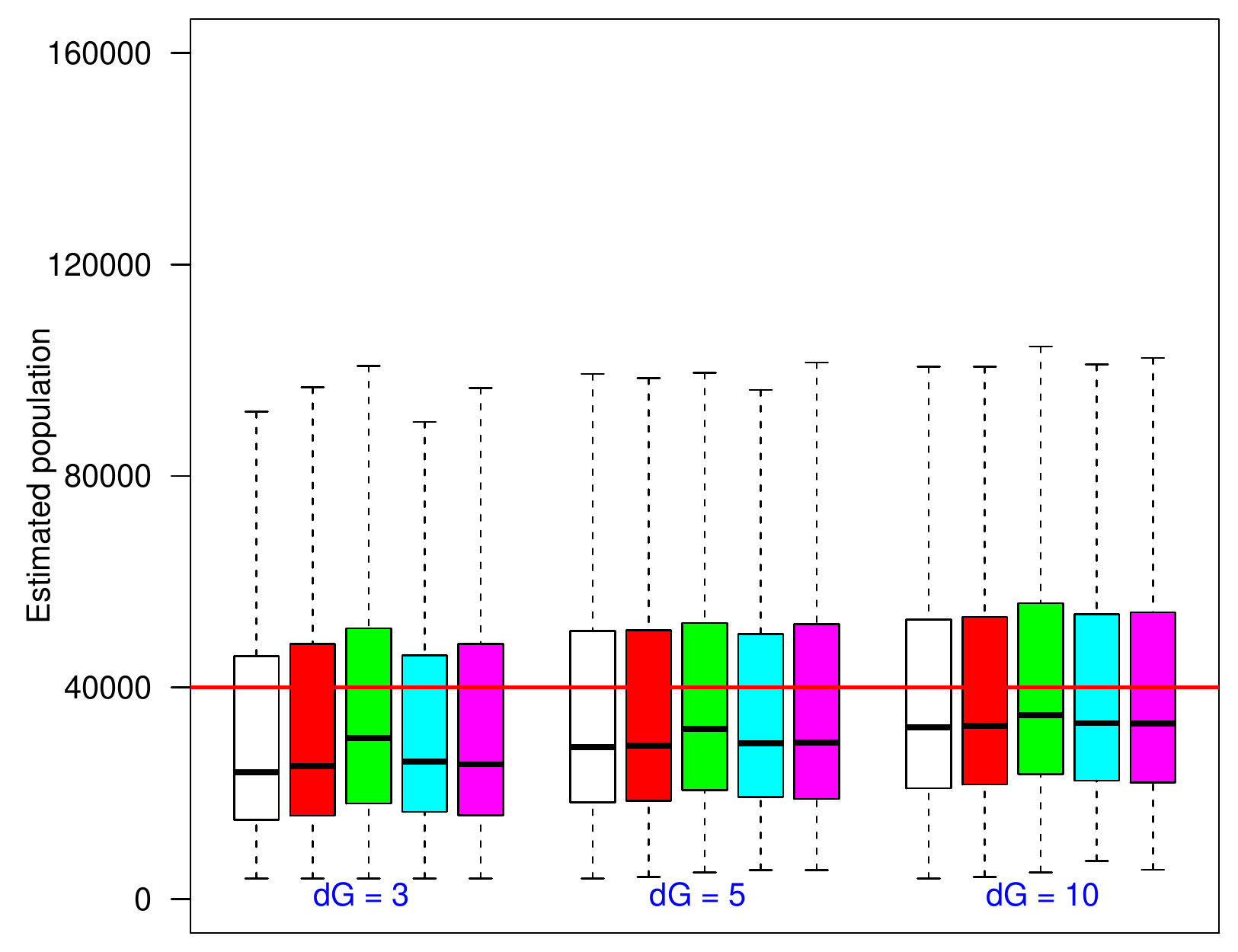}} &
\subcaptionbox{$n$ = $40\cdot 10^3$, $|\Omega| = 32\cdot 10^3$\label{2d}}{\includegraphics[width = 0.3\linewidth]{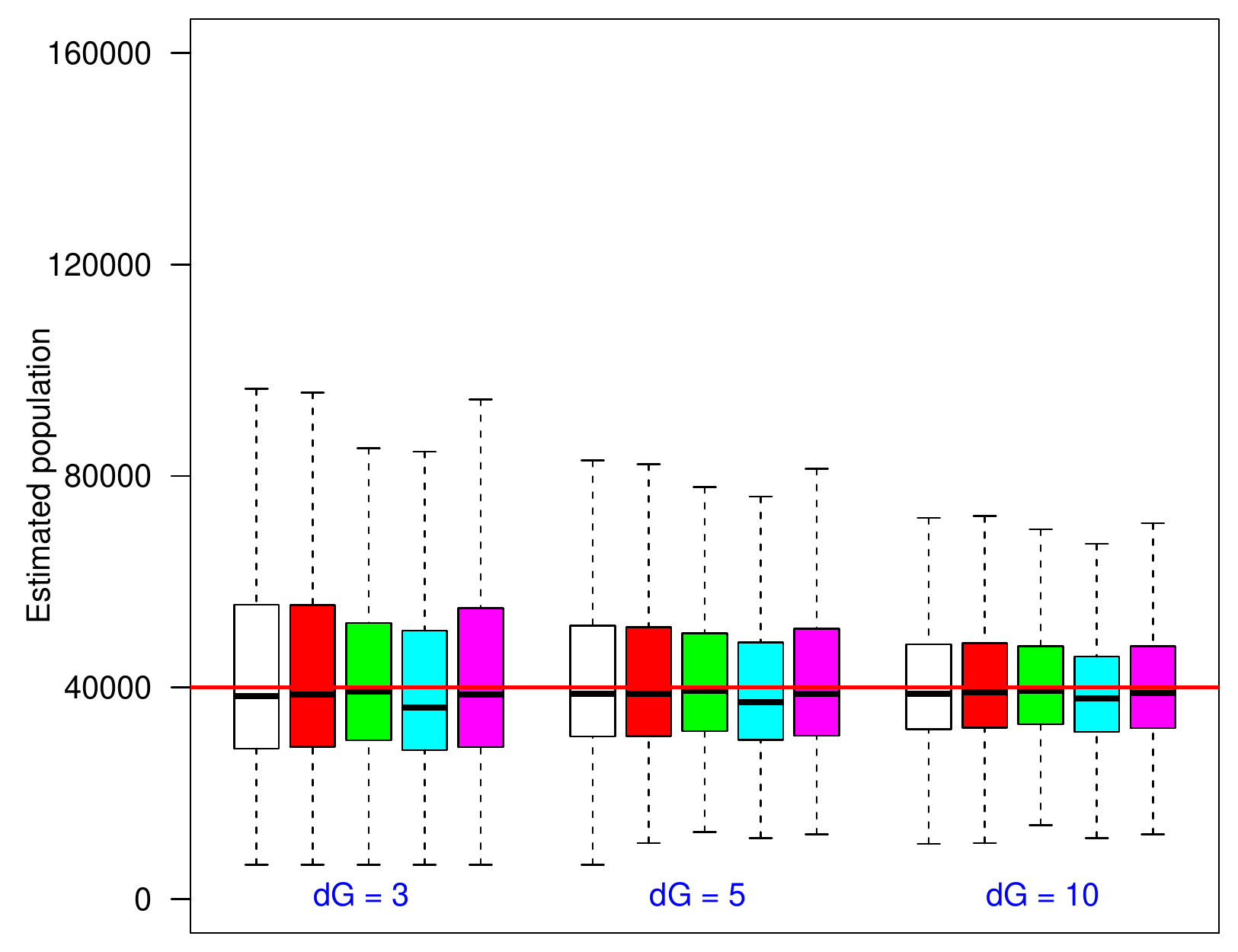}} &
\subcaptionbox{$n$ = $40\cdot 10^3$, $|\Omega| = 256 \cdot 10^3$\label{3d}}{\includegraphics[width = 0.3\linewidth]{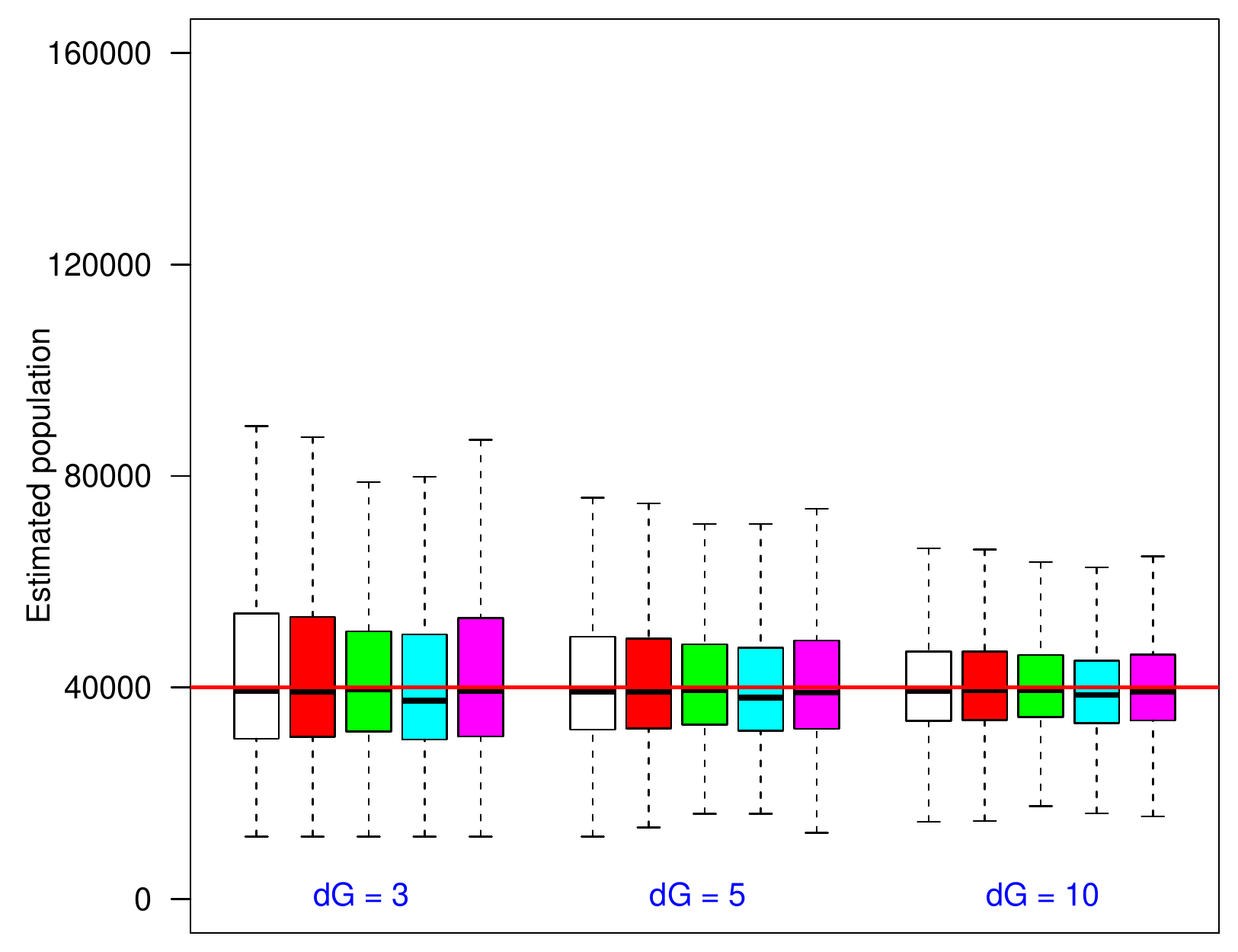}}
\end{tabular}
\caption{Estimator $n^{\psi}_3$ on RDS samples of size $r = 500$ with $|\Omega| =2\cdot 10^3$ to $256 \cdot 10^3$. In each box, the thick line indicates the sample median; the top of the box is the median of the upper half of the estimated values (75\% quartile); the bottom of the box indicates the median of the lower half of the estimated values (25\% quartile; and the whiskers indicate the full range of estimated values. No (finite) outliers were removed.}
\label{results:n3_psi}
\end{figure}

Figure \ref{results:n3_psi} shows that as hash space size increases, the medians of $n^{\psi}_3$ converge to the true population size.  For example, when $n=5\cdot 10^3$ and $|\Omega|=2\cdot 10^3$, Lognormal degree distribution graphs with $\lambda=3$ have a median $n^{\psi}_3$ value of 4667 (a 6.7\% offset from the true value of $n=5\cdot 10^3$).  In comparison, when $|\Omega|=256 \cdot 10^3$, the median for this family of graphs is 4865 (just 2.7\% offset from the true value).  As the hash space size increases from $|\Omega|=2\cdot 10^3$ to $|\Omega|=256 \cdot 10^3$, the error in the median estimate decreases by 4.0\%.  The magnitude of this phenomenon increases as networks grow larger.  For example for a network of size $n=40\cdot 10^3$, increasing the hash space size from $|\Omega|=2\cdot 10^3$ to $|\Omega|=256 \cdot 10^3$ causes the error in the median $n^{\psi}_3$ estimate to undergo a 38.4\% change.

In addition, Figure \ref{results:n3_psi} shows that as hash space size increases, the interquartile ranges of the estimates decrease.  For example, when $n=5\cdot 10^3$ and $|\Omega|=2\cdot 10^3$, Exponential degree distribution graphs with $\lambda=3$ experience a interquartile range of 1491 in their $n^{\psi}_3$ estimates (31.0\% of the median).  In comparison, when $|\Omega|=256 \cdot 10^3$, the interquartile range for this family of graphs decreases to 905 (a 39.3\% reduction).  The magnitude of this effect increases as networks grow larger.  For example for a network of size $n=40\cdot 10^3$, increasing the hash space size from $|\Omega|=2\cdot 10^3$ to $|\Omega|=256 \cdot 10^3$ causes the interquartile range of the $n^{\psi}_3$ estimate to undergo a 43.0\% decrease.

\section{Evaluating Estimators on Real Networks}
\label{sec:Brightkite}
While a range of degree distributions and randomly occurring clusterings can be expected in idealized topologies, the performance of RDS-based estimators $n^{\psi}_2$ and $n^{\psi}_3$ on organically arising, natural human networks may vary. To test this possibility, we perform a number of random-start, RDS-based estimation experiments on the Brightkite data set. Brightkite was once a location-based social networking service provider where users shared their locations by checking-in. The friendship network was collected using their public API, and consists of $|V|=$58,228 nodes and $|E|=$214,078 edges \cite{Cho:2011:FMU:2020408.2020579}. Though originally a directed graph, edges were symmetrized for the purposes of these experiments. Since not all users made a public check-in during the data collection period, the population we used here is 51,406 people. The average clustering coefficient in the network was 0.1723, while the fraction of closed triangles was 0.03979.  The diameter (longest shortest path in the symmetrized network) is 16, though the 90-percentile effective diameter is 6.

For purposes of the experiment we generated 900 respondent-driven samples of size $r=250,500,750$ and hash space size from $|\Omega|=2\cdot 10^3$ to $|\Omega|=256 \cdot 10^3$ within the Brightkite network, each obtained via an RDS process operating as specified in Assumption \ref{def:rds-assumptions}.  The boxplot graphs in Figure \ref{rds-Brightkite-alg3} (a-c) show that estimator $n^{\psi}_2$---where no accommodation is made for the tendency of RDS to oversample tightly clustered network neighborhoods---underestimates the true population size of 51,406 in every case. Given the high clustering coefficient of the network (17.2\%), it seems likely that, for a given sampling tree, the peer-discovery process would necessarily walk across close pairs of nodes that shared one or more common vertices. Of note is that  increasing the sample size and hash space size does little to correct for these effects.  

Graphs (d-f) in Figure \ref{rds-Brightkite-alg4} present the boxplots of Brightkite population estimate using estimator $n^{\psi}_3$. As above, we generated 900 respondent-driven samples of size $r=250,500,750$ and hash space size from $|\Omega|=2\cdot 10^3$ to $|\Omega|=256 \cdot 10^3$ within the Brightkite network. We see that the three different hash space sizes show similar results, while increasing the sample size $r$ from 250 to 500 and 750 improves the accuracy of the median estimate. Unlike the case in Figure \ref{rds-Brightkite-alg3} (a-c), we don't see a consistent pattern of underestimation, indicating that the cross-seed estimator $n^{\psi}_3$ was successful in compensating for the clustering found in the network. As above, the overall size of the hash space has minimal effect on the accuracy of the median, but  we note that an increase in the RDS sample size improves the accuracy of the median estimate and produces smaller interquartile ranges.  

\begin{figure}[t]
\setlength{\belowcaptionskip}{12pt}
\centering
\begin{tabular}{ScScScSc}
\subcaptionbox{ $n^{\psi}_2$ with $|\Omega| = 2 \cdot 10^3$\label{1c}}{\includegraphics[width = 0.3\linewidth]{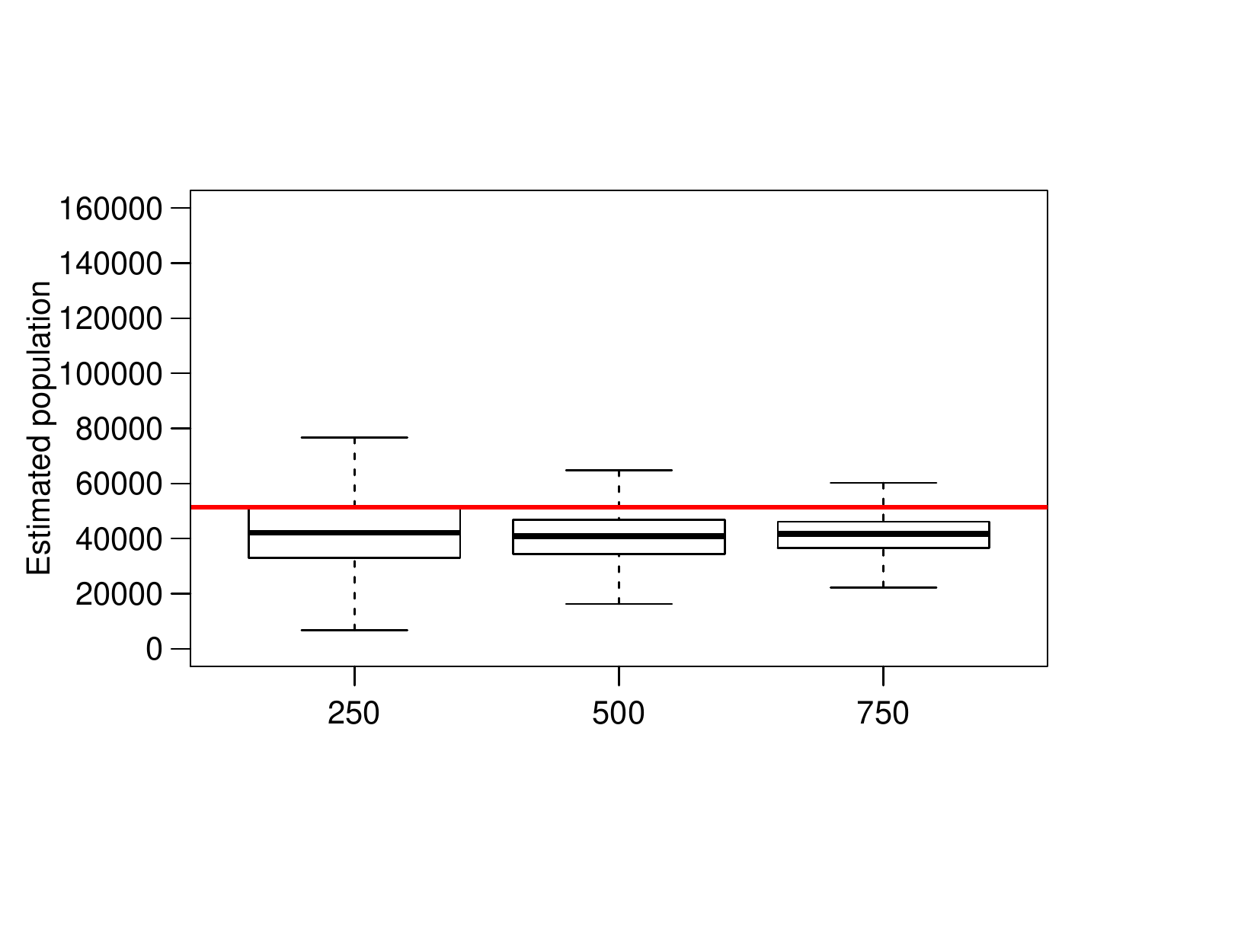}} &
\subcaptionbox{ $n^{\psi}_2$ with $|\Omega| = 32 \cdot 10^3$\label{2c}}{\includegraphics[width = 0.3\linewidth]{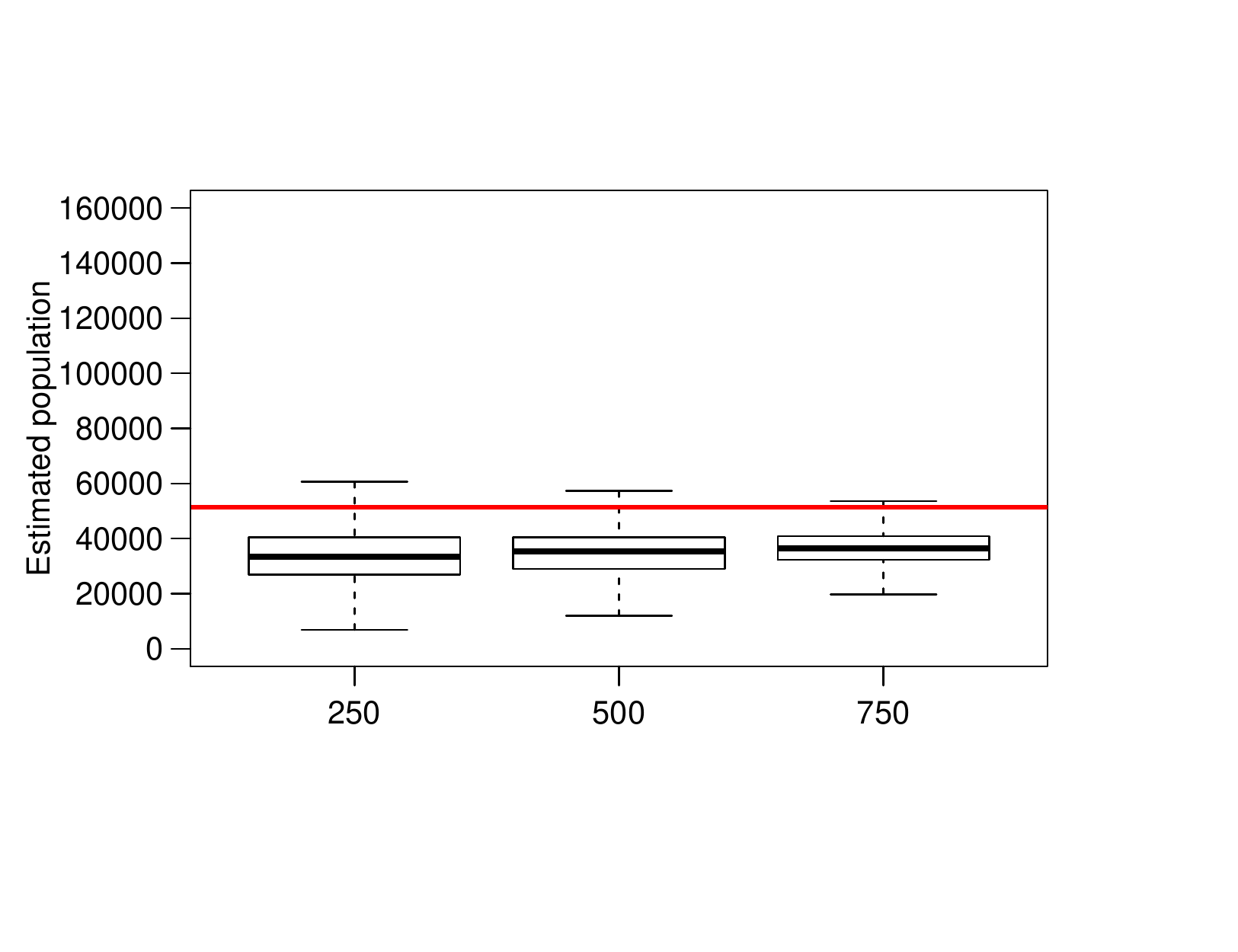}} &
\subcaptionbox{ $n^{\psi}_2$ with $|\Omega| = 256 \cdot 10^3$\label{3c}}{\includegraphics[width = 0.3\linewidth]{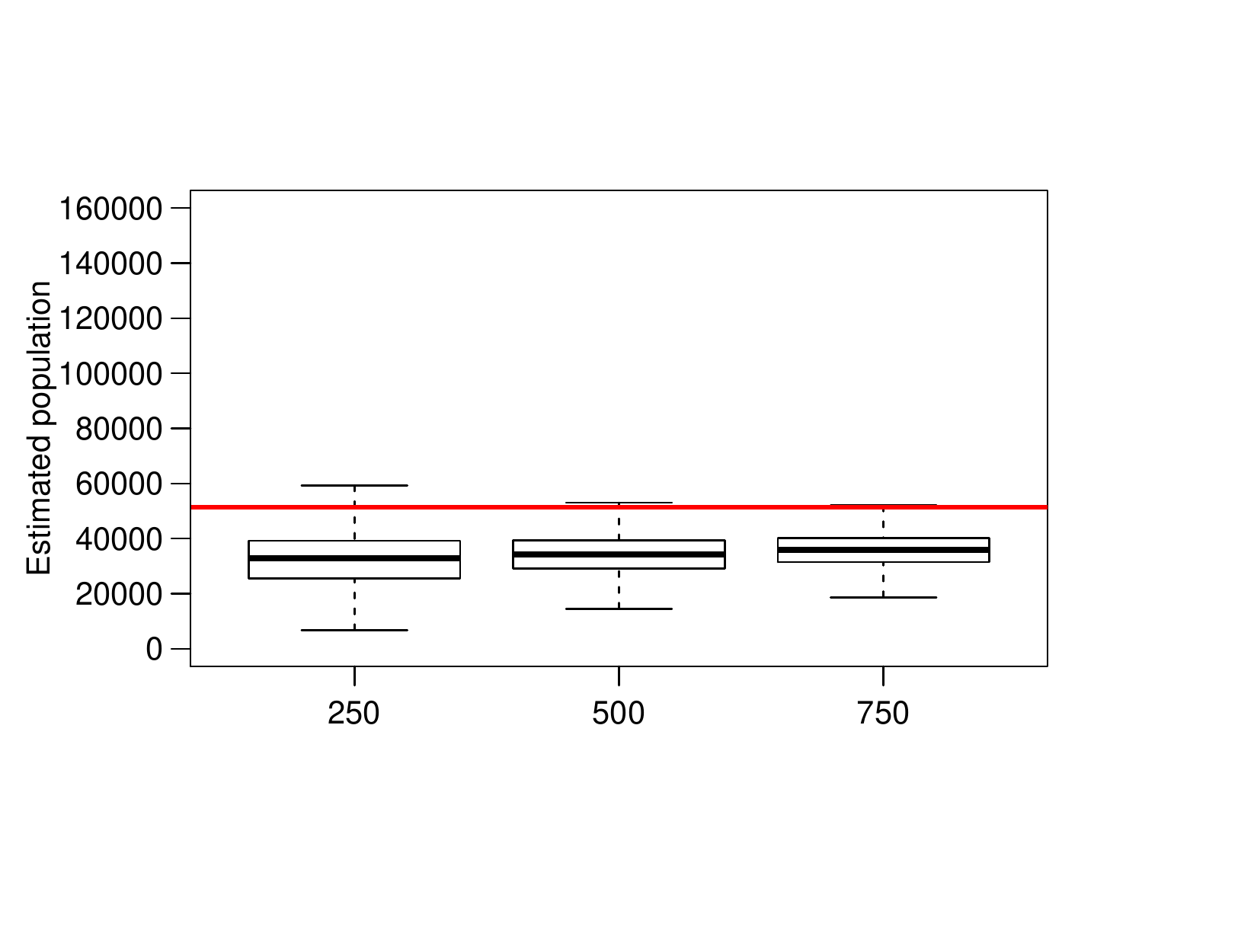}}\\
\subcaptionbox{ $n^{\psi}_3$ with $|\Omega| = 2 \cdot 10^3$\label{1c}}{\includegraphics[width = 0.3\linewidth]{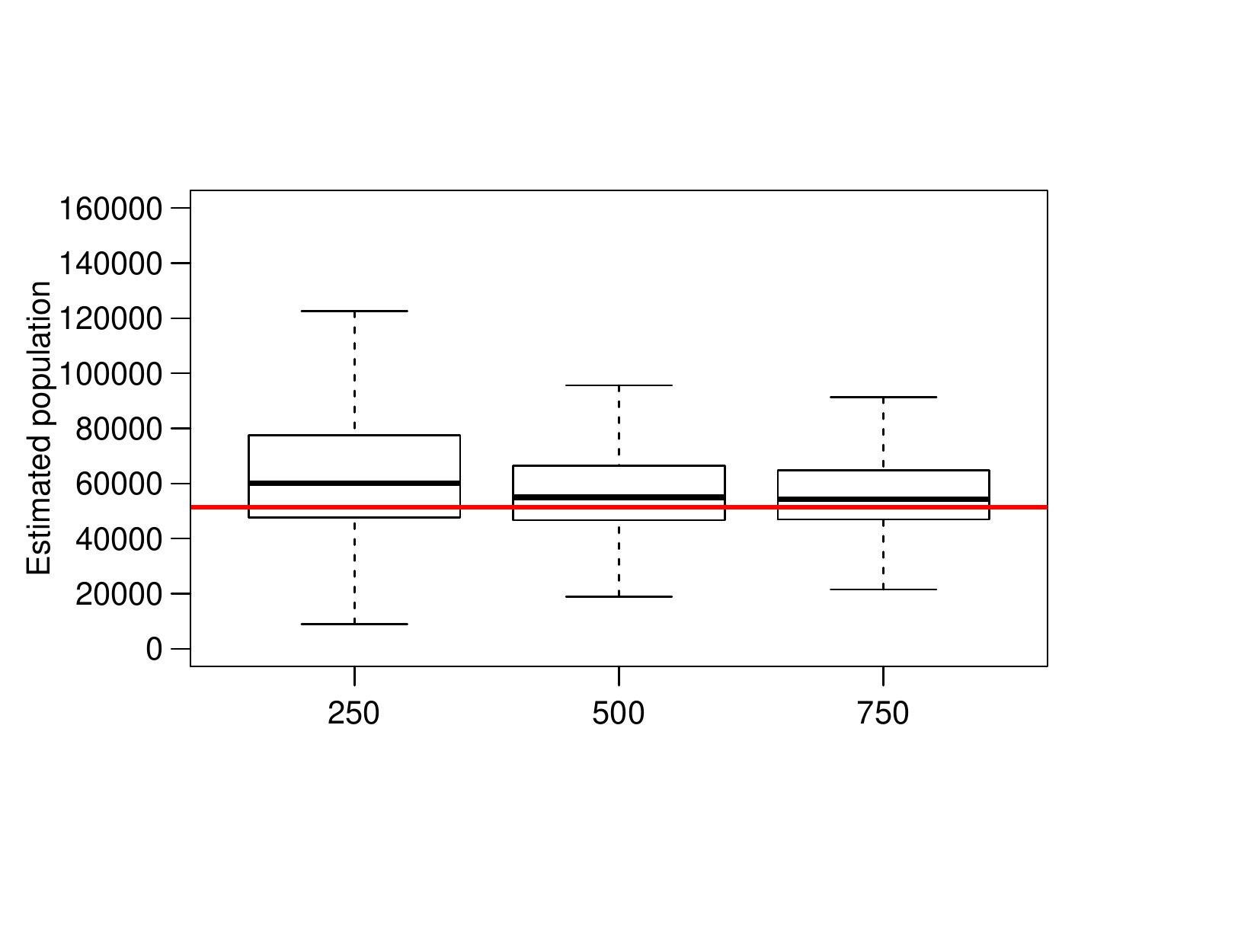}} &
\subcaptionbox{ $n^{\psi}_3$ with $|\Omega| = 32 \cdot 10^3$\label{2c}}{\includegraphics[width = 0.3\linewidth]{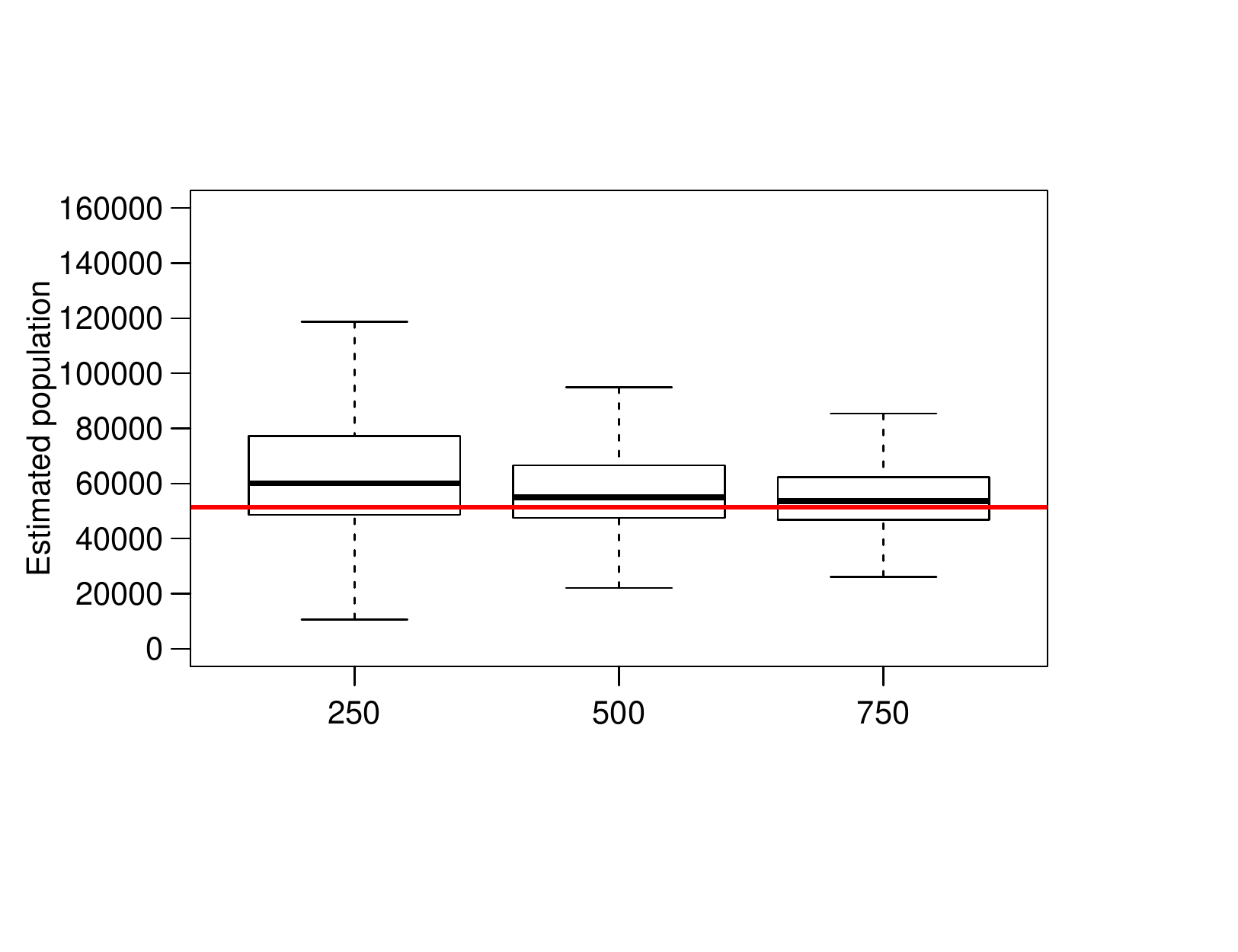}} &
\subcaptionbox{ $n^{\psi}_3$ with $|\Omega| = 256 \cdot 10^3$\label{3c}}{\includegraphics[width = 0.3\linewidth]{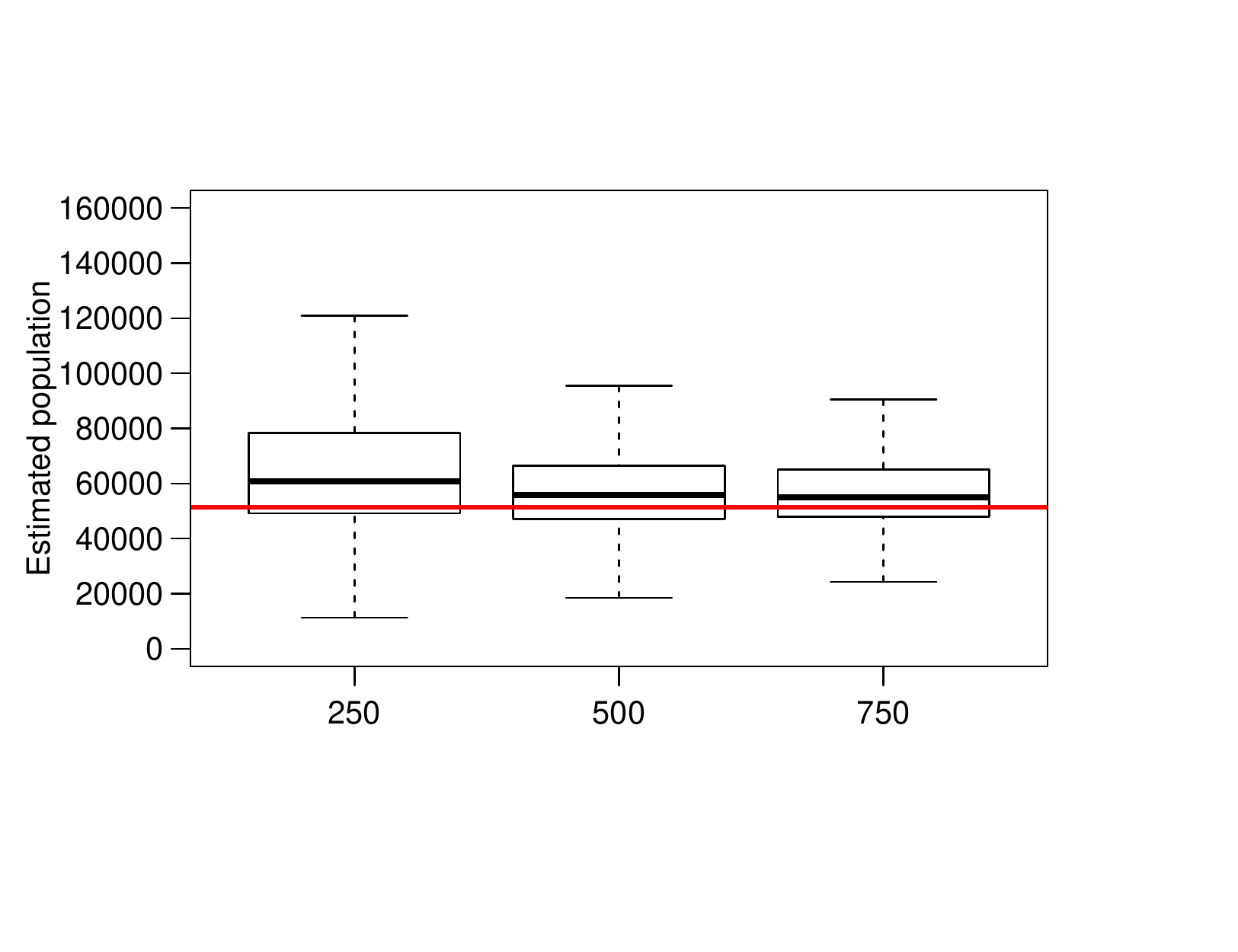}}
\end{tabular}
\caption{Estimator $n^{\psi}_2$ (above) and $n_3^\psi$ (below) on Brightkite network; $|\Omega| =2\cdot 10^3$ to $256 \cdot 10^3$. In each box, the thick line indicates the sample median; the top of the box is the median of the upper half of the estimated values (75\% quartile); the bottom of the box indicates the median of the lower half of the estimated values (25\% quartile; and the whiskers indicate the full range of estimated values. Data points that exceeded the third quartile boundary by over 1.5 times the interquartile range (IQR) were treated as outliers and removed.}
\label{rds-Brightkite-alg3}
\label{rds-Brightkite-alg4}
\end{figure}

\section{Discussion}
\label{sec:discussion}
The results shown here indicate that size estimates for hidden and hard-to-reach populations can be derived from RDS samples across a range of topologies, and in the presence of significant network clustering. As importantly, this can be accomplished under conditions of anonymity by way of identification hashing, e.g. using telefunken codes \cite{TELEFUNKEN2012} or a Privatized Network Sampling (PNS) design \cite{fellowsThesis}. The $n_3^\psi$ estimator joins other successful, RDS-based population estimation procedures, such as those by Handcock and Gile \cite{handcock2014estimating}, and Crawford, Wu, and Heimer \cite{crawford2017hidden}. Like Crawford et al, we make use of half-edge counts. However, our estimator invokes a different strategy---beginning with the original capture-recapture concept---and is shown to be robust across a wide range of topologies. 

A second notable feature of the $n_3$ and $n_3^\psi$ estimators is that a lower level of variance  can be expected at conventional RDS sample sizes. For $r=500$ to $750$, interquartile ranges were low relative to both the median estimate and true population size. Even when hashing was employed towards subject anonymity, sufficiently large hash spaces ($32\cdot 10^3$ or larger),  and samples sizes (500 or above) produced a narrow range of estimates. Given concerns about RDS sample variance generally \cite{verdery_network_2015}, these results indicate robustness against the faults of a single sample.

A third consistent feature observed in these experiments is the performance of the $n_2$ and $n_3$ estimators as graph density increases (see Figures \ref{results:n2} $\&$ \ref{results:n3}). Both estimators show worse performance in sparse (i.e. $dG = 3$) versus more dense networks ($dG = 10$), in terms of the interquartile range of estimates. This was true, with some small variation, across all 5 synthetic topologies. Given the edge-sampling focus of our approach, this is not surprising. Fewer total edges suggest fewer total ``matches'' to discover, leading to greater variability depending on stochastic factors likely associated with the selection of RDS seeds and the random walk features of the RDS sampling process. These results suggest limits on the implementation of the $n_2$ and $n_3$ estimators in sparse graphs.

As researchers increasingly turn to RDS methods for sampling hard-to-reach populations, these results should be of considerable interest to those concerned with what is often referred to as ``the denominator problem''. Where agencies and government administrations seek to understand both the scope of public health challenges, and to measure the effectiveness of their intervention and promotion efforts, the ability to estimate population size (and with this, population prevalence) is widely needed. The results presented here indicate that ``one step'' methods are capable of providing such estimates. Along with the methods mentioned above, this work has the potential to provide public health officials and planners with means to more effectively promote the health of hidden populations---and thus the health of the larger populations in which they are embedded. 
 
\subsection{Limitations}
\label{sec:limitations}

In using uniform random samples to estimate population size, it is possible for the proposed $n_1$ estimator to ``fail'' if one finds that $\langle M(T,\emptyset)\rangle = 0$ in Definition \ref{def:n1}. This happens with greater frequency as the sample size $r \ll n$ the population size.  Figure~\ref{failure-rate}~(a) shows the mean failure rate (the fraction of the 13,500 trials\footnote{We considered: each of $5$ families ${\cal L}(\lambda,n), {\cal P}(\lambda,n), {\cal X}(\lambda,n), {\cal B}(\lambda,n)$, and ${\cal E}(\lambda,n)$ defined in Section \ref{sec:graph-families}, and each $\lambda=3, 5, 10$; from each of these $15$ concrete sample spaces, we used configuration graph sampling to select $30$ random graphs of sizes.  In each of these $5 \times 3 \times 30 = $450 graphs, we generated 30 uniform samples (for $n_1$).  In this manner, a total of $450\times 30 = $13,500 simulations were conducted.} where $n_1$ failed to produce a population estimate), for each choice of population size $n$ (ranging from $5\cdot 10^3$ to $40\cdot 10^3$), and uniform sample size $r$ (chosen to be $250,500$ or $750$).  We see from Figure~\ref{failure-rate}~(a) that failure rate is non-linear in both $r$ and $n$.  For small uniform samples $r=250$, the failure rate of $n_1$ is $\sim 0$ when $n=10\cdot 10^3$, but undergoes an inflection at $n=20\cdot 10^3$, and rises to 3.9\% when the population size again doubles to $n=40\cdot 10^3$.  

Similarly, in using respondent-driven sampling to estimate population size, it is possible for the proposed $n_2$ (resp. $n_3$) estimators to ``fail'' if one finds that $\langle M(S, F)\rangle = 0$ in Definition \ref{def:n2} (resp. $\sum_{s\in D} \langle X(s, F, \gamma) \rangle = 0$ in Definition \ref{def:n3}). Figure~\ref{failure-rate}~(b) shows the mean failure rate (the fraction of the 13,500 trials where $n_2$ failed to produce a population estimate), for each choice of population size $n$ (ranging from $5\cdot 10^3$ to $40\cdot 10^3$), and RDS sample size $r$ (chosen to be $250,500$ or $750$).  RDS samples of size $r=250$ exhibit an $n_2$ failure rate of $\sim 0$ when $n=5\cdot 10^3$, but undergo an inflection at $n=10\cdot 10^3$; the mean failure rate rises to 6.0\% when the population size again doubles to $n=40\cdot 10^3$. In examining the $n_3$ estimator, Figure \ref{failure-rate} (c) shows us that when it is used with RDS samples of size $r=250$, it  exhibits a failure rate of $\sim 0$ when $n=5\cdot 10^3$, but the failure rate undergoes an inflection at $n=10\cdot 10^3$, rising to 8.8\% when the population size again doubles to $n=40\cdot 10^3$.  For sample sizes that are $2X$ and $3X$ as large (i.e. $r=500$ and $r=750$) the inflection point is not yet reached at $n=40\cdot 10^3$ and mean failure rates remain below 0.1\%. This indicates that our estimators based on RDS are more robust against failure than the $n_1$ uniform sampling estimator, and at typical RDS sample sizes ($500 \leqslant r\leqslant 750$), they are robust enough to be used in settings where the population size is expected to be on the order of $n\sim 40\cdot 10^3$.

Figure \ref{failure-rate} (d-e) explore the impact of hash space size on the mean failure rate.  Here we consider a fixed sample size $r = 500$ and vary the size of hash space $|\Omega|$ between $2 \cdot 10^{3}$ and $256\cdot 10^3$.  We observe that the mean failure rate  of $n_2^{\psi}$ and $n_3^{\psi}$ (again taken across 13,500 experiments) grow linearly as $n$ increases, but that the rate of growth depends on $|\Omega|$.  In particular, when $|\Omega|$ is too small (in this case $2 \cdot 10^{3}$ or smaller), the mean failure rate is seen to grow steeply, even for small networks.  The two  graphs (d-e) make evident the tradeoff between subject anonymity/privacy and the failure rates of the estimator.  When the hash space size is sufficiently large ($32 \cdot 10^{3} - 256 \cdot 10^{3}$), failure rates remain low, but smaller hash spaces (which provide for greater anonymity) may produce greater instability in the estimators.

Although $32 \cdot 10^{3} - 256 \cdot 10^{3}$ may appear to be a very large hash space size, we note 
$$10^{4} \leqslant 32 \cdot 10^{3} \leqslant 10^{5} \leqslant 256 \cdot 10^{3}\leqslant 10^{6}.$$
Thus, asking research subjects for the last 5 or 6 digits of their own telephone number and those digits of their friends' phone numbers would be sufficient to provide an accurate estimate (assuming that numerical digits are randomly assigned by phone service providers).  In the event that research subjects remain reluctant to reveal precise digits of their own or their alter's phone numbers, a telefunken code could be constructed \cite{TELEFUNKEN2012} or a Privatized Network Sampling (PNS) design \cite{fellowsThesis} employed.

\begin{figure}[t]
\setlength{\abovecaptionskip}{0cm}
\setlength{\belowcaptionskip}{0cm}
\captionsetup[subfigure]{skip=-18pt,position=top}
    \centering
    \begin{subfigure}[ht!]{0.31\textwidth}
        \caption{Mean failure rate of $n_1$ on uniform samples}
        \includegraphics[width=\textwidth]{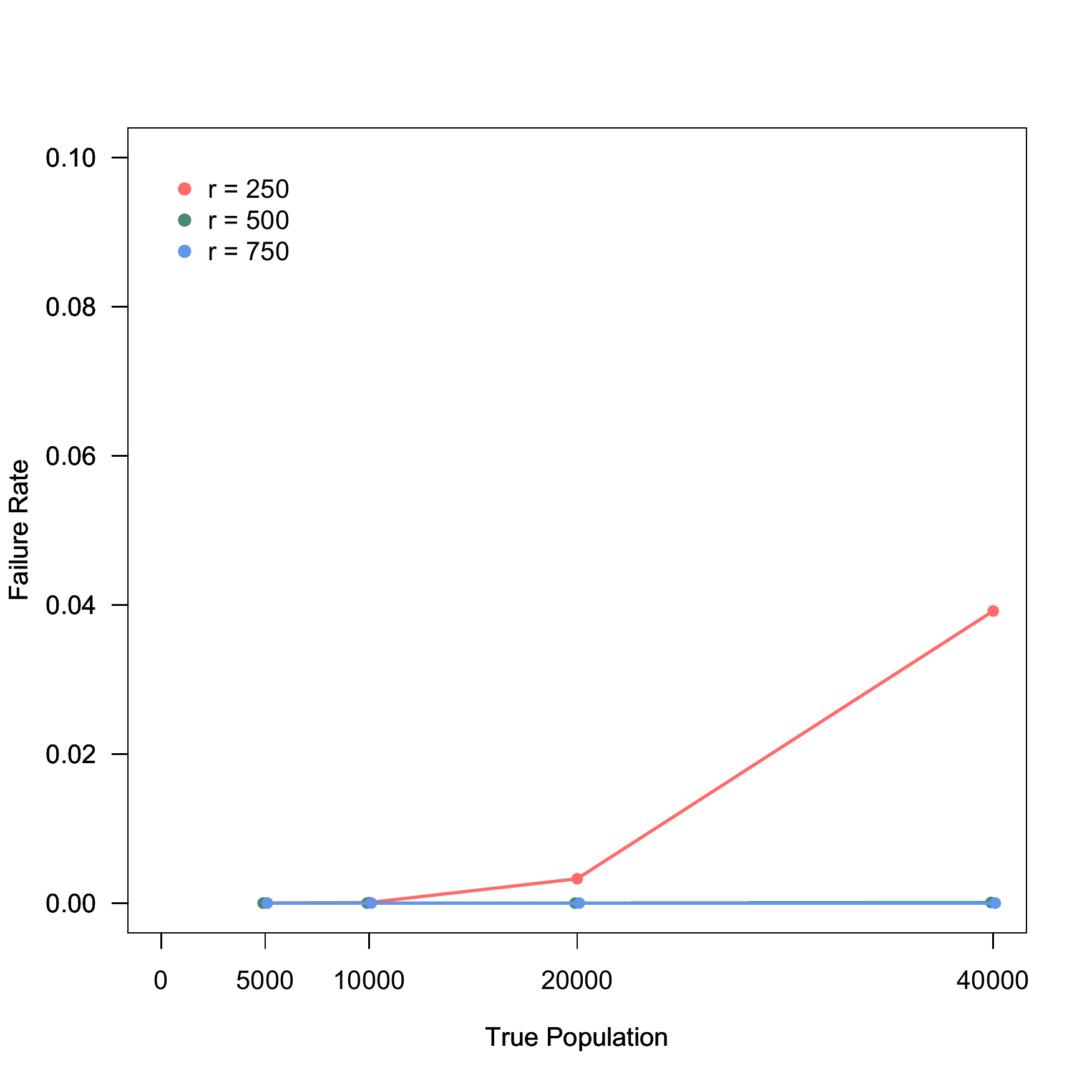}
        \label{}
    \end{subfigure}
    \begin{subfigure}[ht!]{0.31\textwidth}
    	\caption{Mean failure rate of $n_2$ on RDS samples}
        \includegraphics[width=\textwidth]{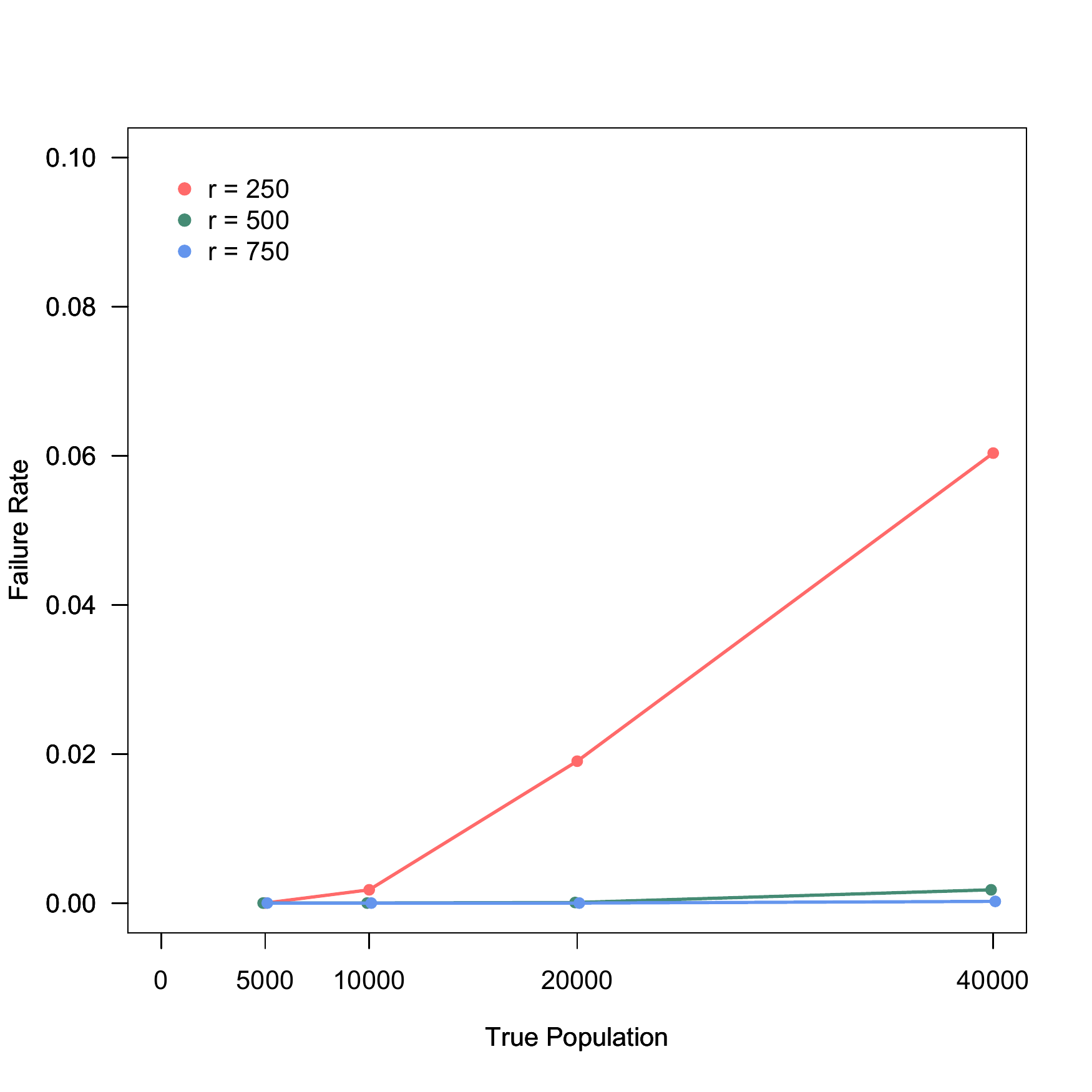}
        \label{}
    \end{subfigure}
        \begin{subfigure}[ht!]{0.31\textwidth}
        \caption{Mean failure rate of $n_3$ on RDS samples}
        \includegraphics[width=\textwidth]{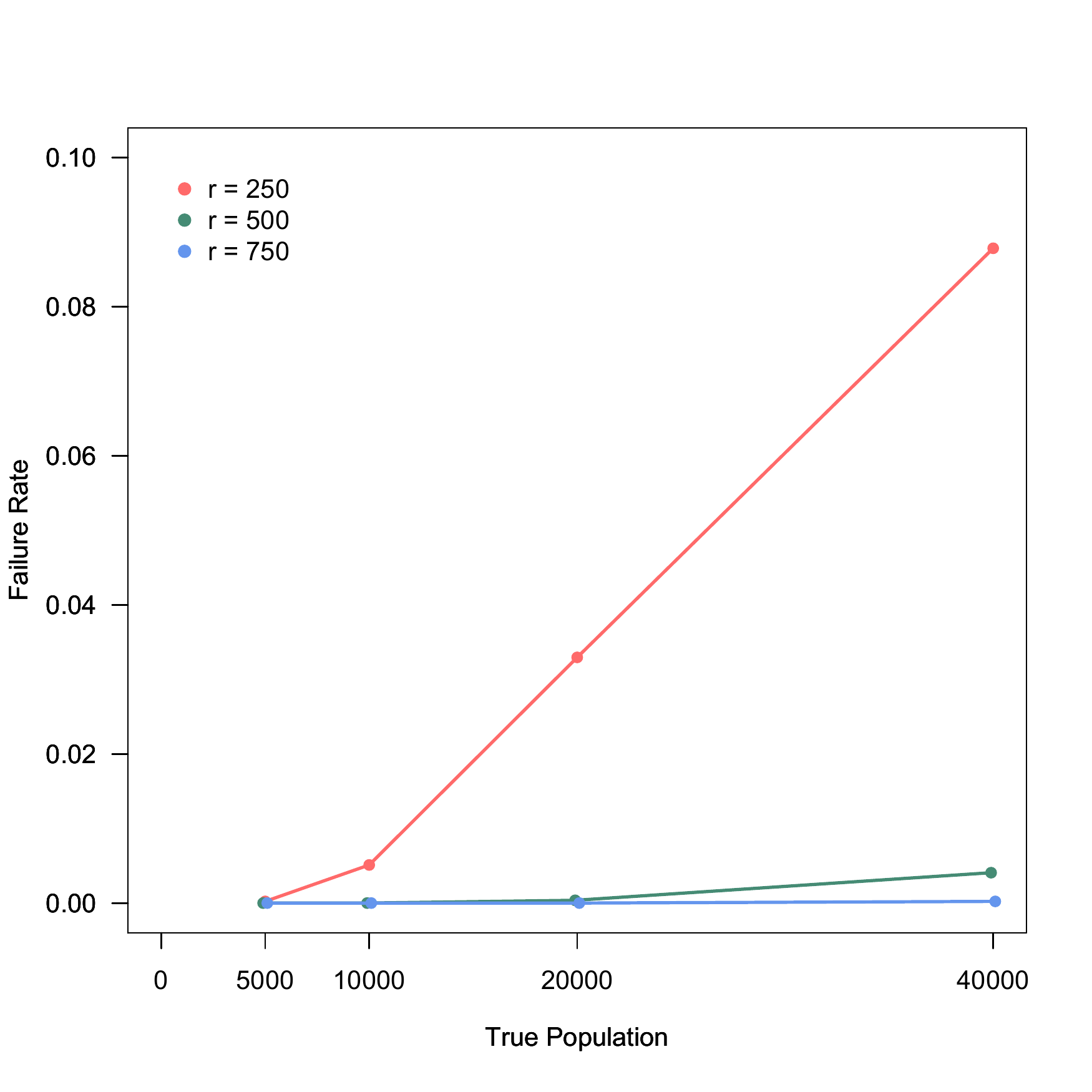}
        \label{}
    \end{subfigure}
    \begin{subfigure}[ht!]{0.48\textwidth}
    	\caption{$n_2^{\uppsi}$ failure rate with sample size $r = 500$}
        \includegraphics[width=\textwidth]{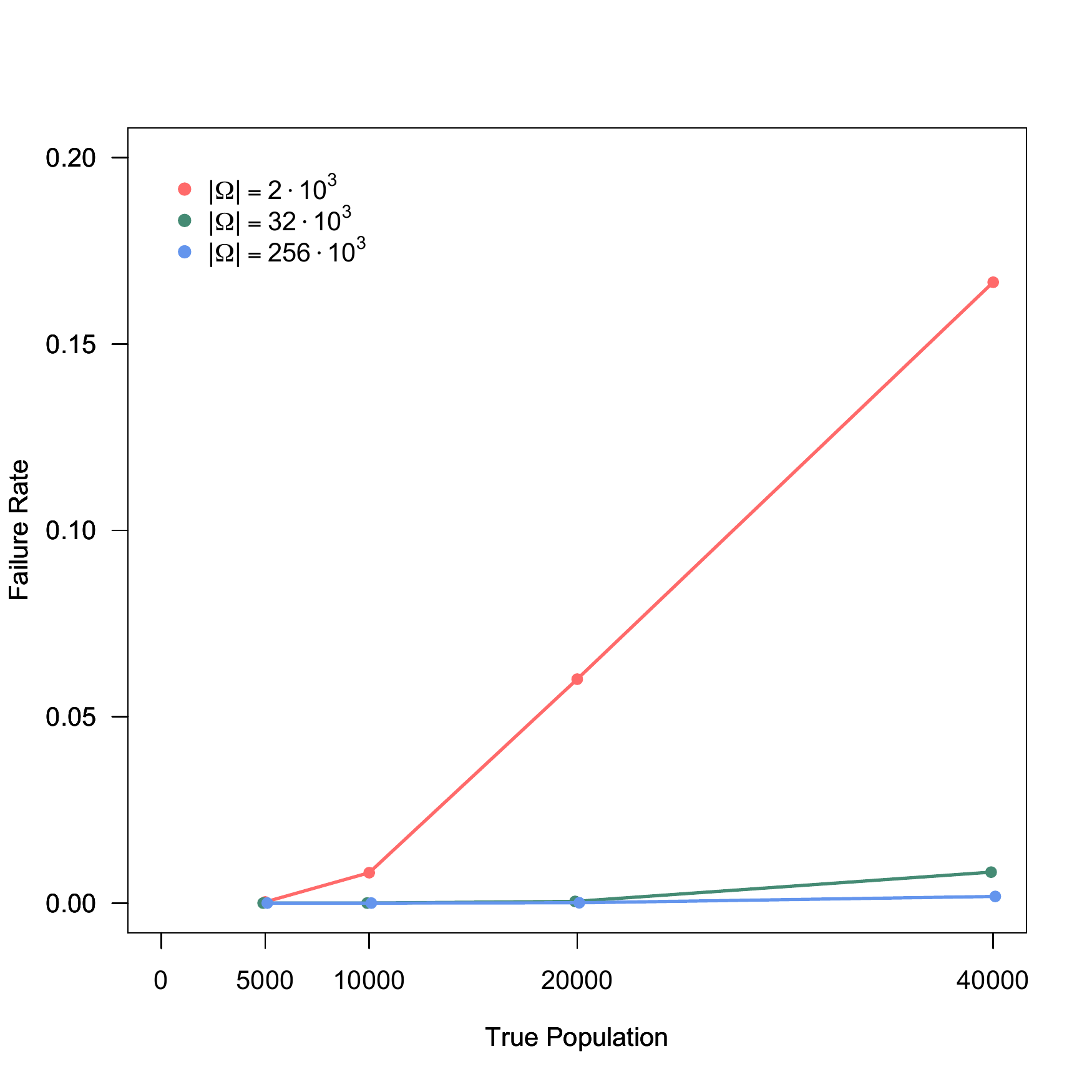}
        \label{}
    \end{subfigure}
    \begin{subfigure}[ht!]{0.48\textwidth}
    	\caption{$n_3^{\uppsi}$ failure rate with sample size $r = 500$}
        \includegraphics[width=\textwidth]{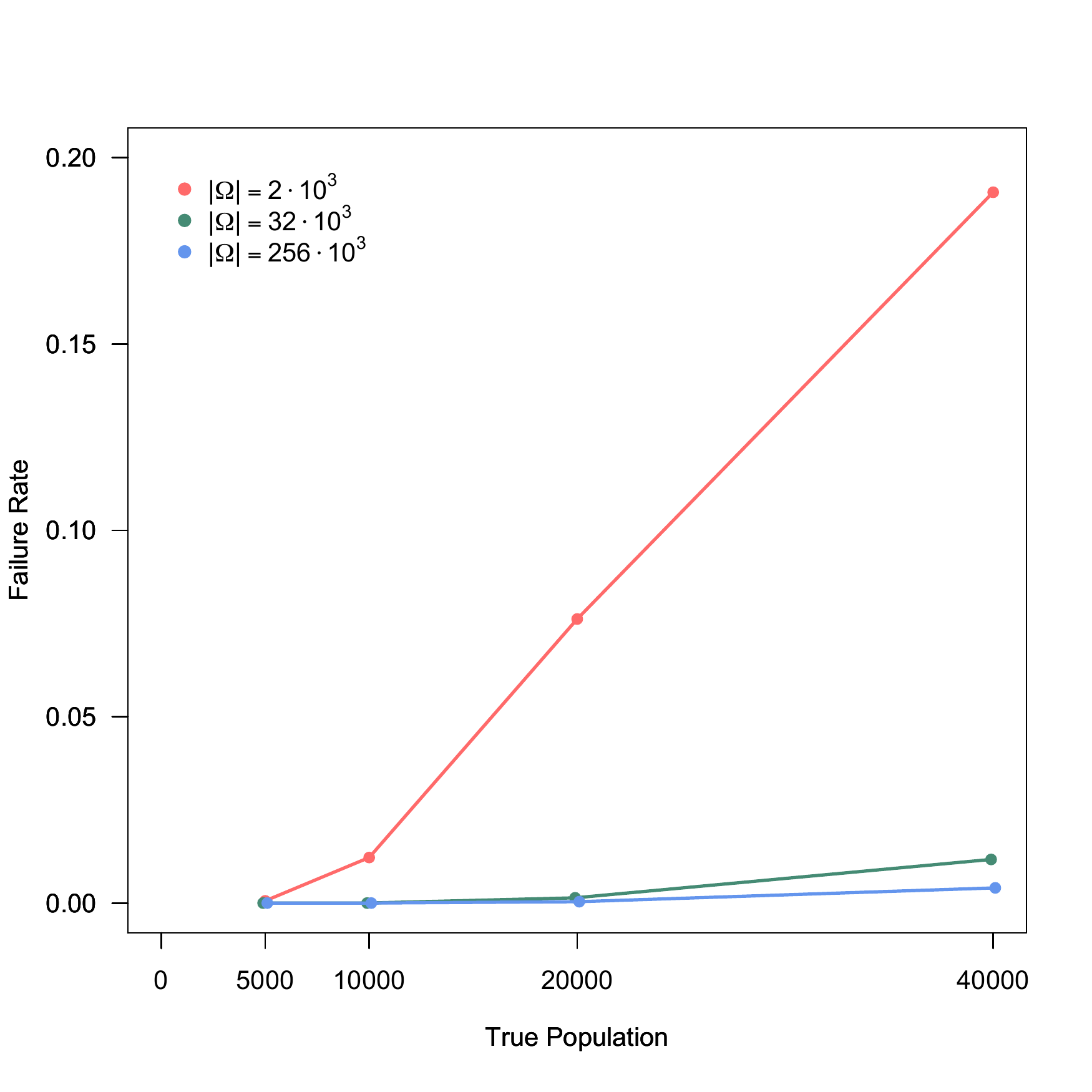}
        \label{}
    \end{subfigure}
    \caption{Failure rates of the $n_1$, $n_2$, and $n_3$ estimators}
    \label{failure-rate}
\end{figure}

\section{Acknowledgements}
\label{sec:acknowledgements}
Research reported in this publication was supported by the National Institutes for Health, National Institute on Drug Abuse under Award Number R01 DA037117 and National Institute for General Medicine R01 GM118427.

\newpage
\bibliographystyle{alpha}
\bibliography{paper1}

\end{document}